\documentclass[usletter,11pt]{article}

\pdfoutput=1

\usepackage{lipsum}
\usepackage{amsmath,amssymb, amstext, amsfonts,bm,dsfont,amsthm}
\usepackage{graphicx}
\usepackage{epstopdf}
\usepackage{algorithmic}
\usepackage{hyperref,cleveref}
\usepackage{thmtools,thm-restate}
\usepackage{enumitem}
\usepackage{subfiles}

\DeclareGraphicsExtensions{.eps,.pdf,.png,.jpg}

\usepackage[margin=1in]{geometry}

\newtheorem{siamremark}{Remark}
\newtheorem{observation}{Observation}

\newtheorem{proposition}{Proposition}

\usepackage[font=small,labelfont=bf]{caption}

\usepackage[backend=bibtex,style=ieee]{biblatex}
\newcommand{\TheTitle}{A Spatial Filtering Approach to Biological Patterning}

\title{{\TheTitle}\thanks{This work was based in part off of \cite{perkins_arcak_2018acc}.
		This work was funded by the United States Air Force Office of Scientific Research award FA9550-14-1-0089.}}

\author{
	Melinda Liu Perkins \thanks{Department of Electrical Engineering, University of California, Berkeley, CA, USA
		(mindylp@eecs.berkeley.edu)}
	\and
	Murat Arcak \thanks{Department of Electrical Engineering, University of California, Berkeley, CA, USA (arcak@eecs.berkeley.edu).}
}

\usepackage{amsopn}

\graphicspath{{.}}

\date{}

\begin{document}
	
	\maketitle
	
	\begin{abstract}
		Interactions between neighboring cells are essential for generating or refining patterns in a number of biological systems.  We propose a discrete filtering approach to predict how networks of cells modulate spatially varying input signals to produce more complicated or precise output signals.  The interconnections between cells determine the set of spatial modes that are amplified or suppressed based on the coupling and internal dynamics of each cell, analogously to the way a traditional digital filter modifies the frequency components of a discrete signal.  We apply the framework to two systems in developmental biology: the Notch-Delta interaction that shapes \textit{Drosophila} wing veins and the Sox9/Bmp/Wnt network responsible for digit formation in vertebrate limbs.  The latter case study demonstrates that Turing-like patterns may occur even in the absence of instabilities.  Results also indicate that developmental biological systems may be inherently robust to both correlated and uncorrelated noise sources.  Our work shows that a spatial frequency-based interpretation simplifies the process of predicting patterning in living organisms when both environmental influences and intercellular interactions are present.
	\end{abstract}
	
	\section{INTRODUCTION}
	\begin{figure}[!hbpt]
		\centering
		\includegraphics[width=\columnwidth]{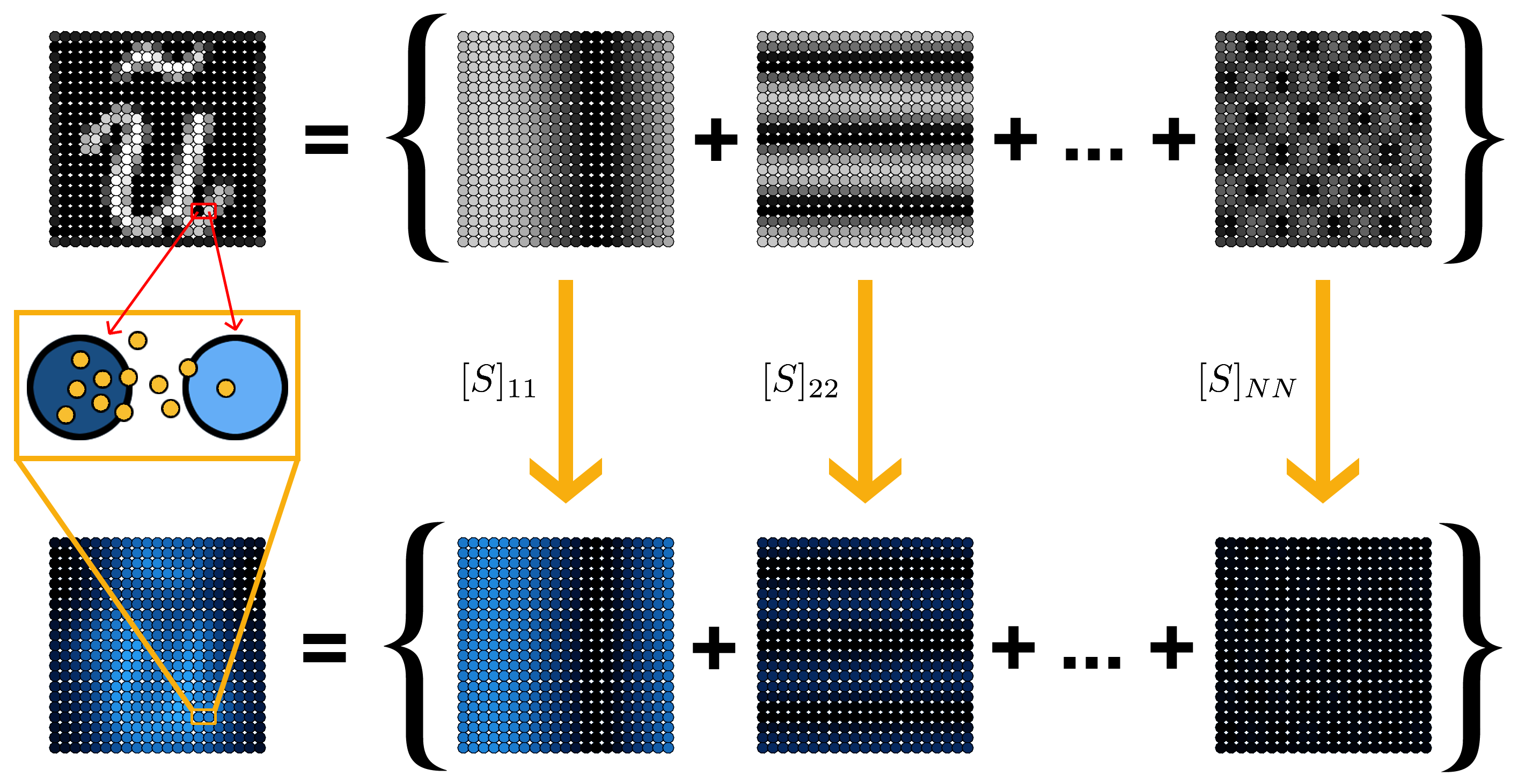}
		\caption{\textbf{Interacting cells filter input to readout by changing the relative weights of spatial modes.}  A network of interacting cells responds to a spatially varying input, or prepattern (gray), which can be represented as a weighted sum of spatial modes.  Each cell produces a readout (blue) in response to the input it receives at a particular point in space as well as to the outputs from other cells (gold dots) with which it is ``interconnected'' (e.g., by diffusible molecules).  The readout across all cells forms a spatially varying pattern that can be approximated as a sum of the same spatial modes as the input but with different weights.  Specifically, the weight of a given component in the readout is the product of the weight of that component in the input with a ``filter coefficient'' determined by both the internal dynamics of, and the interactions between, the individual cells.  These filter coefficients are unique to each patterning component independent of its weight in the input.  The process of modifying weights from input to readout in the manner described above is referred to as ``filtering'' (gold arrows).  In this example, both the input and the signals from neighboring cells promote gene expression such that spatial modes with small wavelengths are attenuated relative to components with large wavelengths.} \label{fig:filt_schematic}
	\end{figure}
	Biological organisms rely on spatial variation in cell activity to coordinate diverse phenomena including contrast enhancement in the visual system \cite{hartline_ratliff_1957} and body planning in developing embryos \cite{jaeger_2011}.  Interactions between neighboring cells play a crucial role in generating spatial patterns spontaneously from stochastic initial conditions or by refining simple inputs, such as chemical concentration gradients, into complex outputs, such as stripes in gene expression \cite{green_sharpe_2015}, \cite{greenwald_rubin_1992}.  Mathematical theory in developmental biology has emphasized spontaneous pattern formation through the reaction-diffusion (Turing) mechanism \cite{turing_1952} as well as contact- or diffusion-mediated lateral inhibition \cite{collier_et_al_1996}, \cite{jaeger_2011}, \cite{sprinzak_et_al_2010}.  In practice, however, the conditions necessary for spontaneous pattern formation may be prohibitively difficult to satisfy.
	
	Prepattern processing---also known as Wolpert's theory of positional information \cite{wolpert_1969}---is an attractive and flexible alternative to spontaneous patterning, but mathematical analysis of prepattern processing has been largely limited to numerical simulations (e.g, \cite{sprinzak_et_al_2010}).  Prepatterns may arise directly from environmental influences that differ by cell or from consistent, preinduced parameter variation across space.
	
	We propose a discrete filtering approach to analyze how networks of interacting cells respond to prepatterns.  The framework elucidates which components of spatial structure are amplified and which ones attenuated by the system to produce an output from any given input.  The insights gained from this perspective challenge the conventional notion that instability is necessary for complex patterning; for example, our approach reveals that Turing-like stripes can emerge from a stable system lacking diffusion-driven instability, and furthermore that external noise reinforces rather than combats this behavior (Section \ref{sec:raspopovic}).
	
	In Section \ref{sec:main} we present the setup for the framework.  We combine the internal dynamics of cell behavior with interaction between cells by modeling each cell as an input-output module coupled to other modules.  We examine the steady-state gains for constant-in-time, spatially varying inputs (prepatterns) and show that the system behaves as a discrete spatial filter, where the interconnectivity between cells determines the spatial modes, while the coupling and input-output dynamics dictate how each mode is scaled to generate a readout pattern.  We also examine the system response to temporal and spatially varying noise inputs, measured with the $\mathcal{H}_2$ norm, to determine which spatial modes are sensitive to stochastic influence.  Lastly, we show how to apply the approach by considering a simple model of gene expression that exemplifies two of the most common classes of filter behaviors---highpass or lowpass---depending on the choice of parameters.
	
	In the remaining three sections we demonstrate the utility of the filtering perspective by examining two biological case studies: the Notch-Delta system in developing fruit fly wings (Section \ref{sec:notch-delta}) and the Sox9/Wnt/Bmp network in vertebrate digit formation (Sections \ref{sec:raspopovic} and \ref{sec:onimaru}).  We conclude with a brief summary and areas for future research.

	\section{THE SPATIAL FILTERING APPROACH} \label{sec:main}
	The main contribution of this paper is a filtering perspective for analyzing prepattern processing in developmental biological systems.  A central component of our approach is spatial mode decomposition, a common tool in distributed systems analysis (e.g., \cite{bamieh_et_al_2002}) that has previously been applied to detect instabilities in cellular networks lacking external inputs \cite{othmer_scriven_1971}.  In this section we introduce generalized notation followed by a derivation of the filter coefficients and the noise amplification factors that we will use throughout the remainder of the paper.
	
	\subsection{Notational Conventions}
	We use the following notational conventions (see also Supplementary Figure \ref{fig:supp_sprinzak_labels}):
	\begin{itemize}
		\item Cells are indexed by $i$ in vector form and spatial modes are indexed by $k$ in vector form or $(m,n)$ in an array, unless noted otherwise.
		\item Inputs except white noise in the context of the $\mathcal{H}_2$ norm are assumed constant in time.
		\item Vectors containing strictly constant-in-space entries are designated with an underline.  The entries corresponding to any fixed point in space are additionally labeled with an overbar, e.g., $\underline{u} = \underline{\bar{u}}\mathds{1}_N$ where $\underline{\bar{u}} \in \mathbb{R}$ and $\mathds{1}_N$ is the length $N$ vector of all ones.
		\item Steady-state values for time-dependent variables are designated with superscript asterisks.  Constant-in-space steady-state (i.e., homogeneous) solutions to nonlinear systems are designated with both an asterisk and an underline, e.g., $\underline{y}^* = \underline{\bar{y}}^*\mathds{1}_N$.
		\item ``Actual'' values in the standard basis are unadorned.  Perturbations from constant-in-space values are designated with a tilde; time-dependent perturbations are understood to be linear approximations of ``actual'' nonlinear solutions, e.g., $\tilde{x}_i(t) \approx x_i(t) - \underline{\bar{x}}^*$.  Perturbed variables in the basis $T$ are designated with a hat, e.g., $\hat{x}^* = T^{-1}\tilde{x}^*$.
	\end{itemize}

	\begin{figure}
		\centering
		\includegraphics[width=\columnwidth]{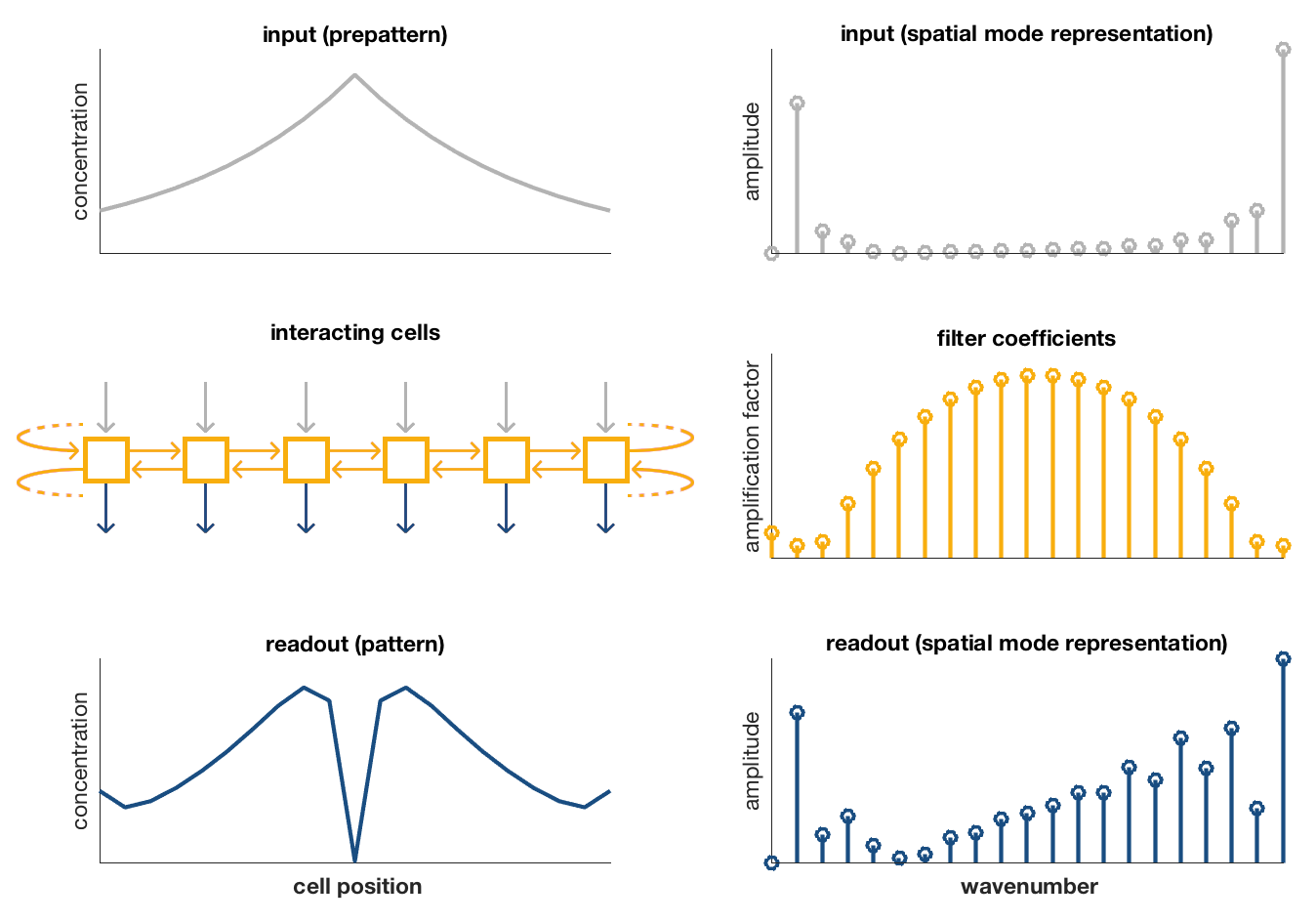}
		\caption{\textbf{Weights of spatial modes in a prepattern are multiplied by filter coefficients, determined by the internal dynamics and interconnectivity among cells, to produce patterns.}  Each cell acts as an input/output module.  The collective activity of cells produces the filtering behavior.  Here, a single line of cells with periodic boundary conditions communicates through contact-based lateral inhibition, resulting in a highpass filter (see Section \ref{sec:notch-delta}; note that wavenumber increases toward the center of the axis).  The spatial modes differ from the standard Fourier basis (Supplementary Observation \ref{obs:DFT_modes_real}), hence the asymmetry in the spatial mode representation of the input and readout (Supplementary Figure \ref{fig:supp_sprinzak_labels}).} \label{fig:math_filt_schematic}
	\end{figure}
	
	\subsection{System Dynamics and Filter Coefficients}
	We consider a generalized system of $N$ identical cells with the state variables of the $i$th cell at time $t$ given by $x_i(t) \in \mathbb{R}^{n}$, readout $y_i(t) \in \mathbb{R}$, and constant-in-time input $u_i \in \mathbb{R}$, which may represent an environmental stimulus or intrinsic parameter variation.  Coupling occurs via $v_i(t) \in \mathbb{R}^{q}$ and output $w_i(t) \in \mathbb{R}^{q}$ where $q \leq n$.  Let the vectors for the full system be the vertical concatenation $x(t)^T := [x_0(t)^T x_1(t)^T ... x_{N-1}(t)^T]$ and similarly for $u$, $y(t)$, $w(t)$, and $v(t)$.  The dynamics of the $i$th cell and the full linear coupling between the $N$ cells are given by
	\begin{equation}
	\begin{cases}
	\dot{x_i}(t) = f(x_i(t),v_i(t),u_i) \\
	w_i(t) = g(x_i(t)) \\
	y_i(t) = h(x_i(t)) \\
	v(t) = \left(M \otimes I_q\right)w(t) \
	\end{cases} \label{eq:sys_percell_eqs} \
	\end{equation}
	where $\otimes$ is the Kronecker product, $I_q$ is the $q \times q$ identity matrix, and $M \in \mathbb{R}^{N \times N}$.
	
	The system \eqref{eq:sys_percell_eqs} accommodates a wide range of specific deterministic models.  Intercellular processes such as gene expression and protein decay are encapsulated by appropriate definition of the evolution function $f$ for chemical concentrations $x_i$, including the effect of environmental stimuli or parameter values $u_i$ as well as signals from neighbors $v_i$.  The output $w_i$ is the subset of elements in $x_i$ that transmit signals to neighbors, with the method of transmission (e.g., diffusion, cell-to-cell contact) and the neighboring cells specified by the interconnection matrix $M$.  The readout $y_i$ isolates a quantity of interest to the user, which may be experimentally measurable (e.g., fluorescence) or simply relevant to a particular model (see examples in Sections \ref{sec:notch-delta}, \ref{sec:raspopovic}, and \ref{sec:onimaru}).
	
	The vectors indexed by $i$ describe patterns by the concentration of chemicals within individual cells at individual points in space.  A full pattern is reconstructed from $N$ elements, each of which represents the concentration of a chemical in a single cell.  However, patterns can also be thought of as combinations of spatially varying components that span multiple cells, e.g., stripes of varying thickness (frequency).  When weighted and summed, these spatial modes can represent arbitrary patterns of interest.  We use the term ``filtering'' to refer to the process by which a network of interacting cells alters the weighting of the spatial modes of the input, thereby producing a readout that is built from the same components as, but differs in appearance from, the input.  A key approximation to facilitate the analysis is that coupling between spatial modes is negligible, such that the readout can be expressed as a linear sum of the same set of spatial modes used to represent the input.  In analogy to traditional signal processing, the network of cells plays the role of a linear time-invariant system (filter) that modifies the frequency components of a (spatially) varying signal.  The following proposition formalizes this concept mathematically.
	
	\begin{proposition} \label{prop:filter_coeffs}
		If the system described by \eqref{eq:sys_percell_eqs} satisfies
		\indent\begin{enumerate}
			\item $M\mathds{1}_N = \mu\mathds{1}_N$ and $M$ is diagonalized by $T$ ($M = T\Lambda T^{-1}$),
			\item given $\underline{\bar{u}} \in \mathbb{R}$, $\exists~\underline{\bar{x}}^{*} \in \mathbb{R}^n$ such that $f\left(\underline{\bar{x}}^{*},\mu g\left(\underline{\bar{x}}^{*}\right),\underline{\bar{u}}\right) = 0$ and $\underline{x}^* := \mathds{1}_N \otimes \underline{\bar{x}}^{*}$, $\underline{u} := \underline{\bar{u}}\mathds{1}_N$,
			\item the homogeneous steady state $(\underline{x}^*,\underline{u})$ is stable,
		\end{enumerate}
		then the system may be linearized about $(\underline{x}^*,\underline{u})$ with linearization matrices
		$$A := \frac{\partial f}{\partial x_i} \bigg|_{(\underline{\bar{x}}^{*},\underline{\bar{v}}^*,\underline{\bar{u}})}, ~
		B_v := \frac{\partial f}{\partial v_i} \bigg|_{(\underline{\bar{x}}^{*},\underline{\bar{v}}^*,\underline{\bar{u}})}, ~
		B_u := \frac{\partial f}{\partial u_i} \bigg|_{(\underline{\bar{x}}^{*},\underline{\bar{v}}^*,\underline{\bar{u}})}, ~
		C := \frac{d h}{d y_i}\bigg|_{\underline{\bar{x}}^{*}},
		G := \frac{d g}{d x_i}\bigg|_{\underline{\bar{x}}^{*}}.$$
		A constant-in-time, varying-in-space input $u$ to \eqref{eq:sys_percell_eqs} equates to a perturbing input $$\hat{u} = T^{-1}(u - \underline{u})$$ to the linearized system in the coordinate system $T$.  Then the steady-state perturbed readout $\tilde{y}^*$ in the basis $T$ is $\hat{y}^* := T^{-1}\tilde{y}^* = S\hat{u}$ where $S$ is a diagonal matrix and $${\left[S\right]_{kk} = -C(A+\lambda_k(M)B_vG)^{-1}B_u}$$ is the steady-state gain of the $k$th eigenvector of $T$.
	\end{proposition}

	The conditions (1) through (3) ensure that the network, when given a constant-in-space input, will admit a stable, homogeneous steady-state solution, and that the expression pattern across cells can be represented in a complete orthonormal basis other than the standard; this basis $T$ comprises the modes.  The $\left[S\right]_{kk}$ collectively form the ``filter coefficients", which dictate how the corresponding $N$ spatial modes are multiplicatively scaled by the system when the input is no longer constant in space (Figure \ref{fig:math_filt_schematic}).  In other words, the matrix $S$ ``filters'' the perturbed input into a perturbed readout with respect to the eigenvectors, or spatial modes, of $M$ as contained in $T$.  In contrast to the conditions for spontaneous pattern formation, our approach does not require the perturbed system to be unstable; large amplification of spatial modes is possible even when the system is stable.  Figure \ref{fig:img_filter_examples} shows eight examples of prototypical filter behaviors that vary with interaction type and cellular interconnectivity.

	Many continuous pattern-forming and distributed dynamical systems exhibit spatial invariance of the dynamics with respect to linear transformations such as reflections, rotations, or translations \cite{cross_greenside_2009}, \cite{bamieh_et_al_2002}.  The discrete-space cellular network has a direct analog: If\ $M$ is invariant under a linear transformation, then $S$ is also invariant under the same transformation, since the system dynamics are identical within each cell and therefore the only spatial information contained within the system is contained in $M$.  Formally, if we let $\Pi \in \mathbb{R}^{N\times N}$ be a linear transformation and $M$ and $\Pi$ commute, then $\Pi$ and the filter coefficient matrix $S$ also commute (i.e., the map from input to output is equivariant).  Thus, $M$'s permutation $M_{\Pi} := \Pi M\Pi^{-1}$ shares the same eigenvectors $T$ and corresponding eigenvalues $\Lambda$ as $M$, which implies that the filter coefficients for a system with interconnection matrix $M$ are the same as for that system with interconnection matrix $M_{\Pi}$.

	\subsection{Stochastic Influence on Patterning}
	The role of stochastic influences in biological patterning is a subject of ongoing theoretical and experimental interest (e.g., \cite{butler_goldenfeld_2011}, \cite{karig_et_al_2018}).  Here, we concern ourselves with the response of spatial modes to time-varying white noise inputs, for which the $\mathcal{H}_2$ norm of the system quantifies the expected power of the perturbed readout.  The $\mathcal{H}_2$ norm has previously been used to analyze energy amplification in channel flows \cite{jovanovic_bamieh_2005}, networks of cells \cite{ferreira_arcak_2013}, and reaction-diffusion systems \cite{hori_hara_2012}, among others.
	
	To begin our analysis we rewrite the linearized ordinary differential equations in the form of a nonlinear Langevin equation (It\^{o} interpretation).  Since the $N$ modes are decoupled we can write the equation for the perturbed states in the $k$th mode as
	\begin{equation}
	d\hat{x}_k(t) = \left[ \left( A \otimes I_N \right) + \left( B_v G \right) \otimes \Lambda \right]\hat{x}_k(t) dt + \left( B_u \otimes I_N \right)d\hat{u}_k(t). \label{eq:xhatk_sde} \
	\end{equation}
	Here $\hat{u}(t)$ is an $n_u$-dimensional independent standard Wiener process, also known as the standard Brownian motion process.  Implicitly we assume that concentrations of reactants are high enough to permit us to neglect molecular-level fluctuations, which cannot accurately be described by the Langevin approach \cite{van_kampen_2007}.
	
	With slight abuse of notation, $d\hat{u}_k(t)$ is stationary, therefore the variance of the readout $y(t) = Cx(t)$ in mode $k$ does not change in time.  The variance is given by
	\begin{equation}
	\mathbb{E}\left[\left|\hat{y}_k\right|^2\right] = \mathbb{E}\left[Tr\left(\hat{y}_k\hat{y}_k^T\right)\right] = Tr\left(C\mathbb{E}\left[\hat{x}_k\hat{x}_k^T\right]C^T\right) = Tr\left(CQ_kC^T\right), \nonumber \
	\end{equation}
	where $Q_k := \mathbb{E}\left[\hat{x}_k\hat{x}_k^T\right]$ is the covariance of the reactants in the $k$th mode.
	
	Let $G_k(t)$ be the impulse response of \eqref{eq:xhatk_sde} for readout $y(t)$.  We could equivalently write
	\begin{align}
	\mathbb{E}\left[Tr\left(\hat{y}_k\hat{y}_k^T\right)\right] &= \int_0^\infty \nonumber \ \mathbb{E}\left[Tr\left(G_k(t)d\hat{u}_kd\hat{u}_k^TG_k(t)^T\right)\right]dt \\
	&= \int_0^\infty Tr\left(G_k(t)\mathbb{E}\left[d\hat{u}_kd\hat{u}_k^T\right]G_k(t)^T\right)dt
	= \int_0^\infty Tr\left(G_k(t)G_k(t)^T\right) dt \nonumber \\
	&=: ||G_k(t)||^2_{\mathcal{H}_2}, \nonumber \
	\end{align}
	from which we deduce that the $\mathcal{H}_2$ norm is equivalent to the variance of $\hat{y}_k$ and can be calculated as $Tr\left(CQ_kC^T\right)$ where $Q_k$ is the positive semi-definite solution to the Lyapunov equation
	\begin{equation}
	\left( A + \lambda_k B_vG\right)Q_k + Q_k\left( A + \lambda_k B_vG\right)^T + B_uB_u^T = 0. \nonumber \
	\end{equation}
	
	The unit variance of $d\hat{u}_k(t)$ allows us to interpret $||G_k(t)||^2_{\mathcal{H}_2}$ as the ratio of the variance of the readout to the variance of the input in mode $k$.  Moreover, since $d\hat{u}_k(t)$ is zero mean, the squared $\mathcal{H}_2$ norm is also equivalent to the time integral of the expected power spectral density, or the factor by which the system amplifies the average power of the readout within mode $k$.  Those modes with the highest $\mathcal{H}_2$ norms are most strongly amplified by the external noise source.
	
	\subsection{Constructing the Interconnection Matrix} \label{sec:constructing_interconnection_matrix}
	In the remainder of this paper we construct the interconnection matrix $M$ for a particular signal as follows:
	
	\begin{enumerate}
		\item The length $N$ vector of all ones $\mathds{1}_N$ is an eigenvector of $M$, which implies that a homogeneous steady-state solution exists.
		\item The $i$th, $j$th entry $\left[M\right]_{ij}$ for $i \neq j$ is 0 if cell $i$ is not connected to cell $j$.  Otherwise $0 < \left[M\right]_{ij}$, where the magnitude $\left[M\right]_{ij}$ captures the ``strength'' of the connection.
		\item The diagonal entries $\left[M\right]_{ii}$ encapsulate the ``signaling cost'' associated with interaction.  Negative values imply the cell loses signal to transmit to its neighbors, e.g., diffusion.
	\end{enumerate}

	In many biological systems, cells can be approximated to have the same distance between them and the same communication strength with each of their neighbors.  In such systems, the corresponding spatial modes are sinusoidal, giving rise to stripes or spots.  Lower-frequency modes correspond to longer-wavelength spatial modes, while higher-frequency modes correspond to shorter-wavelength spatial modes.  The relationship between patterning wavelength and spatial mode frequency enables these systems to be interpreted from the standpoint of how the weights of the frequency components in an input are scaled to produce the readout, analogous to filtering as it is understood in discrete signal processing.  In this paper we will consider basis vectors arising from a line or sheet of regularly spaced cells with periodic or no-flux boundary conditions.  The modes then pertain to two common signal processing transforms: the discrete Fourier transform (DFT) for periodic boundaries or the second discrete cosine transform (DCT-2) for no-flux boundaries.  The eigenvectors and eigenvalues for these transforms are well known (e.g., \cite{strang_1999}); a review is offered in Supplementary Section \ref{sec:spatial_mode_tutorial}.  We will assume modes are indexed in order of increasing frequency with increasing $k$ toward $\frac{N}{2}$ for the DFT and $N$ for the DCT-2.

	\subsection{Minimal Model: Gene Expression with Autoregulation}
	The following example is a simple model that is easy to solve analytically for the filter coefficients.  We begin with a brief description of gene expression for readers who may not be familiar with the biology, including terminology that will be used in later sections.  We then apply the filtering approach to the example, including an expansion of the matrix notation to emphasize the role of the filter coefficients as ``weights'' for the spatial modes.  As this model focuses on biological and filtering concepts, intercellular interaction is described only in the most general terms, leaving exploration of the underlying mechanisms to later examples.
		
	The case studies in this paper deal with gene expression, or the process by which a gene coded in DNA is \textit{transcribed} into mRNA molecules that are then \textit{translated} into protein molecules (Supplementary Figure \ref{fig:supp_gene_exp_schematic}).  The production and degradation rates for mRNA and protein may be modulated by physical or chemical factors; for example, a protein may locally interact with DNA so as to increase (promote) or decrease (inhibit or repress) the production rate for mRNA corresponding to a particular gene.  In this case, the DNA-interacting protein is called a \textit{transcription factor} because it directly influences whether mRNA is transcribed.  The genes expressed by cells during embryonic development will determine the ultimate ``identity'' of the cell (e.g., a nerve or skin cell) in the adult organism.

	Here, we consider a simple model of an autoregulatory process in which each cell transcribes mRNA $m$ that is translated into protein $p$ that in turn modifies the production rate of $m$.  The signaling molecule $v$, generated in exact proportion to $p$, also regulates $p$ production in the self and neighbors by modulating the production rate of $m$.  The system dynamics are
	\begin{equation}
		\begin{cases}
		\dot{m}_i = -\gamma_m m + \alpha_m f(v_i,u_i,p_i) \\
		\dot{p}_i = -\gamma_p p + \alpha_p m  \\
		y = p \\
		v = Mp
		\end{cases} \label{eq:autoreg} \
	\end{equation}
	where $\gamma_m$, $\gamma_p$ are the degradation (decay) rates of mRNA and protein respectively, and $\alpha_m$, $\alpha_p$ are the corresponding transcription or translation rates.  The function $f(v_i,u_i,p_i)$ captures the influence of the coupling, input, and protein on the production rate of the mRNA and therefore of the protein.

	When linearized at steady state, the system becomes
	\begin{equation}
		\begin{cases}
		\dot{\tilde{m}}_i = -\gamma_m \tilde{m}_i + \alpha_m \left( F_v \tilde{v}_i + F_u\tilde{u}_i + F_p\tilde{p}_i \right) \\
		\dot{\tilde{p}}_i = -\gamma_p \tilde{p}_i + \alpha_p \tilde{m}_i \\
		\tilde{y} = \tilde{p} \\
		\tilde{v} = M\tilde{p}
		\end{cases} \nonumber \
	\end{equation}
	where $F_v := \frac{\partial f}{\partial v_i} \big|_{(\bar{m}^*,\bar{p}^*,\bar{v}^*,\bar{u}^*)}$, $F_u := \frac{\partial f}{\partial u_i} \big|_{(\bar{m}^*,\bar{p}^*,\bar{v}^*,\bar{u}^*)}$, and ${F_p := \frac{\partial f}{\partial p_i} \big|_{(\bar{m}^*,\bar{p}^*,\bar{v}^*,\bar{u}^*)}}$.
	
	Define $\alpha := \alpha_m \alpha_p$ and $\gamma := \gamma_m \gamma_p$.  Note that the steady-state protein concentration is a linear multiple of the steady-state mRNA concentration, such that mathematically a molecule produced at rate $\alpha f(v_i,u_i,p_i)$ and decayed at rate $\gamma$ would have the same steady-state concentration as the protein in \eqref{eq:autoreg}.  Indeed, it is not uncommon for the dynamics of transcription and translation to be lumped together (usually by neglecting mRNA dynamics) in mathematical models such as those presented later in this paper.
	
	The steady-state solution to the perturbed system yields filter coefficients
	\begin{equation}
	\left[S\right]_{kk} = \frac{\frac{\alpha}{\gamma}F_u}{1 - \frac{\alpha}{\gamma}\left( F_p + F_v\lambda_k\left(M\right) \right)} \nonumber \
	\end{equation}
	for $k = 0, 1, ..., N-1$.  For the homogeneous steady state to be stable---and therefore for the filtering approach to be applicable---we require
	\begin{equation}
	\frac{\alpha}{\gamma}\left( F_v \lambda_k\left(M\right) + F_p \right) < 1. \nonumber \
	\end{equation}
	We henceforth assume this condition is satisfied.
	
	Recall that the spatial modes are given by $t_k$, the columns of the matrix $T$ that diagonalizes $M$.  The perturbing input can be written as
	\begin{equation}
	\tilde{u} = \left(u - \underline{u}\right) = T\hat{u} = \sum_{k = 0}^{N -1}\hat{u}_kt_k. \nonumber \
	\end{equation}
	The coefficients $\hat{u}_k$ (the entries of $\hat{u}$) are the weights assigned to each of the spatial modes $t_k$.  The steady-state perturbed readout is given by
	\begin{equation}
	\tilde{y}^* = \sum_{k = 0}^{N-1}\left[S\right]_{kk}\hat{u}_kt_k, \label{eq:perturbed_yss_by_mode} \
	\end{equation}
	such that the readout in the $i$th cell is given by $\underline{\bar{y}} + \tilde{y}^*_i$.
		
	The $\mathcal{H}_2$ norm for the $k$th spatial mode is analytically calculated to be $\frac{F_u}{2\alpha} \left[S\right]_{kk}$.  This relationship indicates that the modes in the system respond identically to within a scaling factor to both persistent spatial disturbances and temporally varying white noise inputs.
	
	Figure \ref{fig:img_filter_examples} exemplifies how the choice of interaction type and interconnectivity affects the filtering behavior of the system with no autoregulation.  In particular, activation of neighbors tends to cause the system to amplify low spatial frequencies, while inhibition of neighbors introduces amplification at high spatial frequencies.
		
	To investigate the effect of autoregulation, suppose we fix all parameters except $F_p$.  As $F_p \rightarrow -\infty$ all filter coefficients approach 0.  This attenuating behavior occurs because allowing a protein to effectively shut down its own production prevents the system from responding to signal.
		
	For $F_p > 0$ (autoactivation), increasing $F_p$ disproportionately increases the coefficients at spatial modes with low eigenvalues.  For $T$ corresponding to the DFT or DCT-2, the lower eigenvalues are associated with higher-frequency spatial modes.  In the case of lateral inhibition ($F_v < 0$), the filter coefficients already amplify high-frequency spatial modes relative to intermediate ones (Figure \ref{fig:img_filter_examples}), such that adding autoactivation enhances the filter's intrinsic highpass characteristics.  Indeed, mechanisms involving lateral inhibition and autoactivation have been conjectured to increase the sharpness of boundary formation in systems of patterned cells responding to exponential input \cite{jaeger_2011}\cite{sprinzak_et_al_2010}.

	\begin{figure}[!hbtp]
	\includegraphics[width=\columnwidth]{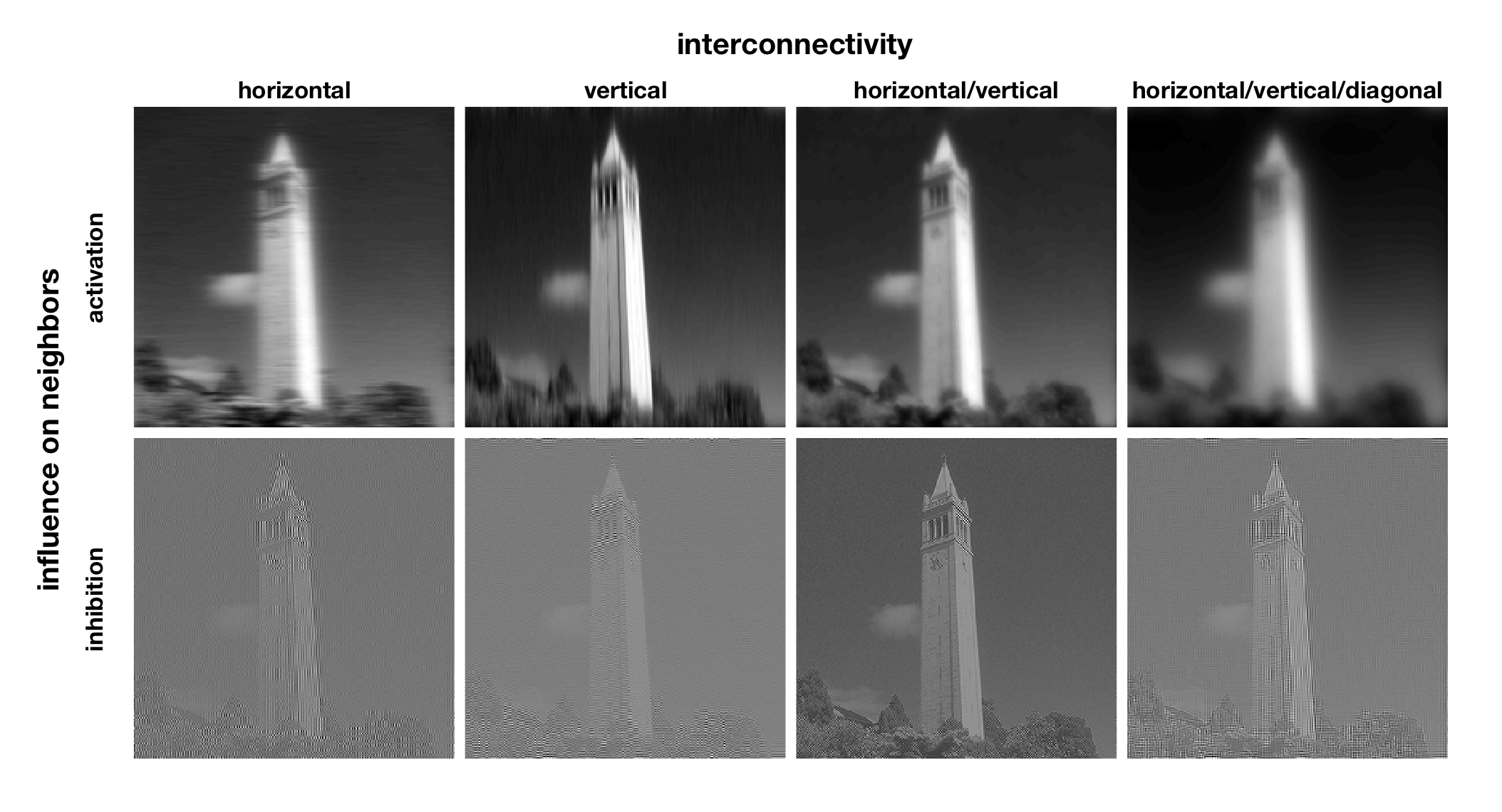} \\
	\caption{\textbf{A minimal model of gene expression demonstrates that the same input prepattern produces different readouts depending on the interconnectivity and interaction type among cells.}  The $i$th cell has dynamics given by \eqref{eq:autoreg} for $\alpha = \gamma = 1$ with no autoregulation, i.e., $F_p = 0$.  $F_v < 0$ corresponds to inhibition of neighbors while $F_v > 0$ implies activation of neighbors.  The filter coefficients are thus $\left[S\right]_{kk} = \left(1 - F_v\lambda_k(M)\right)^{-1}$.  Pictured is the readout $\tilde{y}^*$ given the same perturbing input $\tilde{u}$ to $N = 62,500$ cells arranged in a $250 \times 250$ rectangular array, with one image pixel corresponding to each cell and the intensity of the pixel corresponding to the protein concentration.  Interconnectivities vary by column; boundary conditions in all cases are periodic.  Connection strengths are identical and assumed to incur no cost to the interacting cells (i.e., $[M]_{ii} = 0$).  To best exemplify the effect of the interconnections, $F_v$ was modified for each of the filtered images to give the highest magnitude of eigenvalues without destabilizing the underlying dynamical system.  The readout in each cell is calculated according to \eqref{eq:perturbed_yss_by_mode}.  If cells activate their neighbors then the filter acts as a lowpass (attenuates short wavelengths) that blurs the underlying spatial input along the same dimension as the interconnections.  Inhibition sharpens lines orthogonal to the interconnections by enhancing contrast parallel to the interconnections.  The images are individually normalized.} \label{fig:img_filter_examples}
	\end{figure}

	\clearpage	
	\section{1D APPLICATION: NOTCH-DELTA} \label{sec:notch-delta}
	The Notch-Delta patterning mechanism is a lateral inhibition system that is responsible for diverse developmental phenomena including neural and epidermal fate determination in the fruit fly \textit{Drosophila melanogaster}.  Cells produce both Notch and Delta, which are proteins found in the cell membrane.  Delta on the surface of one cell binds Notch on the surface of neighboring cells to inhibit those neighbors' Delta production, thereby relieving inhibition on the cell's own Delta production by decreasing the potential for the neighbors to bind its Notch.  With the appropriate interaction strengths, such mutual inhibition between neighbors will ultimately generate a checkerboard pattern in which cells expressing high Delta are adjacent to cells expressing low Delta.  This has significant consequences for organismal development: Notch that is bound by Delta on an adjacent cell will cleave in two---preventing it from further signaling---and the portion left inside the cell will signal the cell to express target genes that influence the choice of cell identity.  A cell whose neighbors express more Delta is more likely to have bound Notch and therefore more likely to adopt a particular fate \cite{artavanis-tsakonas_et_al_1999}, \cite{heitzler_simpson_1991}.
	
	Patterning in a Notch-Delta system may arise spontaneously \cite{arcak_2013}, \cite{collier_et_al_1996} or through modification of a prepattern.  In the case of \textit{Drosophila} wing development, the gene \textit{veinless} is expressed in an exponential gradient decreasing in either direction from what will become the center of a vein.  The level of \textit{veinless} expression in a cell determines the Delta production rate at that cell.  Notch activity occurs in two peaks, one on either side of the center, where further vein development is restricted to occur.  One model of the Notch-Delta mechanism suggests that so-called mutual inactivation, when Notch and Delta on the same cell inhibit each other's activity, enables sharper and more robust patterning than is achieved with lateral inhibition alone \cite{sprinzak_et_al_2010}, \cite{sprinzak_et_al_2011}.
	
	The authors of \cite{sprinzak_et_al_2010} considered a line of cells with periodic boundary conditions, corresponding to the interconnection matrix
	\begin{equation}
	M = \frac{1}{2}\left[ \begin{matrix}
	0 & 1 & 0 & \hdotsfor{1} & 0 & 1 \\
	1 & 0 & 1 & \hdotsfor{1} & 0 & 0 \\
	\vdots & \vdots & \vdots & \ddots & \vdots & \vdots \\
	1 & 0 & 0 & \hdotsfor{1} & 1 & 0
	\end{matrix} \right]. \nonumber \
	\end{equation}
	The diagonal entries are zero to reflect the fact that Notch ($N$) and Delta ($D$) interact via contact with neighbors rather than diffusion, while the factor of $\frac{1}{2}$ ensures that $v_N$ is the average Notch from neighbors and $v_D$ is the average Delta from neighbors.  Because $M$ is circulant, the spatial modes correspond to the eigenvectors of the DFT matrix.
	
	We discretize the input gradient of Delta production rate $\beta_D(\cdot)$ into $\beta_{D_i}$, $i = 0, 1, ..., N-1$ and let $\underline{\bar{\beta}}_D$ be the mean of the $\beta_{D_i}$.  We then define $u_i := \beta_{D_i} - \underline{\bar{\beta}}_D$, $x_i^T := [N_i, D_i, R_i]$, and $v_i := [v_{N_i}, v_{D_i}]$ where readout $R$ is a reporter for Notch activity (i.e., is expressed from a target gene for Notch activity).
	
	The authors of \cite{sprinzak_et_al_2011} propose four models of the Notch-Delta patterning mechanism that involve mutual inactivation, lateral inhibition, or both.  As an example we present the linearization for the mutual inactivation (MI) model; equations and linearizations for the lateral inhibition with mutual inactivation (LIMI) and simplest lateral inhibition by mutual inactivation (SLIMI) models may be found in Supplementary Section \ref{sec:D-N_model_calcs}.
	
	The system equations for the MI model are
	\begin{equation}
	\begin{cases}
	\dot{N_i}(t) = \beta_N - \gamma N_i(t) - \frac{N_i(t)v_{D_i}(t)}{k_t} - \frac{N_i(t)D_i(t)}{k_c} \\
	\dot{D_i}(t) = \underline{\bar{\beta}}_D + u_i - \gamma D_i(t) - \frac{D_i(t)v_{N_i}(t)}{k_t} - \frac{N_i(t)D_i(t)}{k_c} \\
	\dot{R}_i(t) = \beta_R\frac{\left(N_i(t)v_{D_i}(t)\right)^n}{k_{RS} + \left(N_i(t)v_{D_i}(t)\right)^n} - \gamma_R R_i(t) \\
	y_i(t) = Cx_i(t) \\
	w_i(t) = \left[ \begin{matrix}
	N_i(t) \\ D_i(t)
	\end{matrix} \right] \\
	v(t) = \left( M \otimes I_2 \right) w(t) \
	\end{cases} \label{eq:sprinzak_MI_1D} \
	\end{equation}
	where $\gamma$, $\gamma_R$ are decay rates, $k_t^{-1}$ is the rate at which Delta and Notch bind each other on neighboring cells, $k_c^{-1}$ is the strength of mutual inactivation, and $k_{RS}$, $n$ are parameters determining how strongly bound Notch promotes reporter expression.  Note that mRNA is not explicitly incorporated into the model, such that the dynamics are effectively lumped into the production and degradation terms for the proteins.
	
	Linearization about the steady state with all $u_i = 0$ yields
	\begin{align}
	A &= \left[ \begin{matrix}
	-\gamma -\frac{\underline{\bar{v}}_{D}^*}{k_t} - \frac{\underline{\bar{D}}^*}{k_c} & -\frac{\underline{\bar{N}}^*}{k_c} & 0 \\
	-\frac{\underline{\bar{D}}^*}{k_c} & - \gamma - \frac{\underline{\bar{v}}_{N}^*}{k_t} - \frac{\underline{\bar{N}}^*}{k_c} & 0 \\
	b_1 & 0 & -\gamma_R
	\end{matrix} \right], \nonumber \\
	B_v &= \left[ \begin{matrix}
	0 & -\frac{\underline{\bar{N}}^*}{k_t} \\
	-\frac{\underline{\bar{D}}^*}{k_t} & 0 \\
	0 & b_2
	\end{matrix} \right], ~~
	B_u = \left[ \begin{matrix}
	0 \\
	1 \\
	0
	\end{matrix} \right], ~~ G = \left[ \begin{matrix} 1 & 0 & 0 \\ 0 & 1 & 0 \end{matrix} \right], \nonumber \
	\end{align}
	where we have defined
	\begin{equation}
	b_1 := \beta_Rnk_{RS}\frac{\underline{\bar{v}}_{D}^{*n}N^{*n-1}}{\left(k_{RS}+\left(\underline{\bar{N}}^*\underline{\bar{v}}_{D}^*\right)^n\right)^2}, ~ b_2 := \frac{\underline{\bar{N}}^*}{\underline{\bar{v}}_{D}^*}b_1 \nonumber \
	\end{equation}
	and $C$ is chosen depending on the readout.  It can be shown that the linearized dynamical system is stable for all nonnegative and thus biologically attainable parameter values (Supplementary Section \ref{sec:supp_mi_model}), strengthening the argument that patterning may not require instability.
	
	To examine the effects of identical (correlated) vs. separate (uncorrelated) white noise inputs to both Delta and Notch, we first modify \ref{eq:sprinzak_MI_1D} so that a single input appears in the equations for both $\dot{N_i}$ and $\dot{D_i}$ in the correlated case and two independent inputs appear in each of these equations for the uncorrelated case.  Accordingly, we then calculate the $\mathcal{H}_2$ norm for
	\begin{equation}
	B_u^{corr} = \left[ \begin{matrix}
	1 \\
	1 \\
	0
	\end{matrix} \right], ~
	B_u^{uncorr} = \left[ \begin{matrix}
	1 & 0 \\
	0 & 1 \\
	0 & 0
	\end{matrix} \right]. \nonumber \
	\end{equation}
	
	\begin{figure}[!hbtp]
		\centering
		\includegraphics[scale=0.5]{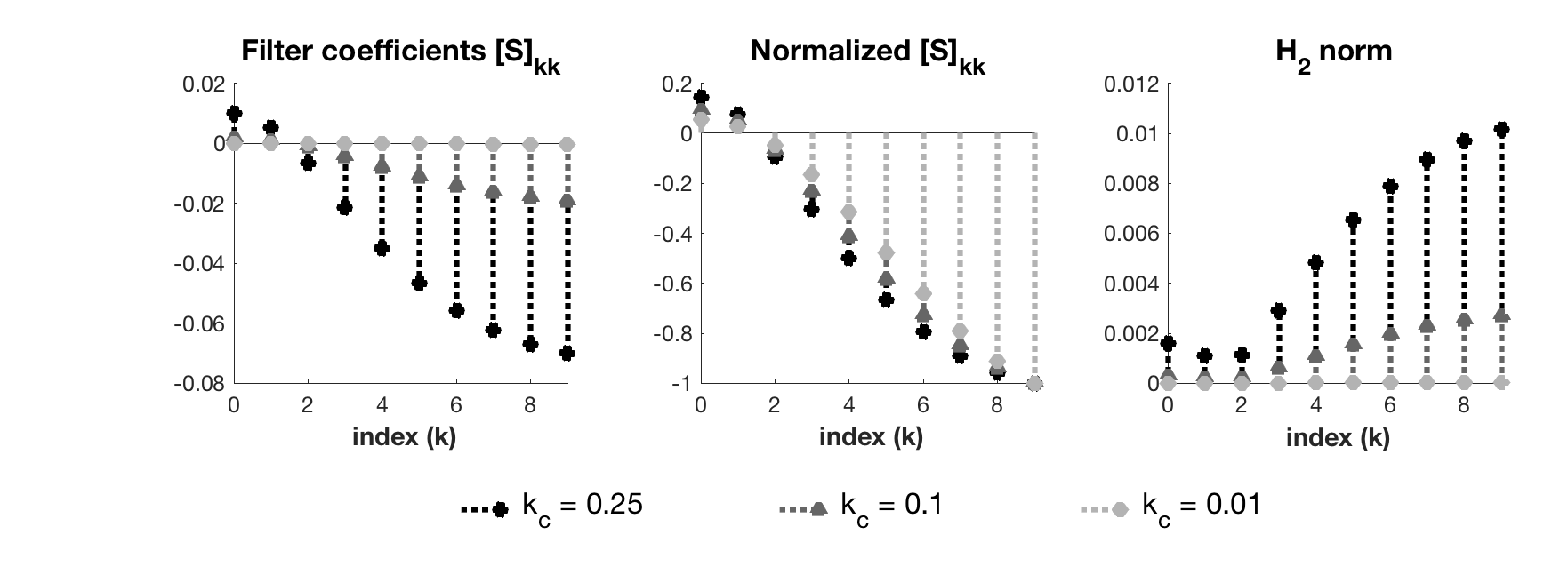}
		\caption{\textbf{The filter coefficients and $\mathcal{H}_2$ norm for the MI model of Notch-Delta interaction reveal how changes to parameter values enhance high frequencies from an input gradient of Delta production to readout Notch activity.}  The magnitude of the coefficients decreases with greater mutual inactivation strength (lower $k_c$), indicating that greater inhibition reduces overall activity.  The spatial modes correspond to the DFT basis and are indexed by $k$ such that the $k$th mode has frequency $\frac{2\pi k}{N}$.  The coefficients exhibit mirror-image symmetry about $k = \frac{N}{2}$; we plot only the first half of the coefficients to better visualize the filter's characteristic highpass shape.  The greater the mutual inactivation, the greater the amplification of high frequencies relative to lower ones, as revealed by a plot of the coefficients individually normalized to the maximum in each set.  The $\mathcal{H}_2$ norm is qualitatively similar to the highpass filter characteristic though smaller in magnitude, with more dramatic relative differences between values of $k_c$.  Parameters are given in Supplementary Table \ref{tab:supp_sprinzak_params}.}
		\label{fig:sprinzak_mi_kc}
	\end{figure}

	\begin{figure}[!hbtp]
		\centering
		\includegraphics[scale=0.5]{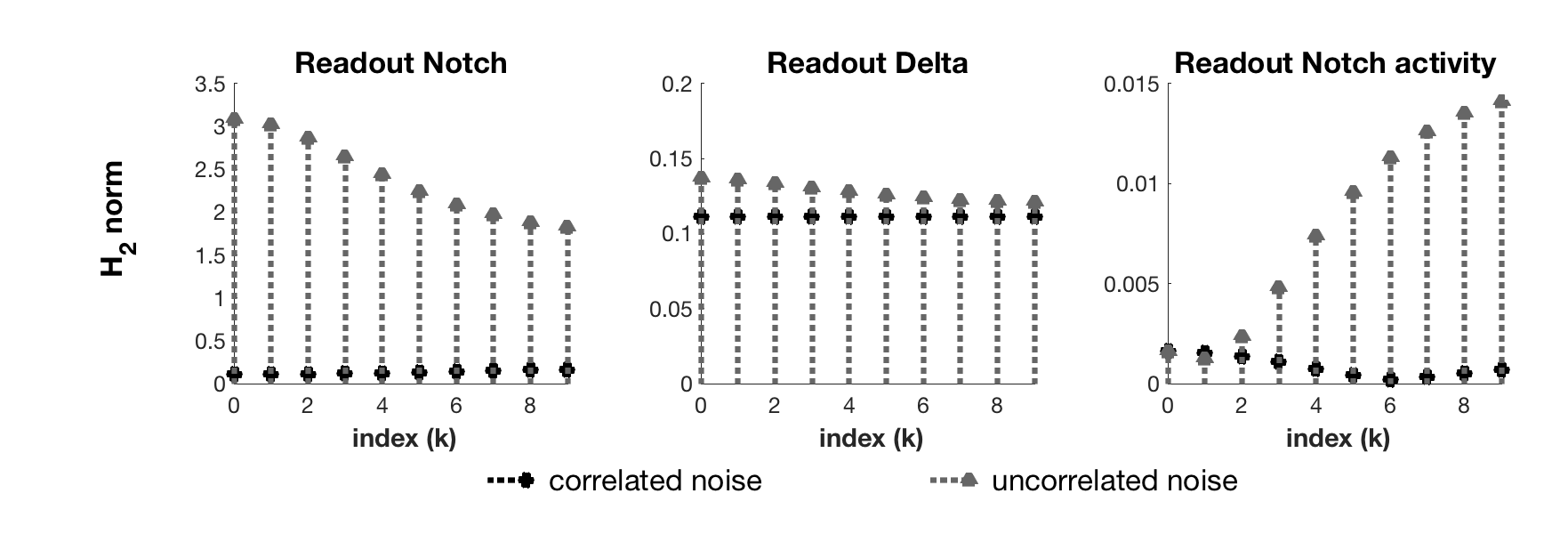}
		\caption{\textbf{Filter behavior is robust with respect to noisy inputs to Notch or Delta.}  Noise that is completely correlated between Notch and Delta is more strongly and uniformly rejected than completely uncorrelated noise.  However, uncorrelated noisy inputs tend to emphasize the inherent highpass characteristics with respect to output Notch activity, suggesting that moderate levels of white noise do not compromise filter behavior.  Pictured here are the norms for the MI model; the other models exhibit similar behavior (Figure \ref{fig:supp_sprinzak_h2_uc_comps}).  Parameters are given in Supplementary Table \ref{tab:supp_sprinzak_params}.}
		\label{fig:sprinzak_h2_uc}
	\end{figure}
	
	\begin{figure}[!hbtp]
		\centering
		\includegraphics[scale=0.5]{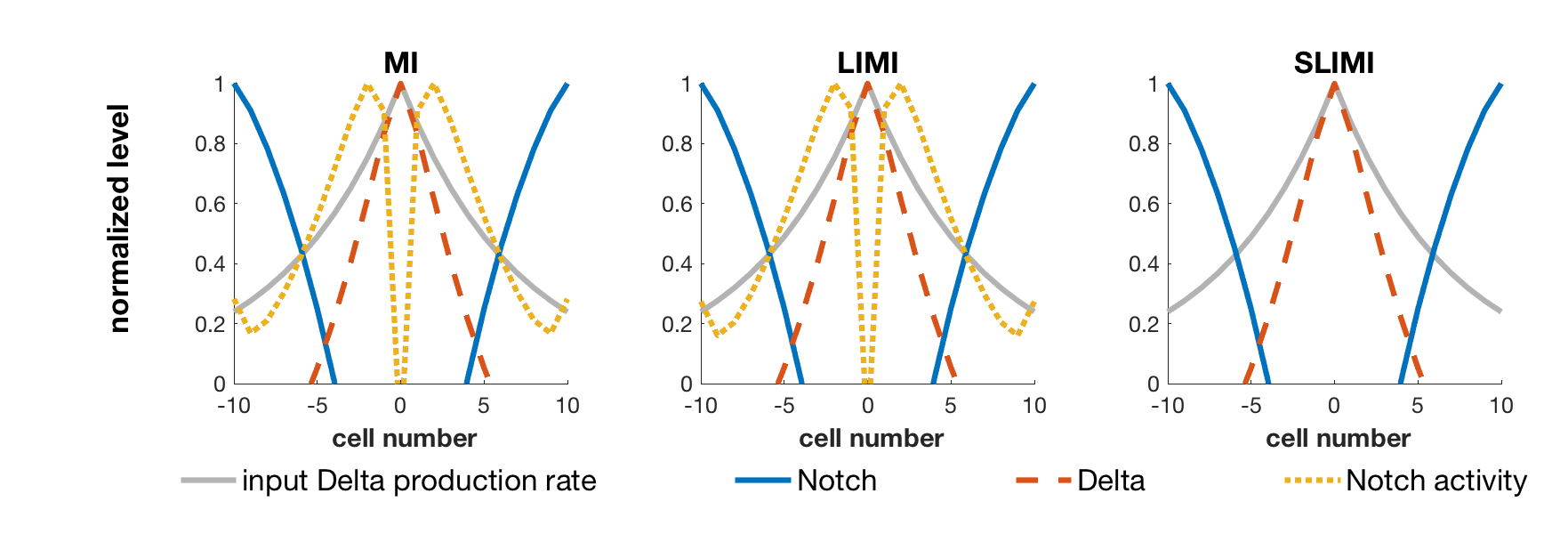}
		\caption{\textbf{A linearized system qualitatively reproduces the double peaks in Notch activity predicted from full nonilnear simulations.}  A two-sided exponential input gradient of Delta production rate (solid light gray) results in two sharp bands of Notch activity (dotted yellow) that spatially segregates steady-state levels of Notch (solid blue) and Delta (dashed  orange).  Curves are normalized to their respective maxima.  Note that the SLIMI model lacks a reporter protein and so does not have an output measure for Notch activity.  These plots correspond to Figure 4C in \cite{sprinzak_et_al_2010}.  See also Figure \ref{fig:supp_sprinzak_comps}.}
		\label{fig:sprinzak_comps_sm}
	\end{figure}

	\subsection{Comparison of Models}
	The MI, LIMI, and SLIMI models from reference \cite{sprinzak_et_al_2011} produce substantially similar readouts (Figure \ref{fig:sprinzak_comps_sm}), filter characteristics, and $\mathcal{H}_2$ norms (Figures  \ref{fig:supp_sprinzak_comps} and \ref{fig:supp_sprinzak_h2_uc_comps}).  Mutual inactivation (lower $k_c$) decreases the magnitude of the coefficients and therefore the final Notch and Delta concentrations, but exaggerates the intrinsic highpass characteristics of the filter, producing the sharper peaks in Notch activity predicted by \cite{sprinzak_et_al_2010}.  Analysis of the $\mathcal{H}_2$ norm reveals that regardless of readout, noise that is completely uncorrelated between Delta and Notch production rates is favored by the same frequencies as the system filter, while noise that is completely correlated between the production rates is almost uniformly rejected relative to uncorrelated noise (Figures \ref{fig:sprinzak_mi_kc} and \ref{fig:sprinzak_h2_uc}).  Together, these observations suggest that time-varying stochastic inputs---unless they are of exceptionally large magnitude---do little to combat the intrinsic behavior of the filter, contributing to the robustness of the developmental program.
	
	\section{APPLICATION: DIGIT FORMATION} \label{sec:raspopovic}
	Digits in developing vertebrate embryos originate from a flat paddle-shaped layer of cells that form the limb bud.  A crucial step in digit patterning involves specifying which cells in the paddle will become digits and which will die to create the space between digits \cite{tickle_2006}, \cite{zeller_et_al_2009}.  This periodic pattern of digit with interdigit has been proposed to originate with spatially periodic expression of the gene \textit{sox9}, which produces a protein that regulates transcription of the genes \textit{wnt} and \textit{bmp}.  In turn, these genes code transcription factors Wnt and Bmp that regulate Sox9 production \cite{raspopovic_et_al_2014}.
	
	Cell cultures from developing embryos grown on plates show Turing-like patterns where Sox9 is out of phase with Wnt and Bmp.  Turing patterns typically arise in chemical reaction systems with at least two types of diffusible molecules produced at every point in space, where the activation/inhibition relationship between the types is such that the homogeneous solution to the resulting dynamical system is unstable owing to the difference in diffusion rates between the two molecules.  Such a reaction-diffusion model has been proposed to generate the observed Sox9/Wnt/Bmp pattern from stochastic initial conditions within a particular parameter range \cite{raspopovic_et_al_2014}.  Our discretization of the model suggests that such a pattern might be observed even if the parameters do not satisfy the conditions for diffusion-driven instability.
	
	Consider the Sox9/Bmp/Wnt network with diffusion distance $l$ between cells.  Let $s$, $b$, and $z$ represent the concentrations of Sox9, Bmp, and Wnt respectively, such that $x_i = [s_i,~b_i,~z_i]$ and $v = [v_b,~v_z]$.  Let the input be random cell-to-cell variation in background protein production rate, i.e., the production rate of protein in the absence of promotion or inhibition, as from cell-to-cell variability in transcription or translation rates (see Supplementary Figure \ref{fig:supp_gene_exp_schematic}).  The dynamics within cell $i$ and the coupling are given by
	\begin{equation}
	\begin{cases}
	\dot{s}_i(t) = \alpha_s + u_i + k_2b_i(t) - k_3z_i(t) - (s_i(t) - s_0)^3 \\
	\dot{b}(t) = \alpha_{bmp} + u_i - k_4s_i(t) - k_5b_i(t) + \frac{d_b}{l^2}v_{b_i}(t) \\
	\dot{z}_i(t) = \alpha_{wnt} + u_i - k_7s_i(t) - k_9z_i(t) + \frac{d_z}{l^2}v_{z_i}(t) \\
	w_i(t) = \left[ \begin{matrix}
	b_i(t) \\ z_i(t) \end{matrix} \right] \\
	v(t) = (M \otimes I_2)w(t) \
	\end{cases} \label{eq:raspopovic_sys} \
	\end{equation}
	where $\alpha$ are background production rates, $k$ are interaction rates, and $d$ are diffusivities.
	
	Linearization about the homogeneous steady state yields
	\begin{align}
	A &= \left[ \begin{matrix}
	-3(\underline{\overline{s}}^* - s_0)^{2} & k_2 & -k_3 \\
	-k_4 & -k_5 & 0 \\
	-k_7 & 0 & -k_9 \
	\end{matrix} \right], \nonumber \\
	B_v &= \left[ \begin{matrix}
	0 & 0 \\
	\frac{d_b}{l^2} & 0 \\
	0 & \frac{d_z}{l^2}
	\end{matrix} \right], ~~ B_u = \left[ \begin{matrix}
	1 \\
	1 \\
	1
	\end{matrix} \right], ~~
	G = \left[ \begin{matrix} 0 & 1 & 0 \\ 0 & 0 & 1 \end{matrix} \right]. \nonumber \
	\end{align}
	
	\subsection{Spatial Modes in 2D} \label{sec:spatial_modes_2D}
	For this example we will consider a two-dimensional, rectangular $N_R \times N_C$ array of $N := N_CN_R$ cells indexed from $0$ to $N_CN_R - 1$ starting in the upper lefthand corner from top to bottom and then left to right, i.e.,
	\begin{equation}
	\begin{array}{c|c|c|c|c}
	0 & N_R & 2N_R & \dots & (N_C-1)N_R \\ \hline
	1 & N_R+1 & 2N_R+1 & \dots & (N_C-1)N_R+1 \\ \hline
	\vdots & \vdots & \vdots & \ddots & \vdots \\ \hline
	N_R-2 & 2N_R-2 & 3R-2 & \dots & N_CN_R-2 \\ \hline
	N_R-1 & 2N_R-1 & 3N_R-1 & \dots & N_CN_R-1 \
	\end{array}. \label{eq:array_indices} \
	\end{equation}
	
	It is known (e.g., \cite{othmer_scriven_1971}) that if any isolated row has interconnection matrix $M_R \in \mathbb{R}^{N_C \times N_C}$ and any isolated column has interconnection matrix $M_C \in \mathbb{R}^{N_R \times N_R}$, then the full matrix $M$ for the interconnectivity of the entire array is
	\begin{equation}
		M := \left(M_R \otimes I_{N_R}\right) + \left(I_{N_C} \otimes M_C\right). \nonumber \
	\end{equation}
	If $T_R$ and $T_C$ diagonalize $M_R$ and $M_C$ respectively then $M$ is diagonalized by
	\begin{equation}
		T := \left(T_R \otimes I_{N_R} \right) \left(I_{N_C} \otimes T_C \right) = T_R \otimes T_C, \nonumber \
	\end{equation}
		giving $N_CN_R$ eigenvalues
	\begin{equation}
		\lambda_{m + nN_R}\left( M \right) = \lambda_m\left(M_C\right) + \lambda_n\left(M_R\right) \nonumber \
	\end{equation}
	where $m = 0, 1, ... N_R-1$, $n = 0, 1, ..., N_C-1$.  The $(m,n)$th spatial mode is given by $T_C^mT_R^{nT}$.

	We can explicitly relate the spatial modes for a 2D array of cells to constituent modes in the horizontal and vertical directions by recasting the vector $\hat{y}$ in matrix form.  Let $U$ and $Y$ be matrices arranged as in \eqref{eq:array_indices} where $u_i$ is the input to compartment $i$.  Vector form is recovered through the vectorization operation ${vec(U) = u}$.  The readout matrix $Y$ is defined similarly.  If the matrices $\tilde{U}$ and $\tilde{Y}$ designate perturbations from steady state in the original basis and $\hat{Y}$, $\hat{U}$ designate perturbations in the basis for the spatial modes, then
	\begin{equation}
	\tilde{Y} = T_C \hat{Y} T_R^T = T_C \left( \Lambda_S \odot \hat{U} \right)T_R^T \nonumber \
	\end{equation}
	where $\odot$ is the Hadamard product (element-by-element multiplication) and $\Lambda_S$ has $m$th, $n$th entry
	\begin{equation}
	\left[\Lambda_S\right]_{mn} = -C\left[A+\left(\lambda_m(M_C) + \lambda_n(M_R)\right)B_vG\right]^{-1}B_u. \nonumber
	\end{equation}
	From this it can be seen that the full system alters the input along the $i$th vertical spatial mode and the $j$th horizontal spatial mode defined by the vertical and horizontal connectivities.  For the remainder of this example, we will assume Neumann boundary conditions such that the spatial modes for the rows and columns of $M$ correspond to the DCT-2 (Figure \ref{fig:rect_modes_dct2}).
	
	\begin{figure}[!hbpt]
		\centering
		\includegraphics[width=\columnwidth]{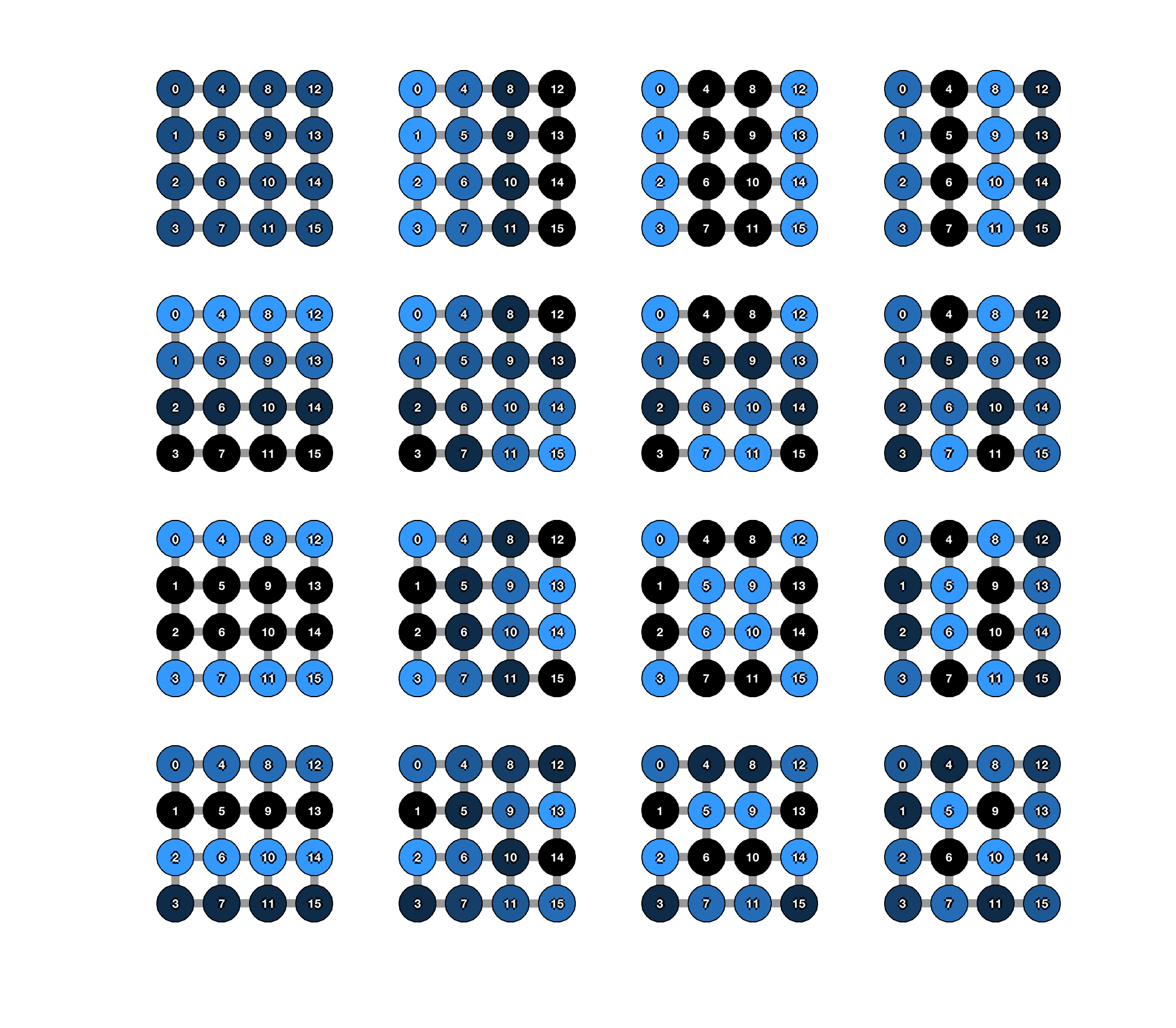}
		\caption{A complete set of spatial modes for a 2D interconnectivity with Neumann boundary conditions (DCT-2 basis) on a $4 \times 4$ rectangular array.  Modes are indexed such that the $(m,n)$th mode has frequency $\frac{\pi m}{N_R}$ in the vertical direction (down rows) and $\frac{\pi n}{N_C}$ in the horizontal direction (across columns).}
		\label{fig:rect_modes_dct2}
	\end{figure}

	\subsection{Analysis}
	
	\begin{figure}[!hbtp]
		\centering
		\includegraphics[width=\columnwidth]{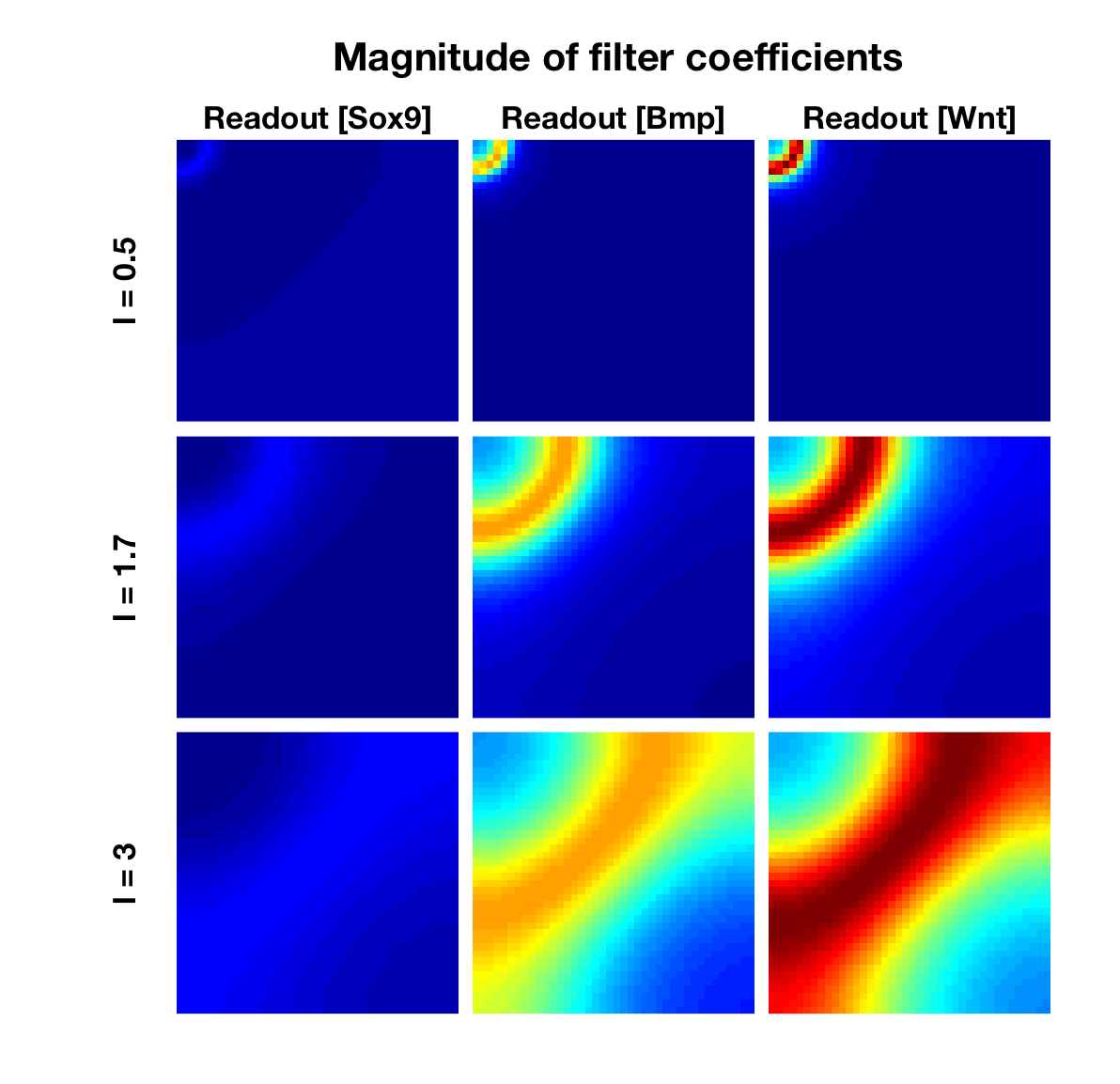}
		\caption{\textbf{The stable, linearized reaction-diffusion system behaves as a bandpass filter for Sox9 (left), Bmp (center), and Wnt (right), resulting in a spatially periodic output} (see Figure \ref{fig:raspopovic_color_plates}).  Pictured are heat maps of the magnitude of the filter coefficients for the three readouts assuming Neumann boundary conditions (DCT-2 basis) in both dimensions, such that vertical frequency increases down rows (higher $m$) and horizontal frequency increases across columns (higher $n$).  Increasing the distance between cells ($l$) increases the frequency of the passband but decreases the sharpness of the dropoff.  The readout concentration of Sox9 is out of phase from the Bmp and Wnt concentrations due to the fact that the coefficients of $S$ have an extra multiplicative factor of ${-1 = e^{i\pi}}$, or a phase shift of $\pi$, relative to the coefficients when the readout is [Bmp] or [Wnt].  Parameters are as given in Table ST4 of \cite{raspopovic_et_al_2014} with $s_0 = 11$ instead of $10$, i.e., $s_0 \neq \underline{\overline{s}}^*$ and therefore $\left[A\right]_{00} \neq 0$ (see also Supplementary Table \ref{tab:supp_raspopovic_params}).  This choice of $s_0$ stabilizes the dynamical system with diffusion, thereby violating Turing conditions.  Here, $N_R = N_C = 40$ for a total of $N = 1600$ cells.  Images are normalized to the same scale (min. 0, max. 26.6).}
		\label{fig:raspopovic_eigs}
	\end{figure}

	\begin{figure}[!hbtp]
		\centering
		\includegraphics[width=\columnwidth]{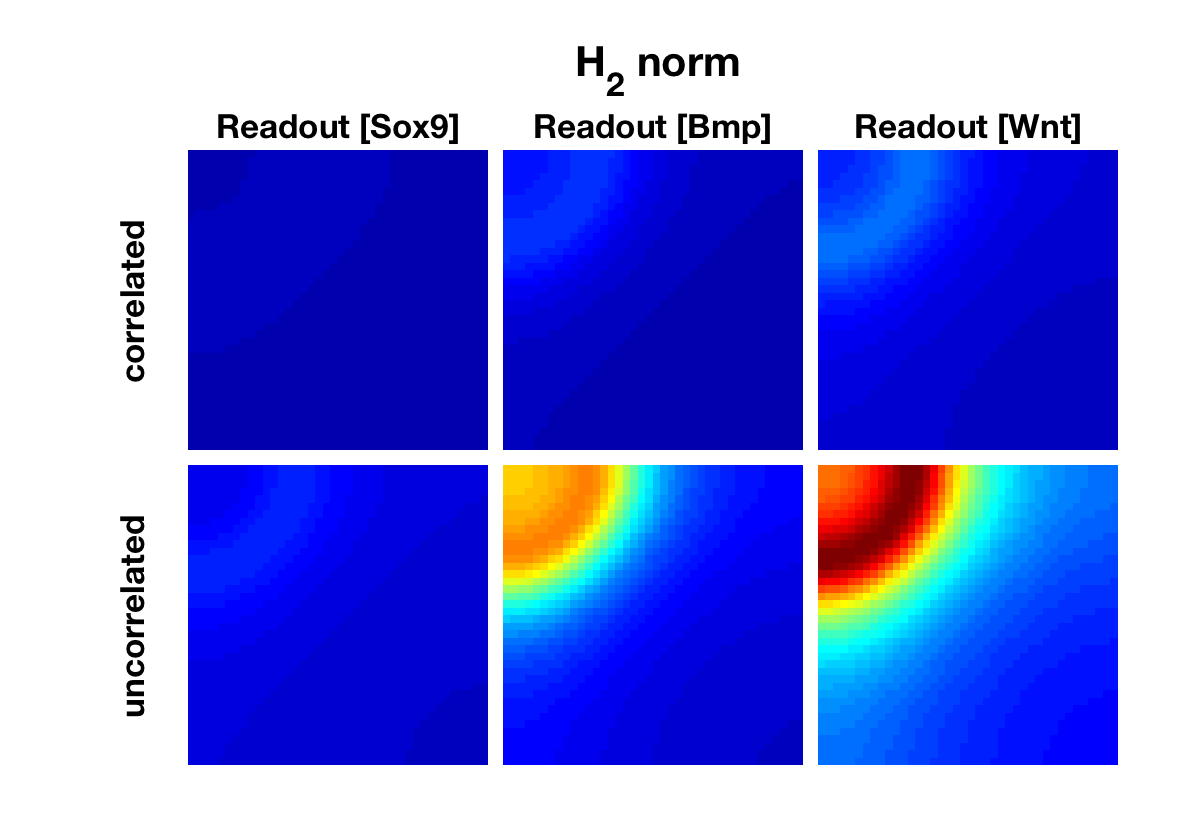}
		\caption{\textbf{The stabilized Sox9-Bmp-Wnt network emphasizes noise in the same frequency bands as those favored by the filter.}  The $\mathcal{H}_2$ norm for correlated noise when the readout is [Sox9] is less than 1 in magnitude, indicating noise rejection at all frequencies, while uncorrelated noise is amplified at all frequencies for readout [Wnt].  Uncorrelated noise, despite being highly amplified for readouts [Bmp] and [Wnt], is rejected at frequencies higher than the upper end of the filter passband for [Sox9] and only weakly amplified at lower frequencies, perhaps as a result of the opposing influences of Bmp and Wnt on \textit{sox9} expression.  Parameters are as in \ref{fig:raspopovic_eigs} with $l = 1.7$ (see also Supplementary Table \ref{tab:supp_raspopovic_params}).  Images are normalized to the same scale (min. 0, max. 111).}
		\label{fig:raspopovic_h2}
	\end{figure}
	
	\begin{figure}[!hbtp]
		\centering
		\includegraphics[width=\columnwidth]{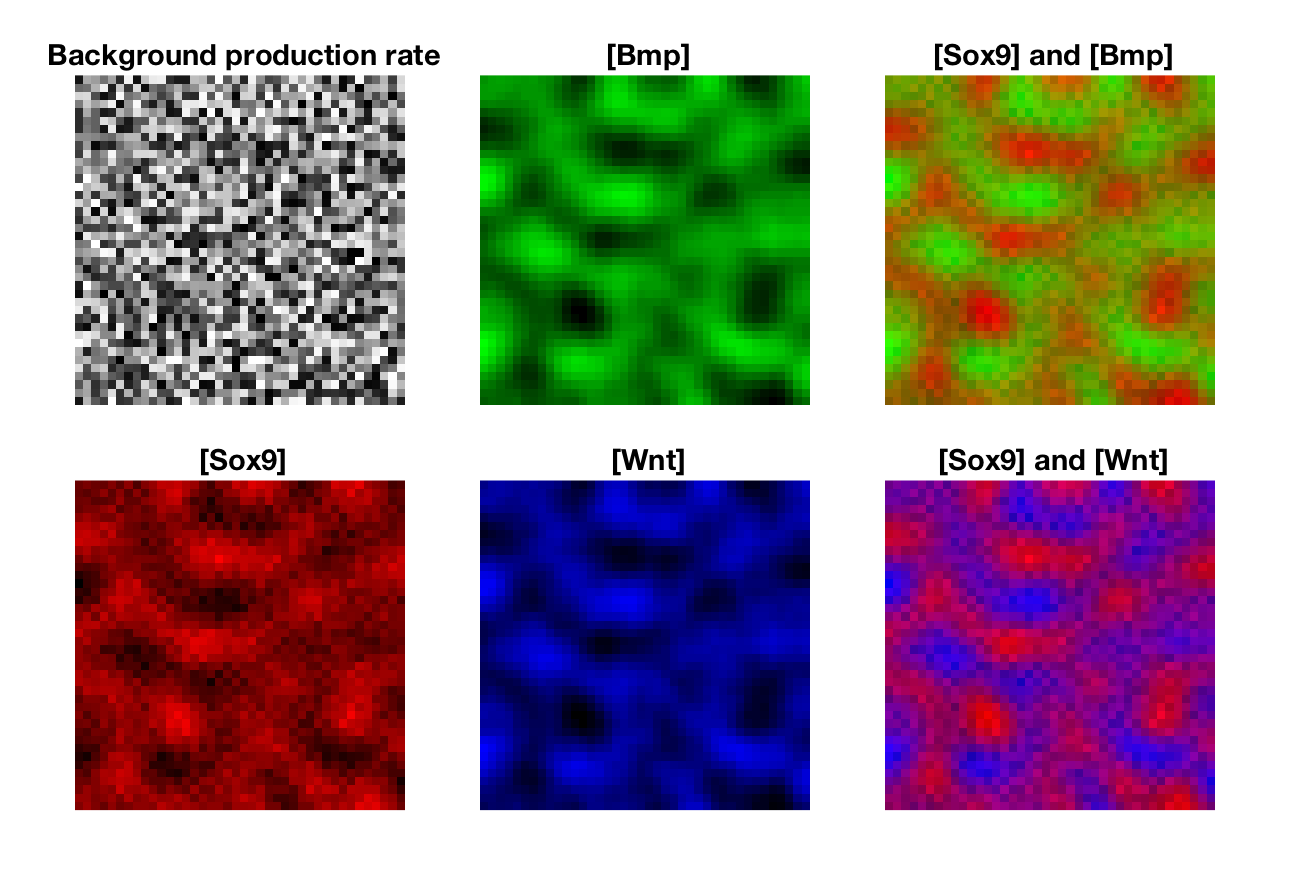}
		\caption{\textbf{A linearized, discrete 2D system with random but constant-in-time variation between points replicates predictions from a Turing reaction-diffusion model even when the parameters do not satisfy the conditions for Turing-driven instability.}  The color-coded visualization reproduces images of micromass cultures from Supplementary Figure S3 in \cite{raspopovic_et_al_2014}, showing similar periodic striped patterns.  Higher intensity corresponds to higher concentration level and intensity is normalized individually by protein species, even in the overlaid images.  The expression boundaries depicted here are not as sharp as the original Turing model due to the stability of the linearized system.  The concentration of Sox9 is out of phase with Bmp and Wnt concentrations, as seen from the overlaid images.  Parameters are as in \ref{fig:raspopovic_eigs} and \ref{fig:raspopovic_h2} with $l = 1.7$ (see also Supplementary Table \ref{tab:supp_raspopovic_params}), with a constant-in-time input background production rate input that is shared by all reactants (as in \eqref{eq:raspopovic_sys}).  For outputs [Sox9], the filter coefficient of greatest magnitude occurs at $(14,6)$ (and symmetrically also $(6,14)$; see Figure \ref{fig:raspopovic_eigs}), corresponding to a spatial mode comprising a sum of two cosines, the higher of which has period $25$\% the length of one side of the array.  As predicted therefrom, the output pattern has approximately four complete periods at an angle about 67$^{\circ}$ from horizontal.}
		\label{fig:raspopovic_color_plates}
	\end{figure}

	We pick $C$ to monitor Sox9, Bmp, or Wnt concentration and choose $s_0$ such that the Turing instability conditions are not satisfied, i.e., the eigenvalues of $\left( I_N \otimes A \right) + M \otimes \left( B_vG \right)$ are all negative.  Nevertheless, the readout still replicates the spatially periodic patterns predicted by \cite{raspopovic_et_al_2014} for a range of intercellular distances (Figure \ref{fig:raspopovic_color_plates}) owing to the bandpass behavior of the filter (Figure \ref{fig:raspopovic_eigs}).  [Sox9] is out of phase with both [Bmp] and [Wnt], as indicated by the opposing signs of the coefficients in the passband.

	The $\mathcal{H}_2$ norm measurements for the readouts qualitatively emphasize the same frequency bands as their respective filters $\left[S\right]_{kk}$ (Figure \ref{fig:raspopovic_h2}).  For correlated or uncorrelated noise sources, readouts [Bmp] and [Wnt] experience much greater magnification than does [Sox9], suggesting that the opposing effects of Bmp and Wnt on \textit{sox9} expression may mostly cancel each other out at the level of Sox9 concentration.  Relative noise amplification in the same modes favored by the $\left[S\right]_{kk}$ may ensure that stochastic influences do not counteract filter behavior, at the same time that attenuation and evenness in the response to other modes might reduce the relative influence of temporally varying inputs on the readout.  The latter especially may be useful to maintain consistent behavior in a process such as digit formation that takes place over a long timespan.
	
	While our simulations do not refute the hypothesis that a diffusion-driven instability constitutes the biological basis for digit formation, the fact that we can produce a similar pattern with an externally perturbed stable system suggests that not all apparent Turing patterns need arise from an instability.  This observation could significantly ease the search for molecules and proteins that contribute to ``spontaneous'' stripe and spot patterning, as the parameter restrictions required for true Turing instabilities may not be biologically plausible.
	
	\clearpage
	\section{APPLICATION: DIGIT FORMATION WITH A MORPHOGEN GRADIENT} \label{sec:onimaru}
	Expanding on the work of \cite{raspopovic_et_al_2014}, reference \cite{onimaru_et_al_2016} demonstrated that changes to the parameters in the proposed Turing network for digit formation in mice can produce \textit{sox9} expression patterns matching those found in embryonic catshark fins, suggesting that the mechanism has been evolutionarily conserved.  The authors augmented the model with an exponential gradient of fibroblast growth factor (Fgf), a morphogen originating at the fin edge that has been experimentally demonstrated to facilitate normal digit arrangement in mice.  In their model, Fgf represses Sox9 repression of $bmp$ expression ($k_4$) and promotes Sox9 repression of $wnt$ expression ($k_7$).  In simulation, the authors observed that increasing the ratio of Wnt production to Bmp production or decreasing Bmp promotion of $sox9$ expression caused the Turing pattern to transition from stripes to spots.
	
	We implemented the model from \cite{onimaru_et_al_2016} using the following evolution equations:
	\begin{equation}
	\begin{cases}
	\dot{F}_i(t) = \alpha_{F} + u_{1i} - \mu_{F}F_i(t) + \frac{d_f}{l^2}v_{F_i}(t) \\
	\dot{s}_i(t) = \alpha_s + u_{2i} + k_2b_i(t) - k_3z_i(t) - s_i(t)^3 \\
	\dot{b}(t) = \alpha_{bmp} + u_{2i} -k_4\left(1-k_fF_i(t)\right)s_i(t) - k_5 b_i(t) + \frac{d_b}{l^2} v_{b_i}(t) \\
	\dot{z}_i(t) = \alpha_{wnt} + u_{2i} - k_7k_fF_i(t)s_i(t) - k_9 z_i(t) + \frac{d_z}{l^2} v_{z_i}(t) \\
	\dot{w}_i(t) = \left[ \begin{matrix}
	F_i(t) \\
	b_i(t) \\
	z_i(t)
	\end{matrix}\right] \\
	v(t) = \left(M \otimes I_3\right)w(t) \
	\end{cases} \label{eq:onimaru_sys} \
	\end{equation}
	where $\alpha_F$ is the Fgf production rate, $u_1$ represents the source of Fgf, and $u_2$ is random constant-in-time spatial variation in background production rate.  Unlike \cite{onimaru_et_al_2016}, we did not normalize $F$ to $[0,1]$, but we chose $u_1, \underline{u}_1$ such that $0 \leq \tilde{F_i}^* + \underline{\bar{F}}^* \leq 1$ and $0 \leq u_{1_i} - \underline{\bar{u}}_1$.

	As compared to \eqref{eq:raspopovic_sys}, the equations in \eqref{eq:onimaru_sys} are rendered as perturbations to prior steady-state protein concentrations, therefore ``negative'' steady-state values should be interpreted as reductions in concentration relative to preexisting levels.
	
	To handle both background production rate and localized Fgf production we use the generalization to $L$ inputs
	\begin{equation}
	\hat{y}^* := -\left(I_N \otimes C\right)\left[ \left( I_N \otimes A \right) + \Lambda \otimes (B_vG) \right]^{-1}\left[\sum \limits_{k = 1}^{L}\left( I_N \otimes B_{u_k} \right)\hat{u}_k\right] \label{eq:ss_yhat_perturbed_multiple_external} \
	\end{equation}
	where $u_k \in \mathbb{R}^{N}$ is the $k$th input vector and $$B_{u_k} := \frac{\partial f}{\partial u_k} \big|_{(\underline{\bar{x}}^{*},\underline{\bar{v}}^*,\underline{\bar{u}}_1,\underline{\bar{u}}_2,...,\underline{\bar{u}}_L)}$$ is the linearization matrix for one subsystem with respect to the $k$th input when all inputs are held constant in time and space.  To avoid ambiguity, the ``filter'' interpretation is defined with respect to one input, i.e., as one term in the summation \eqref{eq:ss_yhat_perturbed_multiple_external}.
	
	The matrices for the linearization are
	\begin{align}
	A &= \left[ \begin{matrix}
	-\mu_F & 0 & 0 & 0 \\
	0 & -3\underline{\overline{s}}^{*2} & k_2 & -k_3 \\
	k_4k_f\underline{\overline{s}}^* & -k_4\left(1-k_f\underline{\bar{F}}^*\right) & -k_5 & 0 \\
	- k_7k_f\underline{\overline{s}}^* & -k_7k_f\underline{\bar{F}}^* & 0 & -k_9
	\end{matrix} \right], \nonumber \\
	B_v &= \left[ \begin{matrix}
	\frac{d_f}{l^2} & 0 & 0 \\
	0 & 0 & 0 \\
	0 & \frac{d_b}{l^2} & 0  \\
	0 & 0 & \frac{d_z}{l^2}
	\end{matrix} \right], ~~ B_{u_1} = \left[ \begin{matrix}
	1 \\
	0 \\
	0 \\
	0
	\end{matrix} \right], ~~ B_{u_2} = \left[ \begin{matrix}
	0 \\
	1 \\
	1 \\
	1
	\end{matrix} \right], ~~
	G = \left[ \begin{matrix} 1 & 0 & 0 & 0 \\ 0 & 0 & 1 & 0 \\ 0 & 0 & 0 & 1 \end{matrix} \right] \nonumber \
	\end{align}
	where the steady-state concentration of Fgf is $\underline{\bar{F}}^* = \frac{\alpha_F+ \underline{\bar{u}}_1}{\mu_F}$ independent of the other variables.  We stabilized the homogeneous steady-state solution by setting $\alpha_F = 0$ and using a small value of $\underline{\bar{u}}_1$ with the remaining parameters taken from Figures 4 and 5 in \cite{onimaru_et_al_2016}.  This choice of $\alpha_F$ completely localizes the source of Fgf to the input $u_2$.
	
	\subsection{Spatial Modes on a Hexagonal Lattice}
	The hexagonal lattice is the tightest 2D packing arrangement for cells of fixed area and is found in a number of natural systems such as the wing epithelial cells in \textit{Drosophila} \cite{classen_et_al_2005}.  For this example we will derive the lattice from a rectangular array where the columns are offset by 30$^\circ$ from vertical and assume periodic boundary conditions as well as identical spacing between all neighbors.  With the cells numbered as shown in Figure \ref{fig:hex_lattice}, the $(m,n)$th spatial mode corresponds to the $m$th mode horizontally and the $n$th mode on a line at a 60$^\circ$ angle from each row.  Cells in the hexagonal lattice interact with each of their six nearest neighbors such that there are ``diagonal interconnections'' between rows of cells.  We account for the diagonal connections as follows: Define  $C_D := min\left(N_R,N_C\right)$ and $R_D := max\left(N_R,N_C\right)$, and let $P \in \mathbb{R}^{C_D\times C_D}$ be the permutation matrix with lower diagonal ones and the last entry of the first column also one.  Define $P_F := diag\left(P^0,P^1,...,P^{R_D}\right)$.  Then the interconnection matrix for a hexagonal lattice with periodic boundary conditions is given by
	\begin{equation}
	M := \left( M_R \otimes I_R \right) + \left( I_C \otimes M_C \right) + P_F^T\left( M_F \otimes I_{R_D} \right)P_F. \nonumber \
	\end{equation}
	If we let $M_R = M_C = M_F = M_0 \in \mathbb{R}^{N_0 \times N_0}$ be the circulant diffusion matrix, then $M$ has eigenvalues
	\begin{equation}
	\left[\Lambda\right]_{mn} = -6 + 2\left(\cos\frac{2\pi m}{N_0} + \cos\frac{2\pi n}{N_0} + \cos\frac{2\pi (m-n)}{N_0}\right). \nonumber \
	\end{equation}
	From this we see that $\Lambda = \Lambda^T$ and $\left[\Lambda\right]_{mn} = \left[\Lambda\right]_{(N_0-m)(N_0-n)}$.
	
	Further discussion of diagonal interconnectivites and planar lattices more generally is available in Supplementary Sections \ref{sec:dft_diag_2d_eigs} and \ref{sec:supp_planar_lattices}.
	
	\begin{figure}
		\centering
		\includegraphics[width=0.5\columnwidth]{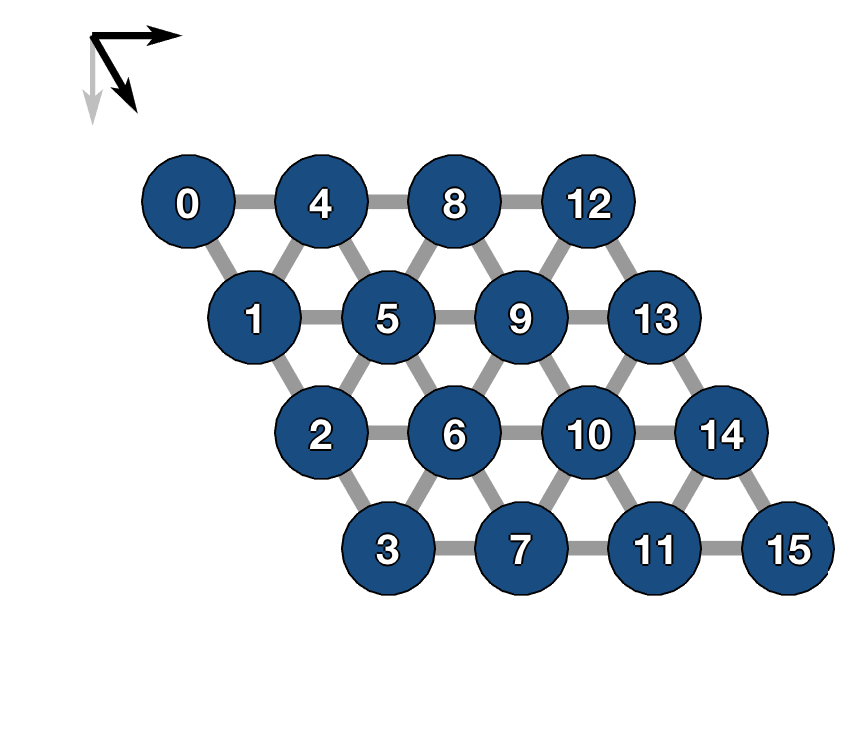}
		\caption{A $4 \times 4$ hexagonal lattice in which each cell is connected to its six nearest neighbors.  Cells are uniformly spaced along the directions indicated by the black arrows, which are separated by a 60$^\circ$ angle.  The gray arrow indicates a 90$^\circ$ angle from horizontal.} \label{fig:hex_lattice}
	\end{figure}

	\subsection{Analysis}
	
		\begin{figure}[!hbtp]
		\centering
		\includegraphics[width=0.75\columnwidth]{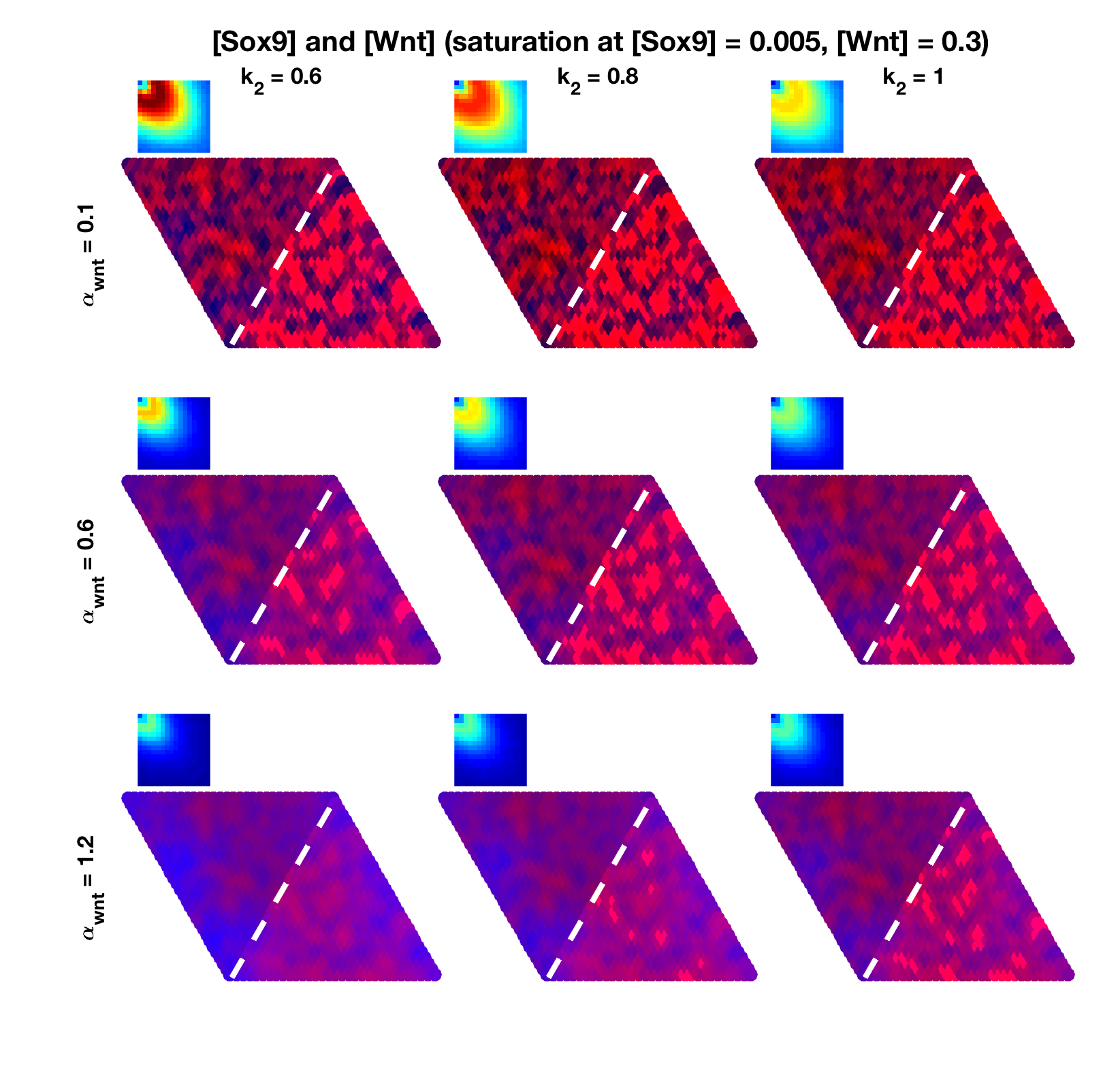}
		\caption{\textbf{Either increasing Wnt production or decreasing Bmp promotion of \textit{sox9} expression shrinks the size of contiguous high-[Sox9] regions, though filter analysis shows the two methods act through different mechanisms.}  Simulations are performed for nine $(\alpha_{wnt},k_2)$ pairs on a hexagonal lattice ($N_R = N_C = 32$) with two external inputs: a random background production rate for Sox9, Wnt, and Bmp; and an Fgf source localized to 5 columns of cells on the left.  The effect of the Fgf is visible as an increase in [Wnt] (blue) relative to [Sox9] (red) at both the left and right boundaries owing to the periodic boundary conditions.  The ``actual'' readouts, normalized across all images independently by channel, are pictured above the dotted white line; readout values below the dotted line have been post-processed to saturate at a threshold (0.005 for [Sox9], 0.3 for Wnt) and are normalized in the same fashion as (but separately from) the ``actual'' readouts.  For visual emphasis, saturated [Sox9] values are displayed at 10$\times$ the threshold intensity.  Inset heat maps display the magnitude of the filter coefficients around the bandpass (from $k_C = 0$ to $k_C = \frac{N_R}{2}$ and $k_R = 0$ to $k_R = \frac{N_C}{2}$) from input background production to readout [Sox9], each normalized to the same range (min. 0, max. 3.36).  Readouts that saturate above a certain threshold show more spotlike patterns for higher $\alpha_{wnt}$ or lower $k_2$, as observed in \cite{onimaru_et_al_2016}.  Increasing $\alpha_{wnt}$ decreases the overall amplitude of the filter and thereby shrinks the width of the passband, which suggests that spots rather than stripes may emerge when fewer cells express above a threshold.  In contrast, decreasing $k_2$ increases the maximum magnitude of the bandpass in a concentrated region, suggesting that spots may also be obtained by exaggerating differences in amplification between frequencies.  Parameters unless noted otherwise are as given in \cite{onimaru_et_al_2016}, Figures 4 and 5.}
		\label{fig:onimaru_fully_loaded}
	\end{figure}

	As in the original Sox9-Bmp-Wnt Turing model in \cite{raspopovic_et_al_2014}, the filter coefficients in the Fgf-augmented model form a bandpass at mid-range frequencies, resulting in the roughly periodic output patterning that alternates between Sox9 and Wnt (Figure \ref{fig:onimaru_fully_loaded}).  Increasing $\alpha_{wnt}$ decreases the magnitude of the bandpass, while decreasing $k_2$ concentrates amplification at a small range of frequencies inside the bandpass.  Either of these parameter changes tends to shrink contiguous regions of high [Sox9], consistent with the transition from stripes to spots observed in \cite{onimaru_et_al_2016}.  Parameters yielding more spotlike patterns also tend to suppress the influence of both correlated and uncorrelated noise for readout [Sox9], though the effect on readout [Wnt] is negligible (Supplementary Figures \ref{fig:supp_onimaru_h2_sox9_uncorr}, \ref{fig:supp_onimaru_h2_sox9_corr}, \ref{fig:supp_onimaru_h2_wnt_uncorr}, and \ref{fig:supp_onimaru_h2_wnt_corr}).  The distal edge where the source of Fgf is localized exhibits relatively higher Wnt than Sox9 expression, as observed \textit{in vivo} \cite{onimaru_et_al_2016}; in our normalized images, the effect is most visible at higher values of $\alpha_{wnt}$.
	
	 If we assume cells are approximately $12$ to $15$ $\mu$m in diameter \cite{milo_et_al_2010}, then for $\alpha_{wnt} = 1.2, k_2 = 1$ the filter and $\mathcal{H}_2$ norm analysis indicate that wavelengths of about 84 to 105 $\mu$m will be most strongly amplified in the result.  The prediction is in decent agreement with the experimental images in Figure 2 of \cite{onimaru_et_al_2016}, which exhibit periodicity on the order of 80 to 100 $\mu$m.  Some of the error may be accounted for by the difference in domain shape between filter simulations and actual limb paddles (rhomboidal vs. elliptical) as well as the presence of growth in the living animal.  Nevertheless, this observation suggests that the framework correctly identifies the range of spatial modes that will be most influential in forming the ``actual'' biological pattern.
	
	\section{CONCLUSIONS}
	\label{sec:conclusions}
	In this paper we have presented a framework to analyze how networks of interacting cells modify spatially varying inputs, either from environmental factors or intrinsic parameter variation, to produce patterned outputs.  Three biologically relevant examples indicate that qualitatively similar patterns may arise from different physical implementations (Section \ref{sec:notch-delta}), from both stable and unstable fixed points (Section \ref{sec:raspopovic}), as well as from variable filter behaviors when certain postprocessing steps are applied (Section \ref{sec:onimaru}).  Furthermore, these biological models appear robust to correlated and uncorrelated space-and-time-varying white noise inputs, a critical feature for maintaining consistency during embryonic development.
	
	We have demonstrated in a theoretical context how a filtering approach can offer insight into system behavior at an intermediate level between the exact physical implementation and the measured result.  In an experimental context, evaluating systems at the filter level may clarify when alterations to the input are capable of distinguishing between alternative explanations for an observed behavior.  For example, systems with near-identical filter coefficients are predicted to respond equivalently to inputs of all kinds (e.g., the three Notch-Delta models in Section \ref{sec:notch-delta}), suggesting that pure input-output probing is unlikely to illuminate the underlying mechanism.  Conversely, model systems with disparate filter coefficients may not vary much in their response to certain inputs but differ drastically in reponse to others, such that experiments in which inputs to the real system can be finely controlled may suffice to differentiate more accurate models from less accurate ones.  Of interest in both cases is the extent to which a particular system may impose structure upon an output pattern as compared to how much structure must be present in the prepattern.
	
	A critical assumption in our development of the filtering framework is that linearization about a homogeneous steady state is sufficient to capture relevant system behavior.  Future work should focus on incorporating nonlinear dynamics as well as investigating the influence of external inputs on spatially distributed, networked systems in the vicinity of unstable or nonhomogeneous steady states.  Additional areas for further research include patterning in time-varying or perturbed networks and system response to non-white noise inputs.  Lastly, although we have focused our applications on models in developmental biology, the generality of our framework suggests possible applications to synthetically engineered biological circuits as well.

	Overall, we believe a spatial frequency-based interpretation simplifies the process of predicting how intermolecular and intercellular interactions affect patterning mechanisms in living organisms.  It is our hope that the viewpoint developed here will help us to elucidate---and elaborate upon---nature's designs.

	\section*{ACKNOWLEDGMENTS}
	The authors would like to thank Andy Packard for providing incisive feedback on an earlier version of this work, and an anonymous reviewer for suggestions to strengthen the presentation of the manuscript.
	
	\clearpage
	\printbibliography

\begin{center}
\Huge{Supplementary Material}
\end{center}

\begin{figure}[!hbpt]
	\centering
	\includegraphics[width=0.75\columnwidth]{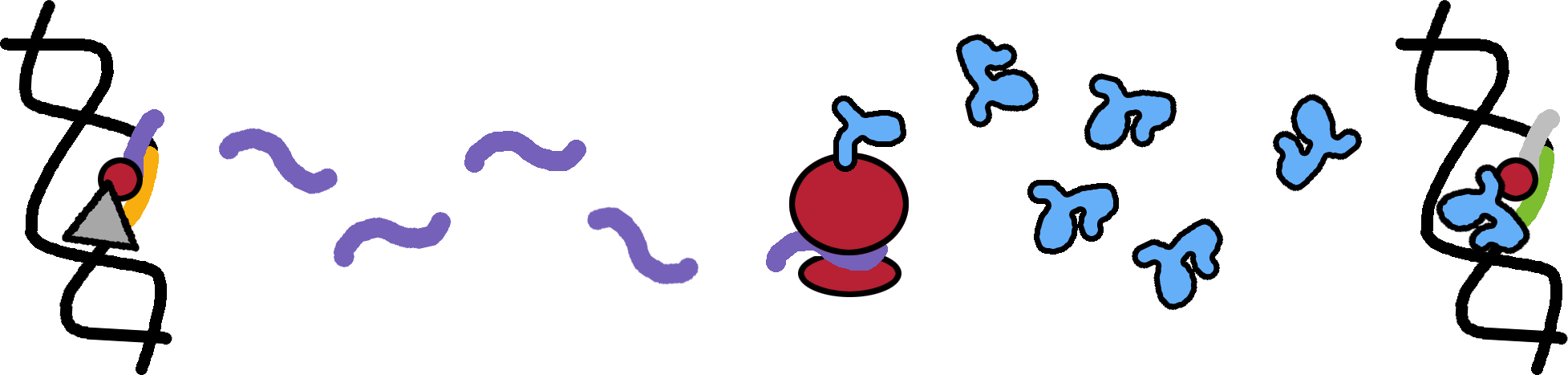}
	\caption{\textbf{A schematic of gene expression.}  An input (gray triangle) modifies the \textit{transcription rate} of mRNA (purple) from a gene (gold).  The mRNA is then \textit{translated} into protein (blue).  The proteins may in turn modify the transcription rate of some other target gene (green).  If the target gene is the same as the expressed gene, then the protein is said to be autoregulatory.  Chemical, mechanical, or electrical signals from neighboring cells may also influence transcription rates (not pictured).  Cell-to-cell variability in production rates that ``persists'' in time (i.e., is not due to the intrinsically stochastic nature of chemical interactions) may arise, for example, from variation in the concentrations of intercellular machinery (red circles) responsible for transcription and translation.} \label{fig:supp_gene_exp_schematic}
\end{figure}

\section{MORE ON FILTER COEFFICIENTS}

\subsection{Derivation of Filter Coefficients}
Here we provide a derivation of the filter coefficients from Proposition \ref{prop:filter_coeffs}.

Let $\underline{u} := \underline{\bar{u}}\mathds{1}_N$ be a spatially homogeneous input and assume $\exists~\underline{\bar{x}}^* \in \mathbb{R}^n$ such that $f\left(\underline{\bar{x}}^*,\mu g\left(\underline{\bar{x}}^*\right),\underline{\bar{u}}\right) = 0$.  Then $\underline{x}^* := \mathds{1}_N \otimes \underline{\bar{x}}^*$ is a homogeneous steady state.  The remaining steady-state quantities are similarly designated $\underline{y}^* = \underline{\bar{y}}^*\mathds{1}_N$, $\underline{\bar{w}}^* = g(\underline{\bar{x}}^*)$, and $\underline{v}^* = M_q\underline{w}^* = \mathds{1}_N \otimes \underline{\bar{v}}^*$ (where $M_q := M \otimes I_q$).  Let $\tilde{x}_i(t), \tilde{u}_i, \tilde{y}_i(t), \tilde{w}_i(t), \tilde{v}_i(t)$ denote perturbations about that steady state.  The full system linearized about $(\underline{x}^*,\underline{u})$ yields perturbed dynamics
\begin{equation}
\dot{\tilde{x}}(t) = \left[ \left( I_N \otimes A \right) + \left( I_N \otimes B_v \right)M_q\left( I_N \otimes G \right) \right]\tilde{x}(t) + \left( I_N \otimes B_u \right)\tilde{u} \label{eq:x_sys_perturbed} \
\end{equation}
where $A := \frac{\partial f}{\partial x_i} \big|_{(\underline{\bar{x}}^*,\underline{\bar{v}}^*,\underline{\bar{u}})}$, $B_v := \frac{\partial f}{\partial v_i} \big|_{(\underline{\bar{x}}^*,\underline{\bar{v}}^*,\underline{\bar{u}})}$,
$B_u := \frac{\partial f}{\partial u_i} \big|_{(\underline{\bar{x}}^*,\underline{\bar{v}}^*,\underline{\bar{u}})}$, $C := \frac{d h}{d y_i}\big|_{\underline{\bar{x}}^*}$, and $G := \frac{d g}{d x_i}\big|_{\underline{\bar{x}}^*}$ are the linearization matrices.

Assume that the interconnection matrix $M \in \mathbb{R}^{N \times N}$ is diagonalizable and let $M = T \Lambda T^{-1}$ be the diagonalization (so that $M_q$ is diagonalized by $T \otimes I_q$).  Define $\hat{x}(t) := \left(T^{-1} \otimes I_n \right)\tilde{x}(t)$, $\hat{u} := T^{-1}\tilde{u}$, and $\hat{y}(t) := T^{-1}\tilde{y}(t)$.  Recasting \eqref{eq:x_sys_perturbed} in the coordinate system $T$, we obtain the dynamical system
\begin{equation}
\dot{\hat{x}}(t) = \left[ \left( I_N \otimes A \right) + \Lambda \otimes (B_vG) \right]\hat{x}(t) + \left( I_N \otimes B_u \right)\hat{u}. \label{eq:xhat_sys_perturbed} \
\end{equation}
In contrast to the conditions for spontaneous pattern formation, we will not require this system to be unstable; large amplification of spatial modes is possible even when the system is stable.  The steady-state perturbed readout in basis $T$ is
\begin{align}
\hat{y}^* &:= -\left(I_N \otimes C\right)\left[ \left( I_N \otimes A \right) + \Lambda \otimes (B_vG) \right]^{-1}\left( I_N \otimes B_u \right)\hat{u} \nonumber \\
&=: S\hat{u}, \label{eq:ss_yhat_perturbed} \
\end{align}
where $S$ is a diagonal matrix with entries $${\left[S\right]_{kk} = -C(A+\lambda_k(M)B_vG)^{-1}B_u}$$ for $k = 0, 1, ..., N-1$ and $\left[S\right]_{kk}$ collectively form the ``filter coefficients" for the corresponding $N$ spatial modes.  The matrix $S$ is thus analogous to a digital filter that processes the input $\tilde{u}$ into readout $\tilde{y}^*$ with respect to the eigenvectors, or spatial modes, of $M$ as contained in $T$.

\begin{figure}[!hbpt]
	\centering
	\includegraphics[width=\columnwidth]{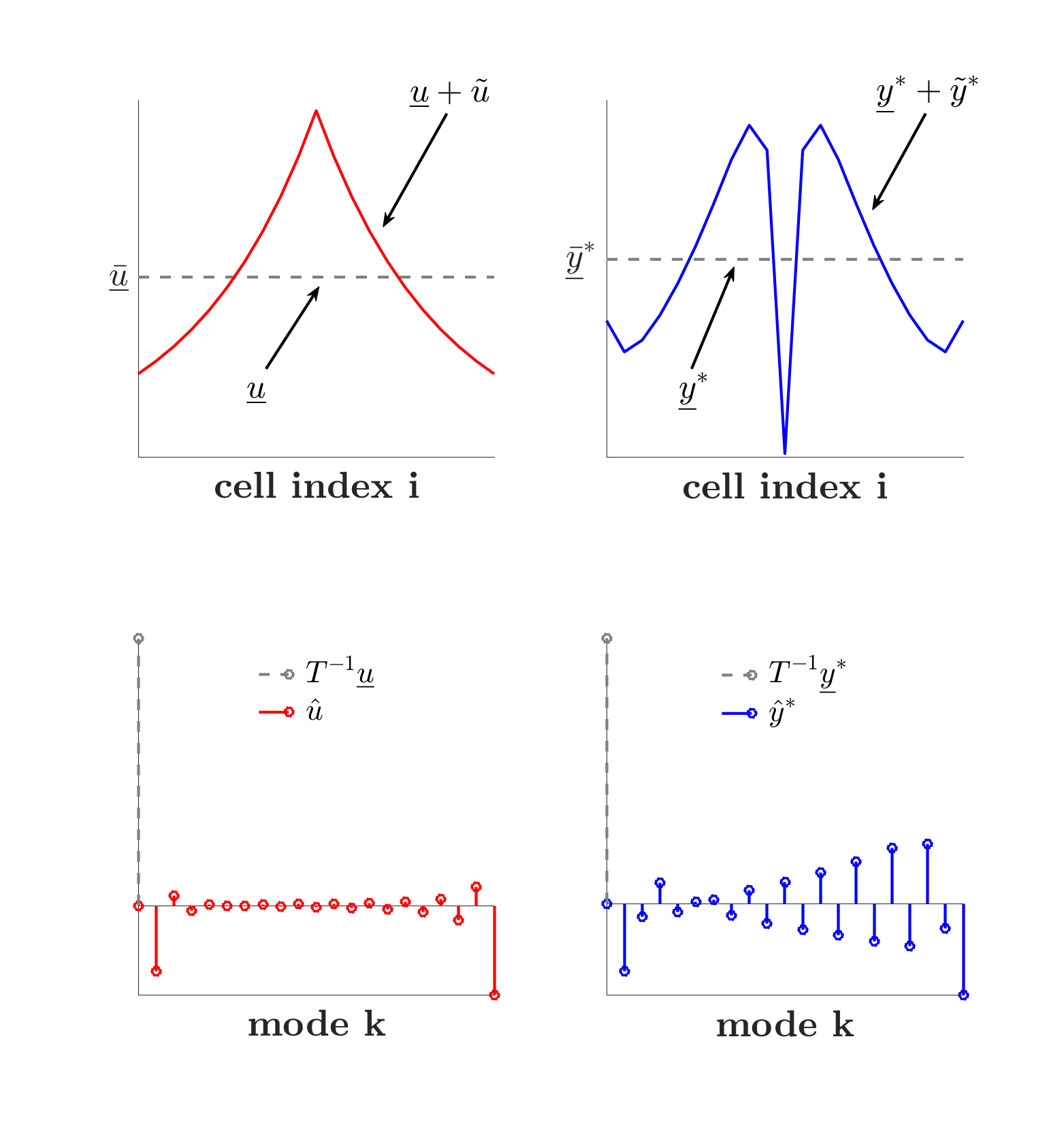}
	\caption{Schematic illustrating the variables utilized in the filter coefficient derivation.  Left, red corresponds to input-related variables; right, blue corresponds to readout-related variables.  Top, plots in the standard basis; the solid line for the input is the exact input to the full system, while the solid blue line is the approximated readout based on the filter coefficient analysis.  Bottom, plots in the basis of the spatial modes.  Dashed gray lines indicate values associated with the linearization; in this example, these constant-in-space inputs produce impulses at constant frequency (i.e., contribution to all other modes is zero---these values are not plotted).  This example uses periodic boundary conditions and sinusoidal modes given by \ref{obs:DFT_modes_real} with $k$ (from $0$ to $N-1$) indexing increasing frequency toward the middle of the x-axis, i.e., modes $k$ and $N-k$ have the same frequency.  The filter coefficients are symmetric about the midpoint but the mode representations of input and readout are not.  This asymmetry captures the ``location'' of the standard-basis input and readout relative to the cell indices $i$, since modes $k$ and $N-k$ have opposite phase (sign).} \label{fig:supp_sprinzak_labels}
\end{figure}

\subsection{Multiple Orthogonal Signals Sharing Same Spatial Modes}
Let $M_0, M_1,...,M_{q-1} \in \mathbb{R}^{N \times N}$ be the interconnection matrices for $q$ orthogonal signals and assume they all commute (share the same basis).  Let $\Delta_i \in \mathbb{R}^{q \times q}$ signify the matrix with $(i,i)$th entry one and all other entries zero, such that the full interconnectivity is
\begin{equation}
M = \sum_{i = 0}^{q-1} \left(M_i \otimes \Delta_i\right). \nonumber \
\end{equation}
If $\Lambda_i = T^{-1}M_iT$ for $i = 0, 1, ..., q-1$, then the basis $T \otimes I_q$ diagonalizes $M$ such that the filter coefficients are given by
\begin{equation}
\left[S\right]_{kk} = -C\left(A + \sum_{i = 0}^{q-1} \lambda_k\left(M_i\right)B_v\Delta_iG\right)B_u. \nonumber \
\end{equation}

\clearpage
\subsection{Derivation of Equivariance} \label{sec:symm_commute_proofs}
Here we derive the equivariance property of the input-output map comprising the filter coefficients.

Let $\tilde{y}^* = S\tilde{u}$ and $\tilde{y}_\Pi^* = S\Pi\tilde{u}$.  An equivalent statement to ``$S$ is equivariant'' is then $\tilde{y}_\Pi^* = \Pi S\tilde{u} \implies \tilde{y}_\Pi^* = \Pi\tilde{y}^*$.  To see this, take
\begin{align}
\tilde{y}_\Pi^* &= S\Pi\tilde{u} = -\left(I_N \otimes C\right)\left[ \left( I_N \otimes A \right) + M \otimes (B_vG) \right]^{-1}\left( I_N \otimes B_u \right)\Pi\tilde{u} \nonumber \\
&= -\left(I_N \otimes C\right)\left[ \left( I_N \otimes A \right) + M \otimes (B_vG) \right]^{-1}\left(\Pi \otimes I_{n}\right)\left( I_N \otimes B_u \right)\tilde{u} \nonumber \\
&= -\left(I_N \otimes C\right)\left[ \left(\Pi^{-1} \otimes I_{n}\right)\left(\left( I_N \otimes A \right) + M \otimes (B_vG)\right) \right]^{-1}\left( I_N \otimes B_u \right)\tilde{u} \nonumber \\
&= -\left(I_N \otimes C\right)\left[ \left(\Pi^{-1} \otimes I_{n}\right)\left( I_N \otimes A \right) + \Pi^{-1}M \otimes (B_vG) \right]^{-1}\left( I_N \otimes B_u \right)\tilde{u}. \label{eq:CPi-1APi-1M} \
\end{align}
Since
\begin{equation}
\left(\Pi^{-1} \otimes I_{n}\right)\left(I_N \otimes A\right) = \left(I_N \otimes A\right)\left(\Pi^{-1} \otimes I_{n}\right), \nonumber \
\end{equation}
then
\begin{equation}
M = \Pi M\Pi^{-1} \implies \Pi^{-1}M = M\Pi^{-1} \implies \Pi^{-1}M \otimes (B_vG) = M\Pi^{-1} \otimes (B_vG) \nonumber \
\end{equation}
such that \eqref{eq:CPi-1APi-1M} becomes
\begin{align}
\tilde{y}_\Pi^* &= -\left(I_N \otimes C\right)\left[\left(\left( I_N \otimes A \right) + M \otimes (B_vG) \right) \left(\Pi^{-1} \otimes I_{n}\right)\right]^{-1}\left( I_N \otimes B_u \right)\tilde{u} \nonumber \\
&= -\left(I_N \otimes C\right)\left(\Pi \otimes I_{n}\right)\left[\left( I_N \otimes A \right) + M \otimes (B_vG)\right]^{-1}\left( I_N \otimes B_u\right) \tilde{u} \nonumber \\
&= -\Pi\left(I_N \otimes C\right)\left[\left( I_N \otimes A \right) + M \otimes (B_vG)\right]^{-1}\left( I_N \otimes B_u\right) \tilde{u} \nonumber \\
&= \Pi\tilde{y}^*. \nonumber \
\end{align}

If $\Lambda = T^{-1}MT$, then since $M = \Pi M\Pi^{-1} \implies \Pi^{-1}M\Pi$, $\Lambda = T^{-1}\Pi^{-1}M\Pi T$ (i.e., $T$ diagonalizes the permuted version of $M$ with the same resultant eigenvalues).  Note that the immutability of $M$ under permutation $\Pi$ confers immutability of $\Lambda$ under permutation $T^{-1}\Pi T$, which is just the permutation in the basis of $M$; i.e.,
\begin{equation}
\Lambda = \left(T^{-1}\Pi T\right)\Lambda\left(T^{-1}\Pi T\right)^{-1}, \nonumber \
\end{equation}
or equivalently, $\Lambda$ and $\left(T^{-1}\Pi T\right)$ commute.

\clearpage
\section{TUTORIAL ON SPATIAL MODES} \label{sec:spatial_mode_tutorial}
Here we present a brief introduction to two sets of 1D spatial modes corresponding to common signal processing transforms.  These spatial modes have direct interpretations as spatial frequencies.  We then provide two useful observations for calculating the spatial modes and eigenvalues for 2D arrays with diagonal interconnections, as well as an interpretation of spatial frequencies for cells arranged in arbitrary planar lattices.

\subsection{Discrete Fourier Transform (DFT)}
If the $N$ cells form a ring indexed clockwise or counterclockwise, then $M$ is circulant.  The eigenvectors of a circulant matrix form the discrete Fourier basis such that the spatial modes of $T$ correspond exactly to the frequencies of sinusoids.

We can choose $T$ to be the discrete Fourier transform matrix (DFT) where the $j$th entry of the $k$th eigenvector, $j, k = 0, 1, ..., N-1$, is given by $$\left[T\right]_{jk} = \frac{1}{\sqrt{N}}e^{-\frac{2\pi ijk}{N}}$$ with $i := \sqrt{-1}$.  $T$ is conjugate symmetric.  If we let $m_0, ~ m_1, ~ ..., m_{N-1}$ denote the entries in the first row of $M$, then the eigenvalues of $M$ are given by $$\lambda_k\left(M\right) = \frac{1}{\sqrt{N}} \sum_{n=0}^{N-1} m_n e^{-\frac{2\pi ijk}{N}},$$
which corresponds to the coefficients of the discrete Fourier transform (DFT) of the first row of $M$.

If $M$ is symmetric in addition to circulant, then we can alternatively select the eigenvectors such that all entries are real.

\begin{observation} \label{obs:DFT_modes_real}
	Let $M \in \mathbb{R}^{N \times N}$ be a symmetric circulant matrix where $m_0 \in \mathbb{R}^N$ is the first row and define $m$ as the periodization of $m_0$.  Let the matrix $T$ have entries $$\left[T\right]_{jk} = \frac{1}{\sqrt{N}}\left(\cos\frac{2\pi jk}{N} + \sin\frac{2\pi jk}{N}\right).$$  Then $T$ is a basis for $M$ with eigenvalues $$\lambda_k\left(M\right) = \sum \limits_{n = 0}^{N-1} m_n \cos\frac{2\pi nk}{N},$$ $k = 0, 1, ..., N-1$.
\end{observation}

\begin{proof}
	Let $W$ be the unitary DFT matrix, i.e., the $j$th entry of the $k$th column is $$W_k^j = \frac{1}{\sqrt{N}}e^{\frac{i2\pi jk}{N}}.$$
	
	We can express $T$ as $$T = \frac{1}{2}\left[ W + W^H + e^{\frac{-i\pi}{2}}W + \left(e^{\frac{-i\pi}{2}}W\right)^H\right].$$  Then
	\begin{align}
	T^{-1}MT ={}& T^HMT = TMT \nonumber \\
	={}& \frac{1}{4}\left[ W + W^H + e^{\frac{-i\pi}{2}}W + \left(e^{\frac{-i\pi}{2}}W\right)^H\right]M\left[ W + W^H + e^{\frac{-i\pi}{2}}W + \left(e^{\frac{-i\pi}{2}}W\right)^H\right] \nonumber \\
	\begin{split} \nonumber
	={}& \frac{1}{4} \left(W + W^H\right)M\left(W + W^H\right) + \frac{1}{4}\left(e^{\frac{-i\pi}{2}}W + e^{\frac{i\pi}{2}}W^H\right)M\left(e^{\frac{-i\pi}{2}}W + e^{\frac{i\pi}{2}}W^H\right) + \\
	&\frac{1}{4}\left(W + W^H\right)M\left(e^{\frac{-i\pi}{2}}W + e^{\frac{i\pi}{2}}W^H\right) + \frac{1}{4}\left(e^{\frac{-i\pi}{2}}W + e^{\frac{i\pi}{2}}W^H\right)M\left(W + W^H\right)
	\end{split} \\
	\begin{split} \nonumber
	={}& \frac{1}{2}W^HMW + \frac{1}{2}WMW^H + \frac{1}{2}\cos\frac{\pi}{2}\left[W^HMW + WMW^H\right] + \\ &\frac{1}{2}e^{\frac{i\pi}{2}}W^HMW^H+\frac{1}{2}e^{\frac{-i\pi}{2}}WMW
	\end{split} \\
	={}& \frac{1}{2}W^HMW + \frac{1}{2}WMW^H + \frac{1}{2}e^{\frac{i\pi}{2}}W^HMW^H+\frac{1}{2}e^{\frac{-i\pi}{2}}WMW. \label{eq:WMWx4} \
	\end{align}
	Since $M$ is real and even (symmetric), the DFT is also real.  The symmetry of $M$ together with the symmetry of $W$ and the fact that diagonal matrices are symmetric also imply that $W^HMW = (W^HMW)^T = W^TM^TW^{*T} = WMW^H$, so we can somewhat simplify \eqref{eq:WMWx4} to
	\begin{align}
	{}& W^HMW + \frac{1}{2}\left[ e^{\frac{i\pi}{2}}W^HMW^H + e^{\frac{-i\pi}{2}}WMW\right] \nonumber \\
	={}& W^HMW + \frac{1}{2}\left[ e^{\frac{i\pi}{2}}\left(W^HMW\right)W^HW^H + e^{\frac{-i\pi}{2}}\left(WMW^H\right)WW\right] \nonumber \\
	={}& W^HMW + W^HMW\frac{1}{2}\left[ e^{\frac{i\pi}{2}}W^HW^H + e^{\frac{-i\pi}{2}}WW\right]. \label{eq:WMW-1/2} \
	\end{align}
	For $k = 1, 2, ..., \frac{N}{2}$ (N odd) or $k = 1, 2, ..., \frac{N-1}{2}$ (N even), the $(N - k)$th row or column of $W$ is equal to the $k$th row or column of $W^H$.  Therefore $\left[WW\right]_{k,N-k} = 1$.  Because of the orthogonality of complex exponentials, the remaining entries are $0$.  By the same logic we deduce an identical structure for $W^HW^H$ such that $WW = W^HW^H$.  Now \eqref{eq:WMW-1/2} becomes
	\begin{align}
	{}& W^HMW + W^HMW\frac{1}{2}\left[ e^{\frac{i\pi}{2}}WW + e^{\frac{-i\pi}{2}}WW\right] \nonumber \\
	={}& W^HMW + W^HMW\left(WW\right)\cos\frac{\pi}{2} \nonumber \\
	={}& W^HMW, \nonumber \
	\end{align}
	which is just $M$ diagonalized by the complex exponential DFT matrices, as desired.  From this we derive that the eigenvalues are the same as the DFT coefficients of $h_0$, which owing to symmetry may be calculated as
	\begin{equation}
	\lambda_m\left(M\right) = \sum \limits_{n = 0}^{N-1} m_n \cos\frac{2\pi nm}{N}. \nonumber \
	\end{equation}
\end{proof}

Owing to the periodicity of cosine, the eigenvalues of symmetric circulant $M$ (and hence the corresponding $\left[S\right]_{kk}$) are symmetric about the highest-frequency eigenvector associated with $k = \frac{N}{2}$.

A situation of particular interest occurs when $v_i$ is a diffusible molecule and the connection strength is equal between cells.  Then $M$ is a scaled version of the circulant finite differences (Laplacian) matrix
\begin{equation}
M = \left[ \begin{matrix}
-2 & 1 & 0 & \hdotsfor{1} & 0 & 1 \\
1 & -2 & 1 & \hdotsfor{1} & 0 & 0 \\
\vdots & \vdots & \vdots & \ddots & \vdots & \vdots \\
1 & 0 & 0 & \hdotsfor{1} & 1 & -2
\end{matrix} \right] \nonumber \
\end{equation}
with eigenvalues $$\lambda_k(M) = -2 + 2\cos{\frac{2\pi k}{N}}.$$  This form of $M$ corresponds to the second differences matrix for a system with periodic boundary conditions \cite{strang_1999}.

 \begin{figure}[!hbtp]
	\centering
	\includegraphics[width=\columnwidth]{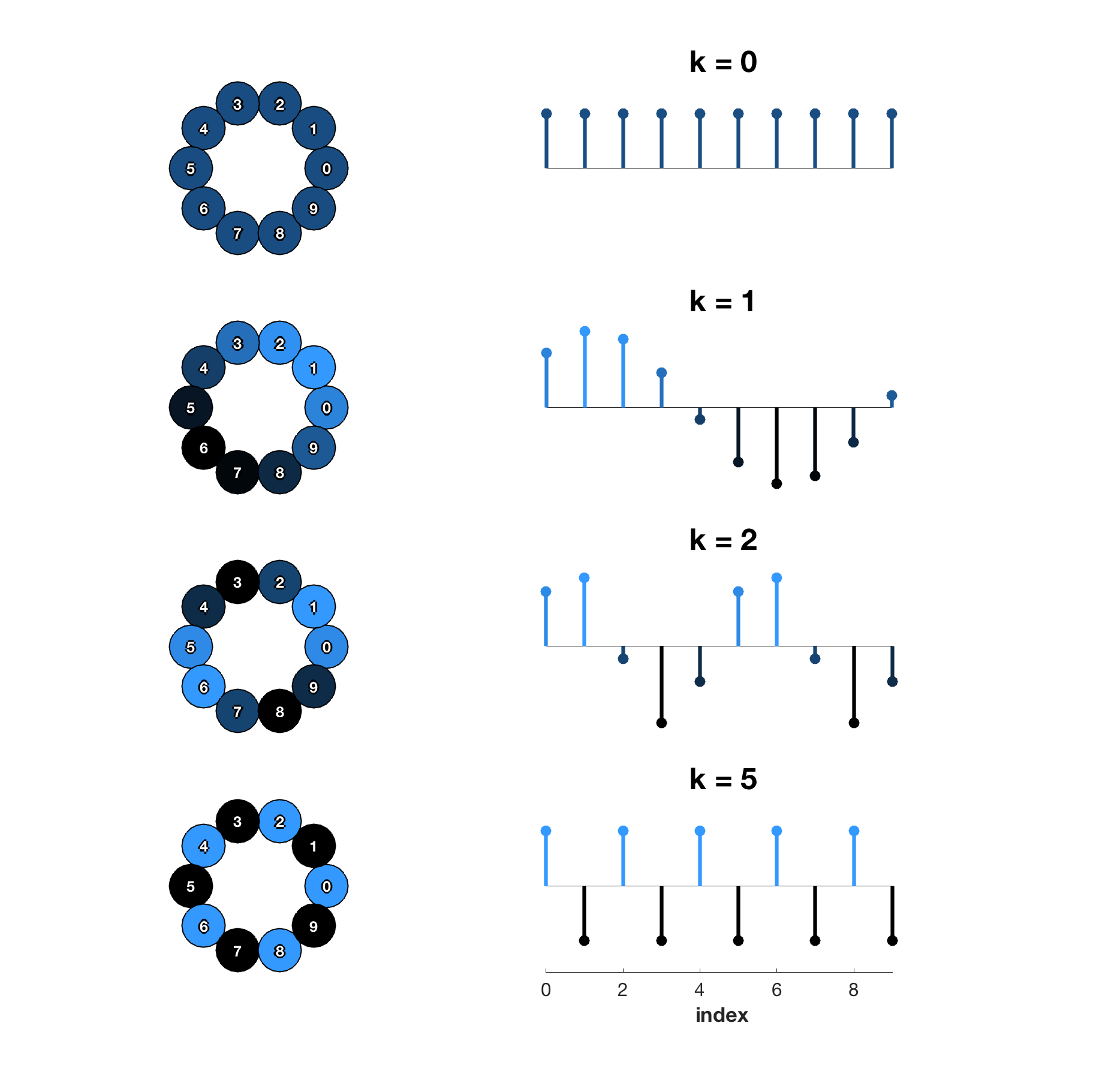}
	\caption{Sample spatial modes for a ring interconnectivity with $N = 10$ cells.  Each cell is connected to each of its two neighbors with equal connection strength.  $M$ is circulant, so the eigenvectors form the discrete Fourier basis such that the spatial modes of $T$ correspond exactly to the frequencies of sinusoids.  The value $k$ corresponds to the spatial frequency, or the number of complete periods present in a single cycle around the ring.  Because the basis is the discrete Fourier transform (DFT) and $N$ is even, the highest frequency is $\frac{N}{2} = 5$.  Neighboring cells in this mode alternate between two values.  Such a pattern is not possible in a ring configuration when $N$ is odd.}
	\label{fig:ring_dft_modes}
\end{figure}

\begin{figure}[!hbtp]
	\centering
	\includegraphics[width=\columnwidth]{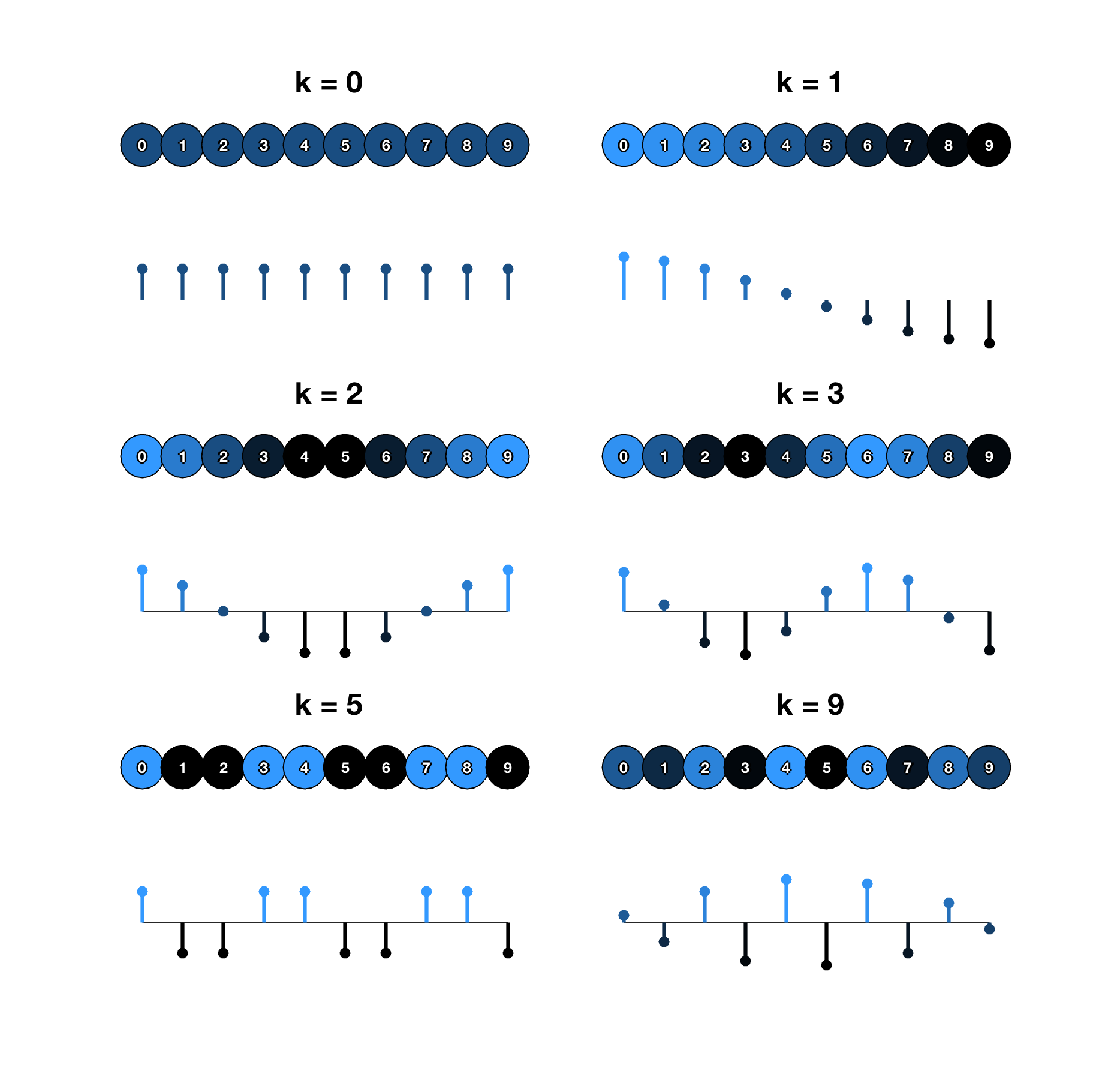}
	\caption{Sample spatial modes for a line interconnectivity with $N = 10$ cells.  The two cells on the end each interact with only one neighbor such that the basis vectors are the DCT-2 vectors.  The value $k$ corresponds to twice the frequency of its corresponding mode, i.e., $\frac{k}{2}$ periods are represented in mode $k$.  ``Even'' modes ($k$ even) have symmetry about the midpoint between cells $\frac{N}{2}-1$ and $\frac{N}{2}+1$, while ``odd'' modes ($k$ odd) are antisymmetric about this same point.  If $N$ were odd, the midpoint would instead be the $\left(\frac{N-1}{2}\right)$th cell.}
	\label{fig:line_dct2_modes}
\end{figure}

\begin{figure}[!hbtp]
	\centering
	\includegraphics[width=\columnwidth]{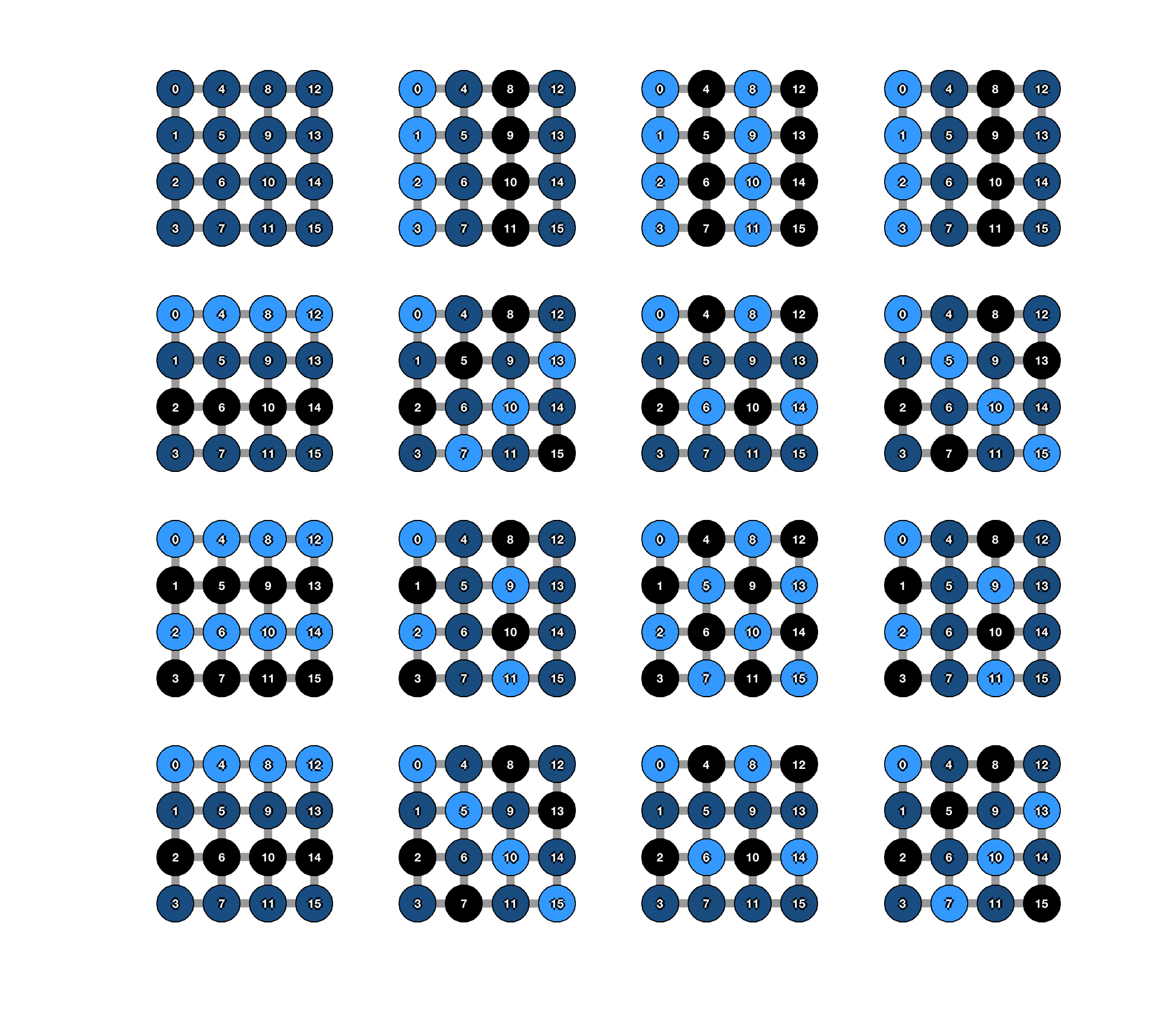}
	\caption{A complete set of spatial modes for a 2D periodic boundary interconnectivity (DFT basis) on a $4 \times 4$ rectangular array.  Due to the symmetry in the eigenvectors, modes $(m,n)$ and $(m,4-n)$ are identical, as are $(m,n)$ and $(4-m,n)$.}
	\label{fig:rect_modes_dft}
\end{figure}

\subsection{Second Discrete Cosine Transform (DCT-2)}
If the $N$ cells are organized in a line, then the two cells on the end each communicate with only one neighbor.  If the mode of communication is a diffusible molecule and the cells are indexed from one end of the line to the other, then the connectivity takes the form of a second differences matrix with Neumann boundary conditions centered at the midpoint:
\begin{equation}
M = \left[ \begin{matrix}
-1 & 1 & 0 & \hdotsfor{1} & 0 & 0 & 0 \\
1 & -2 & 1 & \hdotsfor{1} & 0 & 0 & 0 \\
\vdots & \vdots & \vdots & \ddots & \vdots & \vdots & \vdots \\
0 & 0 & 0 & \hdotsfor{1} & 1 & -2 & 1 \\
0 & 0 & 0 & \hdotsfor{1} & 0 & 1 & -1
\end{matrix} \right]. \nonumber \
\end{equation}
The spatial modes of $T$ form the basis for the second discrete cosine transform (DCT-2).  The $j$th entry of the $k$th eigenvector, $j, k = 0, 1, ..., N-1$, is given by $$\left[T\right]_{jk} = \sqrt{\frac{2}{N}}\cos\left[\left(j+\frac{1}{2}\right)\frac{k\pi}{N}\right]$$ (for $k = 0$, divide by additional factor of $\sqrt{2}$) with corresponding eigenvalue $$\lambda_k(M) = -2 + 2\cos\frac{k\pi}{N}.$$  The highest frequency is $k = N-1$ \cite{strang_1999}.  Unlike the case of circulant $M$, there are no guarantees of symmetry in the filter $S$ and $T$ itself is not conjugate symmetric.

\clearpage
\subsection{Eigenvectors and Eigenvalues for Arrays with Diagonal Interconnections} \label{sec:dft_diag_2d_eigs}
\begin{observation} \label{obs:dft_diag_2d_eigvecs}
	Consider an $N_R \times N_C$ array of cells.  Let $M_F$ describe the interconnectivity of $C_D := min(N_R,N_C)$ elements in the forward diagonal direction and $M_B$ the interconnectivity of $C_D$ elements in the backward diagonal direction.  Let $P \in \mathbb{R}^{C_D \times C_D}$ be the permutation matrix with lower diagonal ones and the last entry of the first column also one, such that $P_F := P_B^T := diag\left(P^0, P^1, P^2, ..., P^{R_D} \right)$ where $R_D := max(N_R,N_C)$.  In total, the interconnectivity of an array with horizontal, vertical, and diagonal components is described by
	\begin{equation}
	M := \left( M_R \otimes I_{N_R} \right) + \left( I_C \otimes M_C \right) + P_F^T\left( M_F \otimes I_{R_D} \right)P_F + P_B^T\left( M_B \otimes I_{R_D} \right)P_B \nonumber \
	\end{equation}
	where $M_R$, $M_C$, $M_F$, and $M_B$ are circulant.  Furthermore, if $N_R = N_C = N_0$ and the real or complex DFT basis $T_0$ diagonalizes each of $M_R$, $M_C$, $M_F$, and $M_B$ individually, then $\left(T_0 \otimes T_0\right)$ diagonalizes $M$.
\end{observation}

\begin{proof}
	Let $T_0$ diagonalize $M_0$ such that $\left(T_0 \otimes T_0\right)$ diagonalizes $\left(M_R \otimes I_{N_R}\right) + \left(I_C \otimes M_C\right)$.  Since $M_0$ is circulant,
	\begin{equation}
	M_0 = \left[ \begin{matrix}
	m_0 & m_{N_R-1} & m_{N_R-2} & \dots & m_1 \\
	m_{1} & m_0 & m_{N_R-1} & \dots & m_2 \\
	\vdots & \vdots & \vdots & \ddots & \vdots \\
	m_{N_R-1} & m_{N_R-2} & m_{N_R-3} & \dots & m_0
	\end{matrix} \right] = m_0P^0 + m_1P + m_2P^2 + ... + m_{N_R-1}P^{N_R-1}\nonumber \
	\end{equation}
	and therefore
	\begin{align}
	P_F^{-1}\left( M_0 \otimes I_{N_R} \right)P_F &= P_F^{-1}\left( \sum \limits_{q = 0}^{N_R-1} m_q P^q \otimes I_{N_R} \right)P_F \\
	&= \sum \limits_{q = 0}^{N_R-1} m_q P_F^{-1}\left(P^q \otimes I_{N_R}\right)P_F.
	\end{align}
	If we expand $P_F$ and use the fact that $P^{-1} = P^{N_R-1} = P^T$, then the summation simplifies to
	\begin{align}
	\sum \limits_{q = 0}^{N_R-1} m_q \left(P^q \otimes P^{-q}\right) \label{eq:sum_mq_Pq_P-q} \
	\end{align}
	which is diagonalized by the complex exponential DFT vectors $T_0^{-1}$ and $T_0$ as follows:
	\begin{align}
	\left( T_0^{-1} \otimes T_0^{-1} \right) \left[ \sum \limits_{q = 0}^{N_R-1} m_q \left(P^q \otimes P^{-q}\right) \right] \left( T_0 \otimes T_0 \right) &= \sum \limits_{q = 0}^{N_R-1} m_q \left( T_0^{-1} \otimes T_0^{-1} \right) \left(P^q \otimes P^{-q}\right) \left( T_0 \otimes T_0 \right) \nonumber \\
	&= \sum \limits_{q = 0}^{N_R-1} m_q \left( T_0^{-1} P^q T_0 \right) \otimes \left( T_0^{-1} \otimes P^{-q} T_0 \right), \label{eq:sum_mq_TPqT_TP-qT} \
	\end{align}
	which is a sum of diagonal matrices because permutation matrices are circulant and therefore diagonalized by the DFT matrices $T_0$, and the Kronecker product of two diagonal matrices is diagonal.  Since $P_B = P_F^T = P_F^{-1}$, the same derivation for diagonal connectivity in the backward direction gives
	\begin{equation}
	P_B^{-1}\left(M_0 \otimes I_{N_R}\right)P_B = \sum \limits_{q = 0}^{N_R-1} m_q \left(P^q \otimes P^q\right), \label{eq:sum_mq_Pq_Pq} \
	\end{equation}
	which is also diagonalized by DFT matrices as
	\begin{equation}
	\sum \limits_{q = 0}^{N_R-1} m_q \left( T_0^{-1} P^q T_0 \right) \otimes \left( T_0^{-1} \otimes P^{q} T_0 \right) \nonumber \
	\end{equation}
	where each summand is diagonal.  This implies that the complex exponential form of the DFT is the basis for an array that is diagonally connected in either or both directions.
	
	If $M_0$ is symmetric, then in addition to the complex exponential basis $T_0$ we might also choose the basis $T'_0$ with the $j$th entry of $k$th eigenvector given by
	\begin{equation}
	T'_0(k) = \frac{1}{\sqrt{R}}\left(\cos\frac{2\pi jk}{R} + \sin\frac{2\pi jk}{R}\right), \label{eq:real_DFT_eigenvec} \
	\end{equation}
	in which case each term in the summation \eqref{eq:sum_mq_TPqT_TP-qT} is no longer diagonal because $T'_0$ does not diagonalize the nonsymmetric permutation matrices.  Hence for the real-valued basis $T'_0$ we require that the array is diagonally connected in both forward and backward directions such that the full connectivity matrix is given by the sum of \eqref{eq:sum_mq_Pq_Pq} and \eqref{eq:sum_mq_Pq_P-q}:
	\begin{equation}
	\sum \limits_{q = 0}^{N_R-1} m_q \left(P^q \otimes P^{-q}\right) + \sum \limits_{q = 0}^{N_R-1} m_q \left(P^q \otimes P^q\right) = \sum \limits_{q = 0}^{N_R-1} m_q P_q \otimes \left( P^{-q} + P^q \right). \label{eq:sum_mq_Pq_P-q+Pq} \
	\end{equation}
	The matrix $\left( P^{-q} + P^q \right)$ is circulant and symmetric and hence diagonalized by $T'_0$.  To complete the argument we appeal to the symmetry of $m_q$.  Specifically, if $N_R$ is odd,
	\begin{align}
	\sum \limits_{q = 0}^{N_R-1} m_q &P^q \otimes \left( P^{-q} + P^q \right) \nonumber \\
	&= 2m_0I_{N_R} + \sum \limits_{q = 1}^{\frac{N_R-1}{2}} m_q P_q \otimes \left( P^{-q} + P^q \right) + \sum \limits_{q = \frac{N_R-1}{2} + 1}^{N_R-1} m_q P^q \otimes \left( P^{-q} + P^q \right) \nonumber \\
	&= 2m_0I_{N_R} + \sum \limits_{q = 1}^{\frac{N_R-1}{2}} m_q P^q \otimes \left( P^{-q} + P^q \right) + m_{R-q} P^{R-q} \otimes \left( P^{q-R} + P^{R-q} \right) \nonumber \\
	&= 2m_0I_{N_R} + \sum \limits_{q = 1}^{\frac{N_R-1}{2}} m_q P^q \otimes \left( P^{-q} + P^q \right) + m_q P^{-q} \otimes \left( P^q + P^{-q} \right) \nonumber \\
	&= 2m_0I_{N_R} + \sum \limits_{q = 1}^{\frac{N_R-1}{2}} m_q \left(P^q + P^{-q}\right) \otimes \left( P^{-q} + P^q \right), \nonumber \
	\end{align}
	where each individual term is diagonalized by $T'_0$, and hence the whole summation is diagonal.  If $N_R$ is even, the summation \eqref{eq:sum_mq_Pq_P-q+Pq} breaks into
	\begin{equation}
	2m_0I_{N_R} + m_{\frac{N_R}{2}}P^{\frac{N_R}{2}} \otimes \left(P^{\frac{N_R}{2}} + P^{-\frac{N_R}{2}} \right) + \sum \limits_{q = 1}^{\frac{N_R-1}{2}} m_q \left(P^q + P^{-q}\right) \otimes \left( P^{-q} + P^q \right). \nonumber \
	\end{equation}
	When $N_R$ is even, $P^{R/2}$ alone is circulant symmetric, hence the additional term is also diagonalized by $T'_0$.  Therefore when $M_0$ is circulant symmetric, the matrix
	\begin{equation}
	M = P_F^T\left( M_0 \otimes I_{N_R} \right)P_F + P_B^T\left( M_0 \otimes I_{N_R} \right)P_B \nonumber \
	\end{equation}
	is diagonalized by $T':= T'_0 \otimes T'_0$.  Note that this implies
	\begin{equation}
	\left( M_0 \otimes I_{N_R} \right) + \left( I_{N_R} \otimes M_0 \right) + P_F^T\left( M_0 \otimes I_{N_R} \right)P_F + P_B^T\left( M_0 \otimes I_{N_R} \right)P_B \nonumber \
	\end{equation}
	is also diagonalized by $T'$.
\end{proof}

\begin{siamremark}
	For a system with diagonal connections only ($M_R = M_C = 0$), then if $N_R = N_C$ odd, the diagonal transformations are identical to a 2D DFT rotated 45$^{\circ}$.  For $N_R = N_C$ even, the array becomes divided into two separate classes that are transformed separately; i.e., the underlying network graph is no longer connected.  This is because for $N_R = N_C$ odd, the array has a compartment at the center, while for $N_R = N_C$ even, the center would (in physical space) represent a crossing of intersections.
\end{siamremark}

\begin{observation} \label{obs:dft_diag_2d_eigvals}
	Let $N_C = N_R = N_0$, $M_0$ circulant with complex exponential basis vectors $T_0$ and consider the full forward diagonal interconnection matrix $M = P_F^T\left(M_0 \otimes I_{N_0}\right)P_F$.  The $(m,n)$th eigenvalue of $M$ is
	\begin{equation}
	\lambda_{m + nN_0}\left(M\right) =
	\begin{cases}
	m_0 + 2 \sum\limits_{q = 0}^{\frac{N_0-1}{2}}m_q\cos\frac{2\pi q\left(n-m\right)}{N_0}, ~ N_0 ~ \text{odd}, \\
	m_0 + m_{\frac{N_0}{2}}(-1)^{n-1} + 2 \sum\limits_{q = 0}^{\frac{N_0}{2}-1}m_q\cos\frac{2\pi q\left(n-m\right)}{N_0}, ~ N_0 ~ \text{even},
	\end{cases} \nonumber \
	\end{equation}
	where $m_k$ is the $k$th entry of the first row or column of $M_0$.
\end{observation}

\begin{proof}
	For $N_C = N_R = N_0$ and $M_F$ as defined above, the $k$th entry of the $N_0$-point DFT of the first row of $P^q$ is $e^{\frac{2\pi ikq}{N_0}}$, $k = 0, 1, ..., N_0$ and the diagonalization $T_0^{-1}P^qT_0$ is the matrix with the DFT entries on the diagonal.  This implies that we can write \eqref{eq:sum_mq_TPqT_TP-qT} as
	\begin{align}
	&\left(T_0^{-1} \otimes T_0^{-1}\right)P_F^T\left(I_{N_0} \otimes M_0\right)P_F\left(T_0 \otimes T_0\right) = \sum \limits_{q = 0}^{N_0 - 1} m_q\left(T_0^{-1}P^qT_0\right) \otimes \left(T_0^{-1}P^{-q}T_0\right) \nonumber \\
	&= \sum \limits_{q = 0}^{N_0 - 1} m_q \left[\begin{matrix}
	1 & & & \\
	& e^{\frac{2\pi iq}{N_0}} & & & \\
	& & e^{\frac{4\pi iq}{N_0}} & & \\
	& & & \ddots & \\
	& & & & e^{\frac{2\pi iq(N_0-1)}{N_0}}
	\end{matrix}\right] \otimes \left[\begin{matrix}
	1 & & & \\
	& e^{\frac{-2\pi iq}{N_0}} & & & \\
	& & e^{\frac{-4\pi iq}{N_0}} & & \\
	& & & \ddots & \\
	& & & & e^{\frac{-2\pi iq(N_0-1)}{N_0}}
	\end{matrix}\right] \nonumber \\
	&=: \sum \limits_{q = 0}^{N_0 - 1} m_q \left[\begin{matrix}
	I_N & & & \\
	& E_1^q & & & \\
	& & E_2^q & & \\
	& & & \ddots & \\
	& & & & E_{N_0-1}^q
	\end{matrix}\right], \label{eq:mqEEE} \
	\end{align}
	where we have defined
	\begin{equation}
	E_k^q := \left[\begin{matrix}
	e^{\frac{2\pi iqk}{N_0}} & & & \\
	& e^{\frac{2\pi iq(k-1)}{N_0}} & & & \\
	& & e^{\frac{2\pi iq(k-2)}{N_0}} & & \\
	& & & \ddots & \\
	& & & & e^{\frac{2\pi iq(k-(N_0-1))}{N_0}}
	\end{matrix}\right]. \nonumber \
	\end{equation}
	For $N_0$ odd, we use the symmetry $m_q = m_{-q}$ to rewrite the summation \eqref{eq:mqEEE} as
	\begin{equation}
	m_0 I_{N_0^2} + \sum \limits_{q = 1}^{\frac{N_0 - 1}{2}} m_q \left[\begin{matrix}
	2I_N & & & \\
	& E_1^q + E_1^{-q} & & & \\
	& & E_2^q + E_2^{-q} & & \\
	& & & \ddots & \\
	& & & & E_{N_0-1}^q + E_{N_0-1}^{-q}
	\end{matrix}\right]. \nonumber \
	\end{equation}
	Conveniently,
	\begin{equation}
	E_k^q + E_k^{-q} = \left[\begin{matrix}
	2\cos{\frac{2\pi iqk}{N_0}} & & & \\
	& 2\cos{\frac{2\pi iq(k-1)}{N_0}} & & & \\
	& & 2\cos{\frac{2\pi iq(k-2)}{N_0}} & & \\
	& & & \ddots & \\
	& & & & 2\cos{\frac{2\pi iq(k-(N_0-1))}{N_0}}
	\end{matrix}\right], \nonumber \
	\end{equation}
	from which we infer
	\begin{equation}
	\lambda_{m + nN_0}\left(P_F^{-1}\left(M_0 \otimes I_{N_0}\right)P_F\right) = m_0 + 2 \sum\limits_{q = 0}^{\frac{N_0-1}{2}}m_q\cos\frac{2\pi q\left(n-m\right)}{N_0} \nonumber \
	\end{equation}
	is the eigenvalue for the $\left(m,n\right)$th spatial mode owing to diagonal connectivity, $N_0$ odd.  If $N_0$ is even, we write \eqref{eq:mqEEE} as
	\begin{align}
	m_0 I_{N_0^2} +
	&m_{\frac{N_0}{2}}\left[\begin{matrix}
	I_N & & & \\
	& E_1^{\frac{N_0}{2}} & & \\
	& & \ddots & \\
	& & & E_{N_0-1}^{\frac{N_0}{2}}
	\end{matrix}\right] + \nonumber \\
	&\sum \limits_{q = 1}^{\frac{N_0}{2}-1} m_q \left[\begin{matrix}
	2I_N & & & \\
	& E_1^q + E_1^{-q} & & \\
	& & \ddots & \\
	& & & E_{N_0-1}^q + E_{N_0-1}^{-q}
	\end{matrix}\right], \nonumber \\
	\end{align}
	and note that
	\begin{equation}
	E_k^{\frac{N_0}{2}} = \left[\begin{matrix}
	(-1)^k & & & \\
	& (-1)^{k-1} & & \\
	& & \ddots & \\
	& & & (-1)^{k-(N_0-1)}
	\end{matrix}\right] \nonumber \
	\end{equation}
	to write the $(m,n)$th eigenvalue as
	\begin{equation}
	\lambda_{m + nN_0}\left(M_F\right) = m_0 + m_{\frac{N_0}{2}}(-1)^{n-1} + 2 \sum\limits_{q = 0}^{\frac{N_0}{2}-1}m_q\cos\frac{2\pi q\left(n-m\right)}{N_0}. \nonumber \
	\end{equation}
\end{proof}

\subsection{Spatial Modes on Planar Lattices} \label{sec:supp_planar_lattices}
It is straightforward to generalize a frequency-based interpretation to cells arranged in periodic planar lattices, which are well described mathematically.  Throughout the following discussion we will refer to coordinates in physical space as $\hat{e}$, the unit vector pointing ``east,'' and $\hat{s}$, the unit vector pointing ``south''.  This choice of vector orientations mimics the numbering scheme in an array, whereby indices increase horizontally left to right (with $\hat{e}$) and vertically top to bottom (with $\hat{s}$).  We will assume a system of cells indexed in an $N_R \times N_C$ array with periodic boundary conditions such that $M$ is diagonalized by $T_R \otimes T_C$ where both $T_C$ and $T_R$ are real or complex DFT bases of appropriate dimension.

Let cells in physical space be arranged in a planar lattice described by vectors $a_R$ and $a_C$ corresponding respectively to the rows and columns of the indexed array.  Without loss of generality we orient $a_R$ along $\hat{e}$ (such that $a_R \cdot \hat{e} = |a_R|$).  We define the unit vectors $\hat{a}_R := \frac{a_R}{|a_R|} = \hat{e}$ and $\hat{a}_C := \frac{a_C}{|a_C|}$.  Letting $\theta$ be the angle between $a_R$ and $a_C$, we can write $a_R = |a_R|\hat{e}$ and $a_C = |a_C|\cos\theta \hat{e} + |a_C| \sin\theta \hat{s}$.  Note that the eigenfunctions are periodic in $n$ with period $N_C|a_R|$ along $a_R$ and periodic in $m$ with period $N_R|a_C|$ along $a_C$.

\begin{observation}
	For an $N_R \otimes N_C$ cellular lattice with lattice vectors $a_R$, $a_C$ and periodic boundary conditions, the $(m,n)$th spatial mode corresponds to a plane wave of frequency
	\begin{equation}
	\frac{n}{|a_R|N_C}\hat{a}_R + \frac{m}{|a_C|N_R}\hat{a}_C =: 
	f_{a_R}\hat{a}_R + f_{a_C}\hat{a}_C = \left(f_{a_R} + f_{a_C}\cos\theta\right)\hat{e} + f_{a_C}\sin\theta\hat{s} \nonumber \
	\end{equation}
	in physical space, with an ``absolute'' frequency of
	\begin{equation}
	f = \sqrt{\left|f_{a_R} + f_{a_C}\cos\theta\right|^2 + \left|f_{a_C}\sin\theta\right|^2} \nonumber \
	\end{equation}
	pointing at an angle
	\begin{equation}
	\phi = \tan^{-1}\left(\frac{f_{a_C}\sin\theta}{f_{a_R}+f_{a_C}\cos\theta}\right) \nonumber \
	\end{equation}
	from the $\hat{x}$ axis.
\end{observation}

\begin{siamremark}
	One may liken the translation from physical space into matrix space to ``sampling'' in space from an underlying pattern with ``spatial sampling frequency'' $\frac{1}{|a_R|}$ along $a_R$ and $\frac{1}{|a_C|}$ along $a_C$.  The translation into spatial modes, or the DFT, recovers normalized frequency components from the discrete samples.  Maintaining a constant surface area but increasing the number of cells occupying that surface area (i.e., $N_R' = c_RN_R$, $N_C' = c_CN_C$, $a_R' = \frac{1}{C_R}a_R$, $a_C' = \frac{1}{c_C}a_C$) does not change the physical range of space over which the modes are described but does increase the ``resolution'' or ``sampling rate'' of the system by a factor of $c_R$ along $a_R$ and $c_C$ along $a_C$, enabling the system to modify higher frequencies than before and therefore permitting finer filtering of a continuous-in-space input gradient.
\end{siamremark}

\begin{figure}[!hbpt]
	\centering
	\includegraphics[width=\columnwidth]{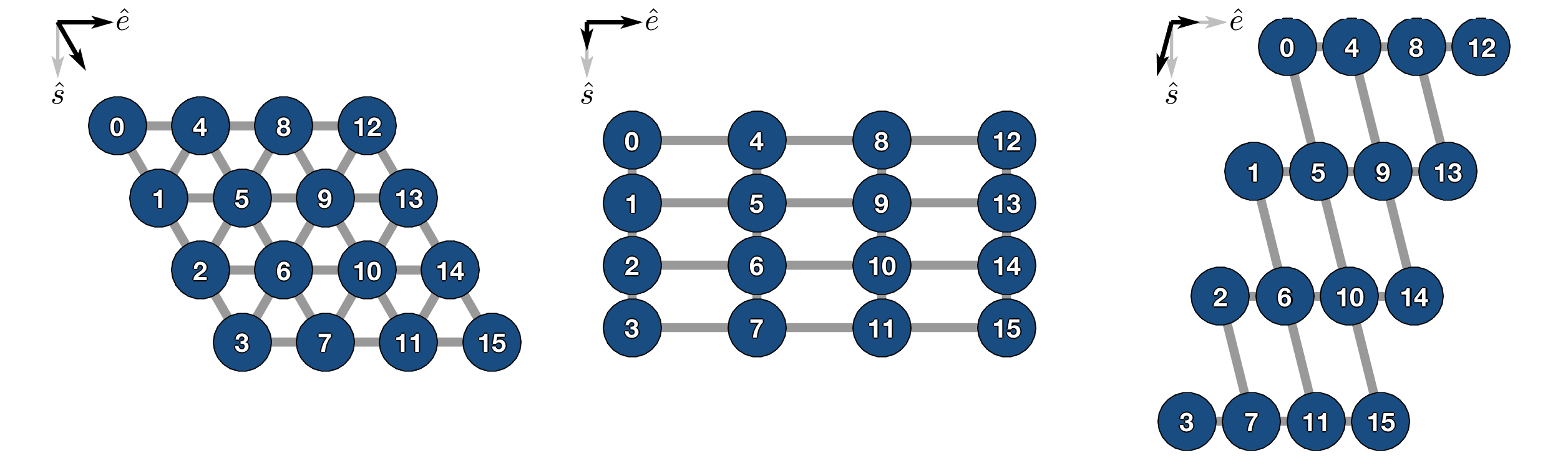}
	\caption{Examples of lattice configurations.  Left, hexagonal ($\theta = 60^{\circ}$) with nearest-neighbor (row/column/forward diagonal) interconnectivity; center, rectangular ($\theta = 90^{\circ}$) with $|a_1| = 2|a_2|$ and row/column interconnectivity; right, an arbitrary lattice with $\theta = 105^{\circ}$, $|a_2| = 2|a_1|$, and row/backward diagonal interconnectivity.}
\end{figure}

\clearpage
\section{NOTCH-DELTA MODELS} \label{sec:D-N_model_calcs}

	\begin{table}[!bp]
	\centering
	\scalebox{0.75}{
	\begin{tabular}{| c | c | c | c |} \hline
		\bf Parameter & \bf Value & \bf Description & Source \\ \hline
		$\alpha_N$ & 10 & ``leakiness'' of Notch expression (RFU/hr) & \cite{sprinzak_et_al_2011} Table S1 (Figure S4A) \\
		$\beta_{D0}$ & 17.5 & max. Delta production rate (RFU/hr) & \cite{sprinzak_et_al_2010} Table S3 (Figure 4C) \\
		$x_0$ & 7 & number of cell diameters & \cite{sprinzak_et_al_2010} Table S3 (Figure 4C) \\
		$\beta_{D_i}$ & $\beta_{D0}e^{-|i|/x_0}$ & Delta production rate (RFU/hr) for cell $i$ & \cite{sprinzak_et_al_2010} (Figure 4C) \\
		$\underline{\bar{\beta}}_D$ & 9.09 & Delta production rate (RFU/hr) for linearization &  $\frac{1}{N}\sum_{i = 0}^{N-1} \beta_{D_i}$ \\
		$\beta_N$ & 10 & Notch production rate (RFU/hr) & \cite{sprinzak_et_al_2010} Table S3 (Figure 4C) \\
		$\beta_R$ & 150 & reporter production rate (RFU/hr) & \cite{sprinzak_et_al_2010} Table S3 (Figure 4C) \\
		$\gamma$ & 0.1 & Notch, Delta decay rate (1/hr) & \cite{sprinzak_et_al_2010} Table S3 (Figure 4C) \\
		$\gamma_R$ & 0.05 & reporter decay rate (1/hr) & \cite{sprinzak_et_al_2010} Table S3 (Figure 4C) \\
		$k_c$ & 0.25 & inverse cis-interaction strength & \cite{sprinzak_et_al_2010} Table S3 (Figure 4C) \\
		$k_t$ & 5 & inverse trans-interaction strength & \cite{sprinzak_et_al_2010} Table S3 (Figure 4C) \\
		$n$ & 2 & Hill coefficient for Notch-Delta activation of reporter & - \\
		$m$ & 2 & Hill coefficient for reporter repression of Delta & - \\
		$k_{RS}$ & 300,000 & affinity of reporter induction & \cite{sprinzak_et_al_2011} Table S1 (Figure S4A) \\
		$k_{NS}$ & $5 \times 10^7$ & affinity of reporter induction & - \\ \hline
	\end{tabular}
	}
	\caption{Parameters used in the Notch-Delta model simulations, unless noted otherwise in the text (Figures \ref{fig:sprinzak_mi_kc} to \ref{fig:sprinzak_comps_sm}, \ref{fig:supp_sprinzak_comps}, and \ref{fig:supp_sprinzak_h2_uc_comps}).}
	\label{tab:supp_sprinzak_params}
\end{table}

\subsection{Mutual Inactivation (MI)} \label{sec:supp_mi_model}

For this system we can explicitly calculate the steady-state values $\underline{\bar{N}}^*, \underline{\bar{D}}^*, \underline{\bar{v}}_N^*, \underline{\bar{v}}_D^*$ for $u = 0$.  First we note that for our choice of $M$ the homogeneous solution satisfies $\underline{\bar{v}}_N^* = \underline{\bar{N}}^*$ and $\underline{\bar{v}}_D^* = \underline{\bar{D}}^*$.  After algebra, we find that $\underline{\bar{N}}^*$ is the positive root of a quadratic, $\underline{\bar{D}}^*$ is found in terms of $\underline{\bar{N}}^*$, and $\underline{\bar{R}}^*$ is expressed in terms of $\underline{\bar{N}}^*$ and $\underline{\bar{D}}^*$:
\begin{equation}
\begin{cases}
-\frac{\gamma}{K}\underline{\bar{N}}^{*2} + \left( \frac{\beta_N}{K} - \gamma^2 - \frac{\underline{\bar{\beta}}_D}{K} \right)\underline{\bar{N}}^* + \beta_N\gamma = 0 \\
\underline{\bar{D}}^* = \frac{\underline{\bar{\beta}}_D}{\gamma + K\underline{\bar{N}}^*} \\
\underline{\bar{R}}^* = \frac{\beta_R}{\gamma_R} \frac{\left( \underline{\bar{N}}^* \underline{\bar{D}}^* \right)^n}{k_{RS} + \left( \underline{\bar{N}}^* \underline{\bar{D}}^* \right)^n} \
\end{cases} \nonumber \
\end{equation}
where $K := \frac{k_ck_t}{k_c+k_t}$.

The filter coefficients $[S]_{kk}$ are given by $-C\left(A + \lambda_k(M)B_vG\right)^{-1}B_u$, $k = 0, 1, ..., N-1$.  To find them we can exploit the structure of $C$ and $B_u$.  For the sake of demonstration we will take the readout to be the reporter protein such that $C = [0~0~1]$, although the procedure applies equally well to arbitrary choices of $C$.

First we notate
\begin{equation}
A + \lambda_k(M)B_vG = \left[ \begin{matrix}
A_1 & 0 \\
[b_1,~b_2\lambda_k(M)] & -\gamma_R
\end{matrix} \right] \nonumber \
\end{equation}
and apply the matrix inversion lemma to obtain
\begin{equation}
\left(A + \lambda_k(M)B_vG\right)^{-1} = \left[ \begin{matrix}
A_1^{-1} & 0 \\
\frac{1}{\gamma_R}[b_1,~b_2\lambda_k(M)]A_1^{-1} & -\frac{1}{\gamma_R}
\end{matrix} \right]. \label{eq:inv_A_lambMBG} \
\end{equation}
Observe that
\begin{equation}
A_1^{-1} = \frac{1}{\det A_1}\left[ \begin{matrix}
-\gamma - \frac{\underline{\bar{N}}^*}{K} & \frac{\underline{\bar{N}}^*}{k_c} + \lambda_k(M_0)\frac{\underline{\bar{N}}^*}{k_t} \\
\frac{\underline{\bar{D}}^*}{k_c} + \lambda_k(M)\frac{\underline{\bar{D}}^*}{k_t} & -\gamma - \frac{\underline{\bar{D}}^*}{K}
\end{matrix} \right]. \nonumber \
\end{equation}
Premultiplying \eqref{eq:inv_A_lambMBG} by $C$ extracts the bottom row, while postmultiplying by $B_u$ extracts the middle entry of that row, which is given by
\begin{equation}
\frac{1}{\gamma_R \det A_1}\left[b_1\left( \frac{\underline{\bar{N}}^*}{k_c} + \lambda_k(M)\frac{\underline{\bar{N}}^*}{k_t} \right) + b_2\lambda_k(M)\left( -\gamma - \frac{\underline{\bar{D}}^*}{K} \right)\right]. \nonumber \
\end{equation}
Substituting $b_1$ and $b_2$, we simplify the expression to
\begin{equation}
[S]_{kk} = -\frac{\underline{\bar{N}}^*b_1}{\gamma_R k_c \det A_1}\left[ 1 - \lambda_k\left(M\right) \left( 1 + \frac{\gamma k_c}{\underline{\bar{D}}^*}\right) \right] \nonumber \
\end{equation}
where
\begin{subequations}\begin{align}
	\det A_1 &= \left( -\gamma - \frac{\underline{\bar{N}}^*}{K} \right) \left( -\gamma - \frac{\underline{\bar{D}}^*}{K} \right) - \left( \frac{\underline{\bar{N}}^*}{k_c} + \lambda_k\left(M\right)\frac{\underline{\bar{N}}^*}{k_t} \right) \left( \frac{\underline{\bar{D}}^*}{k_c} + \lambda_k\left(M\right)\frac{\underline{\bar{D}}^*}{k_t} \right) \nonumber \\
	&= -\left[ \frac{\underline{\bar{N}}^*\underline{\bar{D}}^*}{k_t^2}\lambda_k\left(M\right)^2 + 2\frac{\underline{\bar{N}}^*\underline{\bar{D}}^*}{k_tk_c}\lambda_k\left(M\right) + \left( \frac{\underline{\bar{N}}^*\underline{\bar{D}}^*}{k_c^2} - \gamma^2 - \frac{\gamma}{K} \left( \underline{\bar{N}}^* + \underline{\bar{D}}^* \right) - \frac{\underline{\bar{N}}^*\underline{\bar{D}}^*}{K^2}  \right) \right]. \nonumber \
\end{align}\end{subequations}
The dynamical system corresponding to these filter coefficients is analytically stable for our chosen $M$ with any biologically relevant parameter values (i.e., when the parameters in \ref{tab:supp_sprinzak_params} are positive, as they must be in a living system).  Since $A + \lambda_k\left(M_0\right)B_vG$ is a block triangular matrix, its eigenvalues are the eigenvalues of the diagonal blocks, i.e., $-\gamma_R$ along with the eigenvalues of $A_1$.  Since $-\gamma_R$ is always negative, checking for stability amounts to checking the sign of the eigenvalues of $A_1$.

Using the fact that $\underline{\bar{v}}_D^* = \underline{\bar{D}}^*$ and $\underline{\bar{v}}_N^* = \underline{\bar{N}}^*$ for our choice of $M$ yields
\begin{equation}
A_1 = \left[\begin{matrix}
-\gamma - \frac{\underline{\bar{D}}^*}{K} & -\frac{\underline{\bar{N}}^*}{k_c} - \lambda_k\left(M_0\right)\frac{\underline{\bar{N}}^*}{k_t} \\
-\frac{\underline{\bar{D}}^*}{k_c} - \lambda_k\left(M_0\right)\frac{\underline{\bar{D}}^*}{k_t} & -\gamma - \frac{\underline{\bar{N}}^*}{K}  \nonumber \
\end{matrix}\right]
\end{equation}
with $K$ as and $A_1$ as defined earlier.  The eigenvalues are given by the zeros of the characteristic polynomial, found by solving for $s$ in
\begin{equation}
\left(-\gamma - \frac{\underline{\bar{D}}^*}{K} - s\right)\left(-\gamma - \frac{\underline{\bar{N}}^*}{K} - s\right) - \left(-\frac{\underline{\bar{N}}^*}{k_c} - \lambda_k\left(M_0\right)\frac{\underline{\bar{N}}^*}{k_t}\right)\left(-\frac{\underline{\bar{D}}^*}{k_c} - \lambda_k\left(M_0\right)\frac{\underline{\bar{D}}^*}{k_t}\right) = 0. \ \nonumber
\end{equation}
The first term multiplies out to
\begin{equation}
s^2 + \left[2\gamma + \frac{\underline{\bar{N}}^* + \underline{\bar{D}}^*}{K}\right]s + \left[\gamma^2 + \frac{\underline{\bar{N}}^*\underline{\bar{D}}^*}{K^2} + \frac{\gamma}{K}\left(\underline{\bar{N}}^* + \underline{\bar{D}}^*\right)\right] \nonumber \
\end{equation}
and the second contributes the following terms, independent of $s$:
\begin{equation}
-\left(\frac{\underline{\bar{N}}^*}{k_c} + \lambda_k\left(M_0\right)\frac{\underline{\bar{N}}^*}{k_t}\right)\left(\frac{\underline{\bar{D}}^*}{k_c} + \lambda_k\left(M_0\right)\frac{\underline{\bar{D}}^*}{k_t}\right) = -\frac{\underline{\bar{N}}^*\underline{\bar{D}}^*}{k_c^2} - 2\lambda_k\left(M_0\right)\frac{\underline{\bar{N}}^*\underline{\bar{D}}^*}{k_ck_t} - \lambda_k\left(M_0\right)^2\frac{\underline{\bar{N}}^*\underline{\bar{D}}^*}{k_t^2}. \nonumber \
\end{equation}
To be biologically attainable the parameters and steady-state values must all be positive, such that the quadratic in $s$ has positive coefficients for the first- and second-order terms.  If the roots are complex then assuming nonzero decay and nontrivial solutions, the real part is given by
\begin{equation}
\frac{-\left[2\gamma + \frac{\underline{\bar{N}}^* + \underline{\bar{D}}^*}{K}\right]}{2} < 0, \nonumber \
\end{equation}
guaranteeing stability.

If the roots are real, then they will be negative if the zeroth-order term is positive.  However, if the zeroth-order term is negative, then one root will be positive and the system will not be stable.  Neglecting the expressions in $\gamma$, which by observation must be positive, the contributions to the zeroth-order term are
\begin{align}
& \frac{\underline{\bar{N}}^*\underline{\bar{D}}^*}{K^2} -\frac{\underline{\bar{N}}^*\underline{\bar{D}}^*}{k_c^2} - 2\lambda_k\left(M_0\right)\frac{\underline{\bar{N}}^*\underline{\bar{D}}^*}{k_ck_t} - \lambda_k\left(M_0\right)^2\frac{\underline{\bar{N}}^*\underline{\bar{D}}^*}{k_t^2} \nonumber \\
&= \frac{\underline{\bar{N}}^*\underline{\bar{D}}^*\left(k_c + k_t\right)^2 - k_t^2\underline{\bar{N}}^*\underline{\bar{D}}^* - 2k_ck_t\lambda_k\left(M_0\right)\underline{\bar{N}}^*\underline{\bar{D}}^* - \lambda_k\left(M_0\right)^2k_c^2\underline{\bar{N}}^*\underline{\bar{D}}^*}{k_c^2k_t^2} \nonumber \\
&= \frac{\underline{\bar{N}}^*\underline{\bar{D}}^*\left(k_c^2 + 2k_ck_t\right) - 2k_ck_t\lambda_k\left(M_0\right)\underline{\bar{N}}^*\underline{\bar{D}}^* - \lambda_k\left(M_0\right)^2k_c^2\underline{\bar{N}}^*\underline{\bar{D}}^*}{k_c^2k_t^2} \nonumber \\
&= \frac{\underline{\bar{N}}^*\underline{\bar{D}}^*}{k_c^2k_t^2}\left[k_c^2\left(1-\lambda_k\left(M_0\right)^2\right) + 2k_ck_t\left(1-\lambda_k\left(M_0\right)\right)\right]. \nonumber \
\end{align}
By our choice of $M_0$, $\lambda_k\left(M_0\right) \in \left[-1,1\right]$, therefore the bottom expression is minimized to $0$ by $\lambda_k\left(M_0\right) = 1$.  Since this expression contains all the possible negative contributions to the zeroth-order term, the overall zeroth-order term cannot be negative, and hence the roots of the characteristic polynomial must be negative, implying stability of the system with given $M_0$ for all biologically relevant parameter choices.

\subsection{Lateral Inhibition with Mutual Inactivation (LIMI)}
The system equations are the same as for the mutual inactivation model \eqref{eq:sprinzak_MI_1D}, except that now Delta production is repressed by reporter protein:
\begin{equation}
\begin{cases}
\dot{N_i}(t) = \beta_N - \gamma N_i(t) - \frac{N_i(t)v_{D_i}(t)}{k_t} - \frac{N_i(t)D_i(t)}{k_c} \\
\dot{D_i}(t) = \left(\underline{\bar{\beta}}_D + u_i\right)\frac{1}{1+R_i(t)^m} - \gamma D_i(t) - \frac{D_i(t)v_{N_i}(t)}{k_t} - \frac{N_i(t)D_i(t)}{k_c} \\
\dot{R_i}(t) = \beta_R\frac{\left(N_i(t)v_{D_i}(t)\right)^n}{k_{RS} + \left(N_i(t)v_{D_i}(t)\right)^n} - \gamma_R R_i(t) \\
y_i (t)= Cx_i(t) \\
w_i (t)= \left[ \begin{matrix}
N_i(t) \\ D_i(t)
\end{matrix} \right] \\
v(t) = \left( M \otimes I_2 \right) w(t) \
\end{cases}. \label{eq:sprinzak_LIMI_1D} \
\end{equation}

When linearized at steady state, the relevant matrices are
\begin{subequations}\begin{align}
	A &= \left[ \begin{matrix}
	-\gamma -\frac{\underline{\bar{v}}_D^*}{k_t} - \frac{\underline{\bar{D}}^*}{k_c} & -\frac{\underline{\bar{N}}^*}{k_c} & 0 \\
	-\frac{\underline{\bar{D}}^*}{k_c} & - \gamma - \frac{\underline{\bar{v}}_N^*}{k_t} - \frac{\underline{\bar{N}}^*}{k_c} & -a \\
	b_1 & 0 & -\gamma_R
	\end{matrix} \right], \nonumber \\
	B_v &= \left[ \begin{matrix}
	0 & -\frac{\underline{\bar{N}}^*}{k_t} \\
	-\frac{\underline{\bar{D}}^*}{k_t} & 0 \\
	0 & b_2
	\end{matrix} \right], ~~
	B_u = \left[ \begin{matrix}
	0 \\
	1 \\
	0
	\end{matrix} \right], \nonumber \\
	C &= \left[ \begin{matrix} 0 & 0 & 1 \end{matrix} \right], ~~ G = \left[ \begin{matrix} 1 & 0 & 0 \\ 0 & 1 & 0 \end{matrix} \right] \nonumber \
	\end{align}\end{subequations}
where $b_1, b_2$ are defined as before and
\begin{equation}
a := m\frac{\underline{\bar{R}}^{*m-1}}{\left(1+\underline{\bar{R}}^{*m}\right)^2}. \nonumber \
\end{equation}

\subsection{Simplest Lateral Inhibition by Mutual Inactivation (SLIMI)}
The system equations are
\begin{equation}
\begin{cases}
\dot{N_i}(t) = \alpha_N + \beta_N\frac{\left(N_i(t)v_{D_i}(t)\right)^n}{k_{NS}+\left(N_i(t)v_{D_i}(t)\right)^n} - \gamma N_i(t) - \frac{N_i(t)v_{D_i}(t)}{k_t} - \frac{N_i(t)D_i(t)}{k_c} \\
\dot{D_i}(t) = \underline{\bar{\beta}}_D + u_i - \gamma D_i - \frac{D_i(t)v_{N_i}(t)}{k_t} - \frac{N_i(t)D_i(t)}{k_c} \\
y_i(t) = Cx_i(t) \\
w_i(t) = \left[ \begin{matrix}
N_i(t) \\ D_i(t)
\end{matrix} \right] \
\end{cases}. \label{eq:sprinzak_SLIMI_1D} \
\end{equation}

When linearized at steady state, the relevant matrices are
\begin{subequations}\begin{align}
	A &= \left[ \begin{matrix}
	b_1 - \gamma -\frac{\underline{\bar{v}}_D^*}{k_t} - \frac{\underline{\bar{D}}^*}{k_c} & -\frac{\underline{\bar{N}}^*}{k_c} \\
	-\frac{\underline{\bar{D}}^*}{k_c} & - \gamma - \frac{\underline{\bar{v}}_N^*}{k_t} - \frac{\underline{\bar{N}}^*}{k_c}
	\end{matrix} \right], \nonumber \\
	B_v &= \left[ \begin{matrix}
	0 & b_2 - \frac{\underline{\bar{N}}^*}{k_t} \\
	-\frac{\underline{\bar{D}}^*}{k_t} & 0
	\end{matrix} \right], ~~
	B_u = \left[ \begin{matrix}
	0 \\
	1
	\end{matrix} \right], \nonumber \\
	G &= \left[ \begin{matrix} 1 & 0 \\ 0 & 1 \end{matrix} \right] \nonumber \
	\end{align}\end{subequations}
where now
\begin{equation}
b_1 := \beta_Nnk_{NS}\frac{\underline{\bar{v}}_D^{*n}\underline{\bar{N}}^{*n-1}}{\left(k_{NS}+\left(\underline{\bar{N}}^*\underline{\bar{v}}_D^*\right)^n\right)^2}, ~ b_2 := \frac{\underline{\bar{N}}^*}{\underline{\bar{v}}_D^*}b_1. \nonumber \
\end{equation}

\begin{figure}[!hbpt]
	\centering
	\includegraphics[width=\columnwidth]{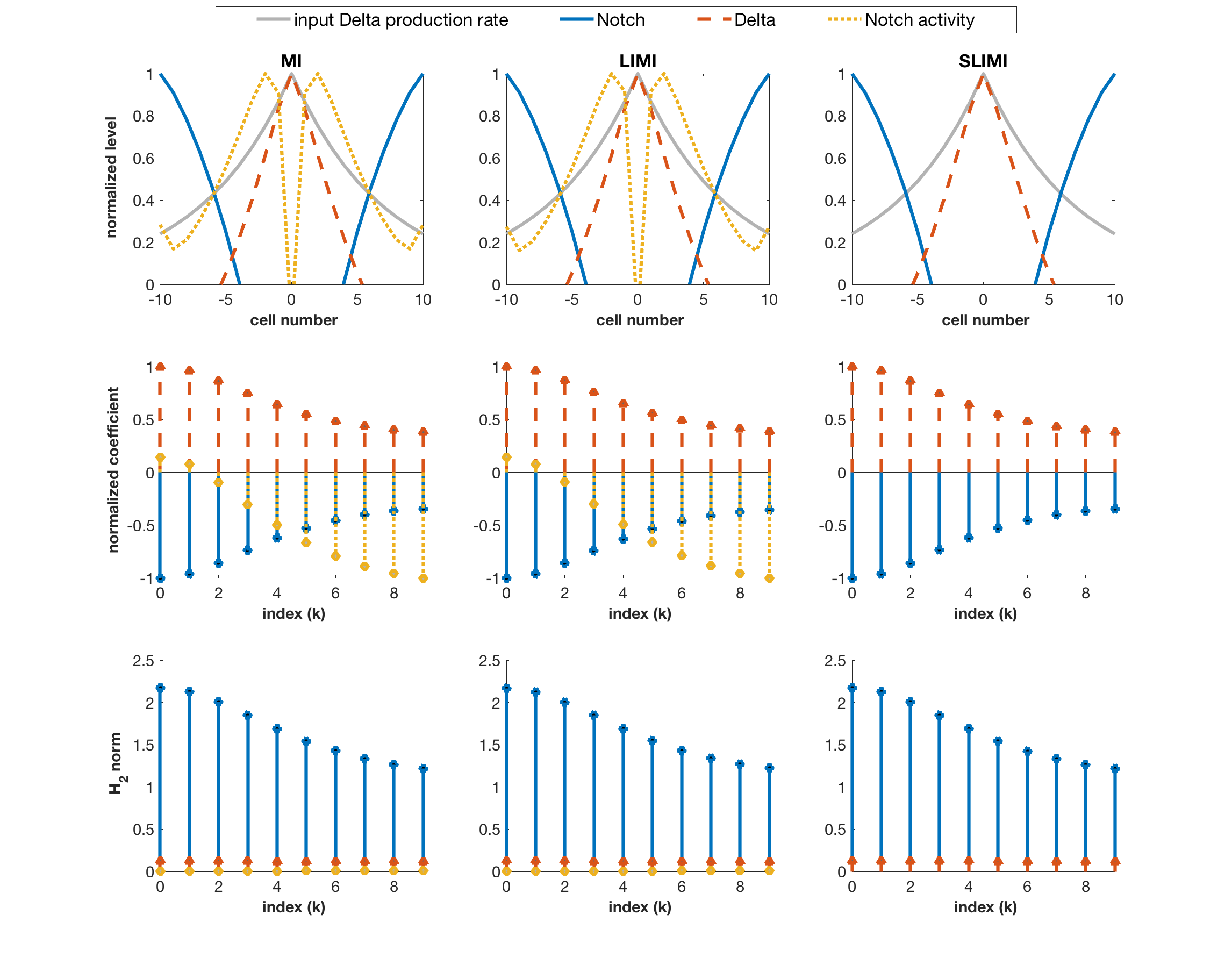}
	\caption{\textbf{Three different models for the Notch-Delta interaction produce qualitatively similar filter characteristics.}  Top row, a two-sided exponential input gradient of Delta production rate (solid light gray) results in two sharp bands of Notch activity (dotted yellow) that spatially segregates steady-state levels of Notch (solid blue) and Delta (dashed  orange).  Curves are normalized to their respective maxima.  Note that the SLIMI model lacks a reporter protein and so does not have an output measure for Notch activity.  Middle row, the magnitude of the filter coefficients for each possible output.  Because the spatial modes correspond to the DFT basis, the coefficients exhibit mirror-image symmetry about $k = \frac{N}{2}$; we plot only the first half of the coefficients to better visualize the filter's characteristic highpass shape for output Notch activity, and lowpass shape for Delta and Notch (with a $\frac{pi}{2}$ phase shift in Notch expression).  Each set of coefficients has been individually normalized to the maximum in each set.  Bottom row, the $\mathcal{H}_2$ norm is qualitatively similar to the filter characteristic for the corresponding output.  The coefficients here are not normalized in order to better visualize the large gain in Notch expression relative to Delta or activity levels.  Parameters for all models are given in Table \ref{tab:supp_sprinzak_params}.}
	\label{fig:supp_sprinzak_comps}
\end{figure}

\begin{figure}[!hbpt]
	\centering
	\includegraphics[width=\columnwidth]{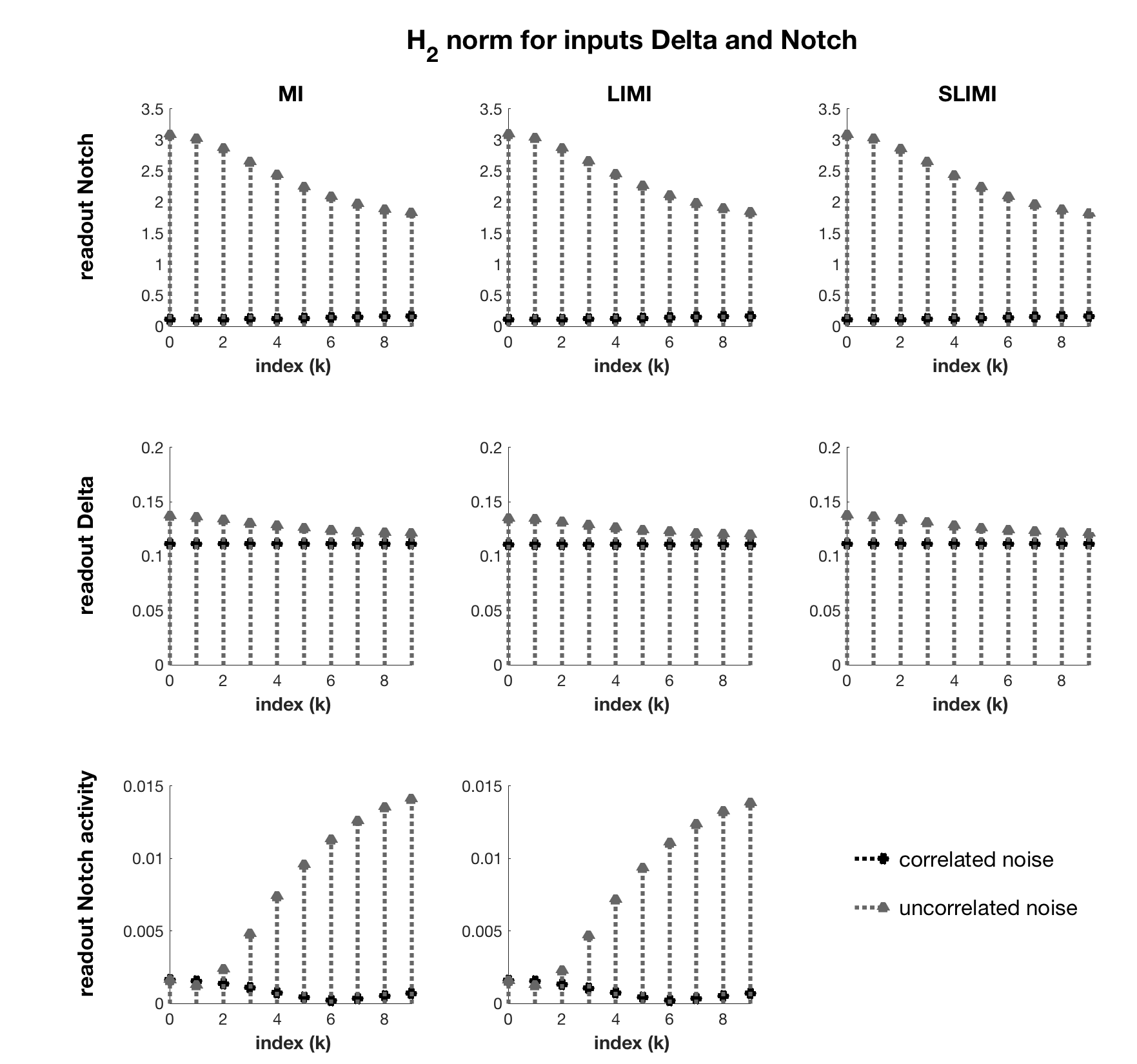}
	\caption{\textbf{System response to correlated vs. uncorrelated noise is similar across models.}  The most notable behavioral difference between models is that the LIMI model rejects uncorrelated noise slightly more strongly for readout Delta.  Parameters for all models are given in Table \ref{tab:supp_sprinzak_params}.}
	\label{fig:supp_sprinzak_h2_uc_comps}
\end{figure}

\clearpage
\section{DIGIT FORMATION}

\begin{table}[!bp]
	\centering
	\begin{tabular}{| c | c | c |} \hline
		\bf Parameter & \bf Value & \bf Description \\ \hline
		$\alpha_{sox9}$ & 0 & constitutive Sox9 production rate \\
		$\alpha_{bmp}$ & 16.9 & constitutive Bmp production rate \\
		$\alpha_{wnt}$ & 13.7 & constitutive Wnt production rate \\
		$k_2$ & 1 & Bmp promotion of \textit{sox9} expression \\
		$k_3$ & 1 & Wnt repression of \textit{sox9} expression \\
		$k_4$ & 1.59 & Sox9 repression of \textit{bmp} expression \\
		$k_5$ & 0.1 & Bmp decay rate \\
		$k_7$ & 1.27 & Sox9 repression of \textit{wnt} expression \\
		$k_9$ & 0.1 & Wnt decay rate \\
		$d_b$ & 2.5 & diffusion coefficient for Bmp \\
		$d_w$ & 1 & diffusion coefficient for Wnt \\
		$l$ & 1.7 & distance between cells \\
		\hline
	\end{tabular}
	\caption{Parameters used in the simulations of digit formation, unless noted otherwise in the text (Figures \ref{fig:raspopovic_color_plates}, \ref{fig:raspopovic_eigs}, \ref{fig:raspopovic_h2}).  Values are from Table ST4 and text of \cite{raspopovic_et_al_2014}.}
	\label{tab:supp_raspopovic_params}
\end{table}

\begin{table}[!hbpt]
	\centering
	\begin{tabular}{| c | c | c |} \hline
		\bf Parameter & \bf Value & \bf Description \\ \hline
		$\alpha_{sox9}$ & 0 & constitutive Sox9 production rate \\
		$\alpha_{bmp}$ & 0.1 & constitutive Bmp production rate \\
		$\alpha_{wnt}$ & 1.2 & constitutive Wnt production rate \\
		$\mu_F$ & 0.1 & Fgf decay rate \\
		$k_2$ & 1 & Bmp promotion of \textit{sox9} expression \\
		$k_3$ & 3 & Wnt repression of \textit{sox9} expression \\
		$k_4$ & 6 & Sox9 repression of \textit{bmp} expression \\
		$k_5$ & 0.1 & Bmp decay rate \\
		$k_7$ & 2.4 & Sox9 repression of \textit{wnt} expression \\
		$k_9$ & 0.1 & Wnt decay rate \\
		$k_f$ & $\frac{2}{3}$ & strength of Fgf influence on $k_4$, $k_7$ \\
		$d_b$ & 160 & diffusion coefficient for Bmp \\
		$d_w$ & 25 & diffusion coefficient for Wnt \\
		$d_F$ & 600 & diffusion coefficient for Fgf \\
		$l$ & 4 & distance between cells \\
		\hline
	\end{tabular}
	\caption{Parameters used in the simulations of digit formation with a morphogen gradient, unless noted otherwise in the text (Figures \ref{fig:onimaru_fully_loaded}, \ref{fig:supp_onimaru_eigs_wnt_bmp} to \ref{fig:supp_onimaru_h2_wnt_corr}).  Values are from Methods in \cite{onimaru_et_al_2016} (Figures 4 and 5).}
	\label{tab:supp_onimaru_params}
\end{table}
	
\begin{figure}[!bp]
	\centering
	\includegraphics[width=\columnwidth]{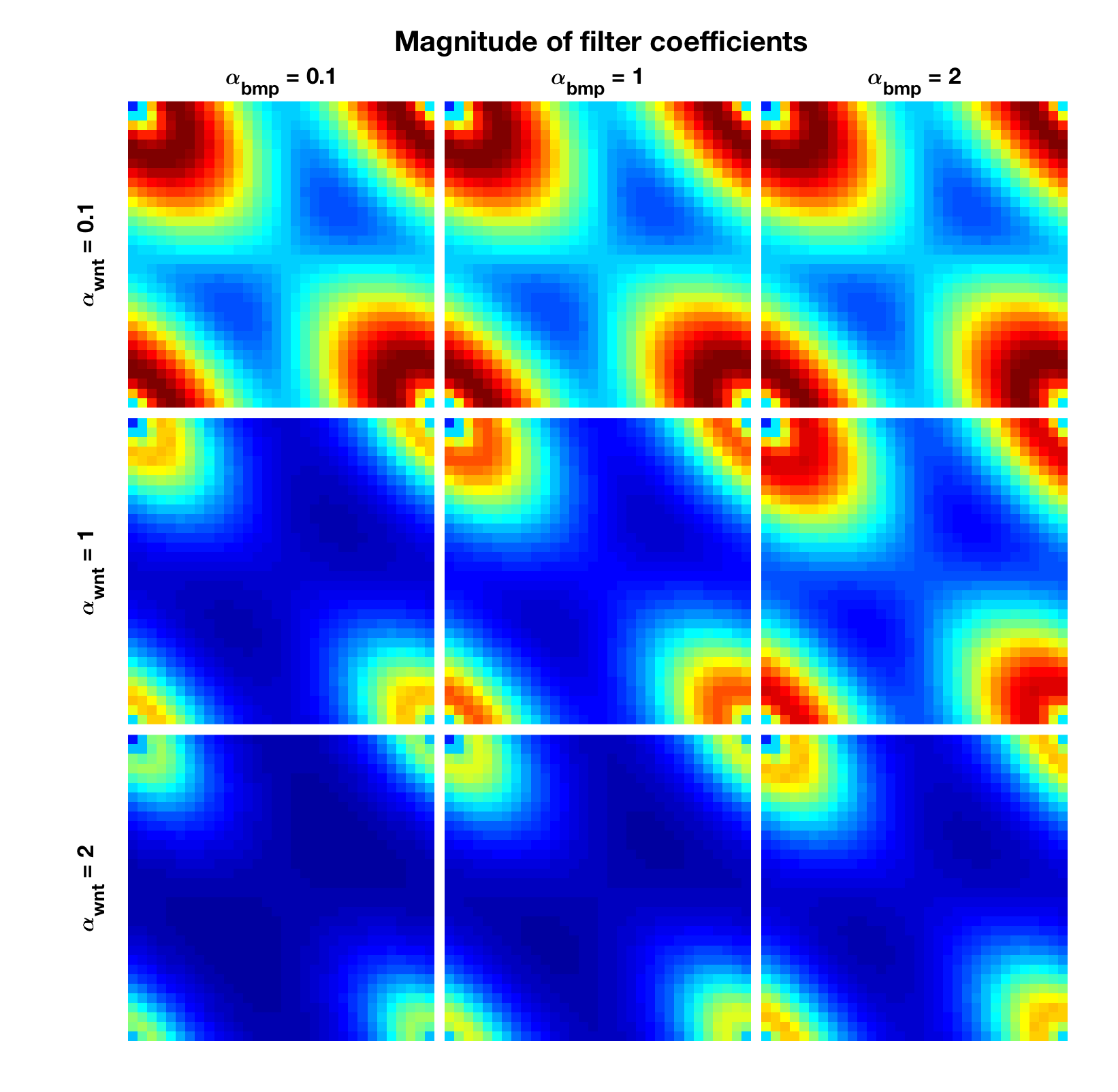}
	\caption{Magnitude of filter coefficients for input background production rate and readout [Sox9] as production rates for Wnt and Bmp are varied.  Increasing the ratio of $\alpha_{wnt}$ to $\alpha_{bmp}$ shrinks the size and magnitude of the passband.  Images are normalized to the same scale (min. 0, max. 2.19).  Other parameters are from Table \ref{tab:supp_onimaru_params}.}
	\label{fig:supp_onimaru_eigs_wnt_bmp}
\end{figure}

\begin{figure}[!bp]
	\centering
	\includegraphics[width=\columnwidth]{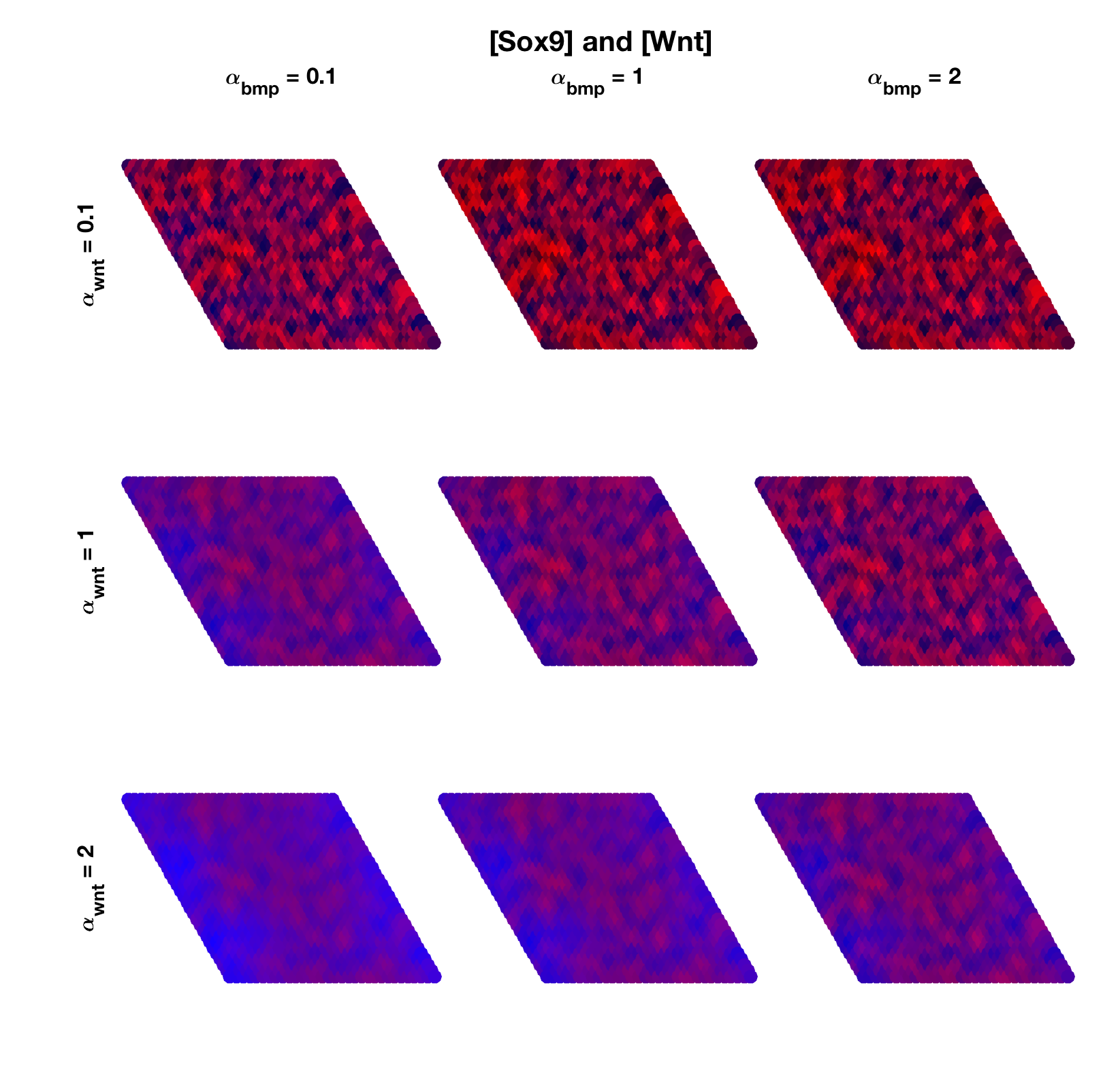}
	\caption{Simulated outputs Sox9 (red) and Wnt (blue) for changing production rates of Wnt and Bmp.  Smaller $\alpha_{bmp}$ exaggerates the effect of the Fgf gradient on [Wnt].  Readouts are normalized independently to the same scale across all images.  Other parameters are from Table \ref{tab:supp_onimaru_params}.}
	\label{fig:supp_onimaru_color_plates_wnt_bmp}
\end{figure}

\begin{figure}[!bp]
	\centering
	\includegraphics[width=\columnwidth]{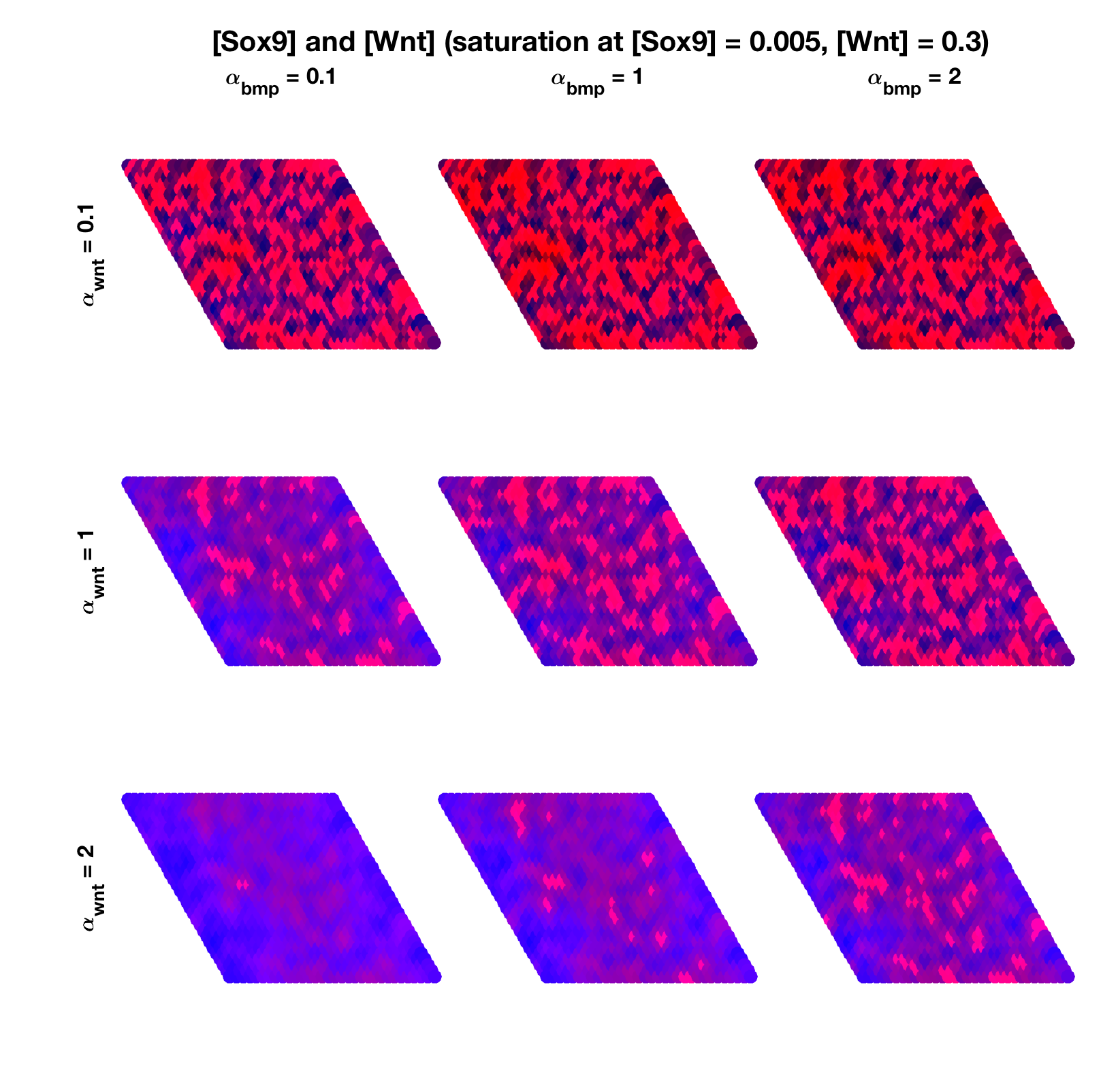}
	\caption{Simulated readouts with saturation for Sox9 (red) and Wnt (blue) for changing production rates of Wnt and Bmp.  As observed in \cite{onimaru_et_al_2016}, higher ratios of $\alpha_{wnt}$ to $\alpha_{bmp}$ produce more spotlike patterns.  Readouts are normalized independently to the same scale across all images.  Other parameters are from Table \ref{tab:supp_onimaru_params}.}
	\label{fig:supp_onimaru_sat_color_plates_wnt_bmp}
\end{figure}

\begin{figure}[!bp]
	\centering
	\includegraphics[width=\columnwidth]{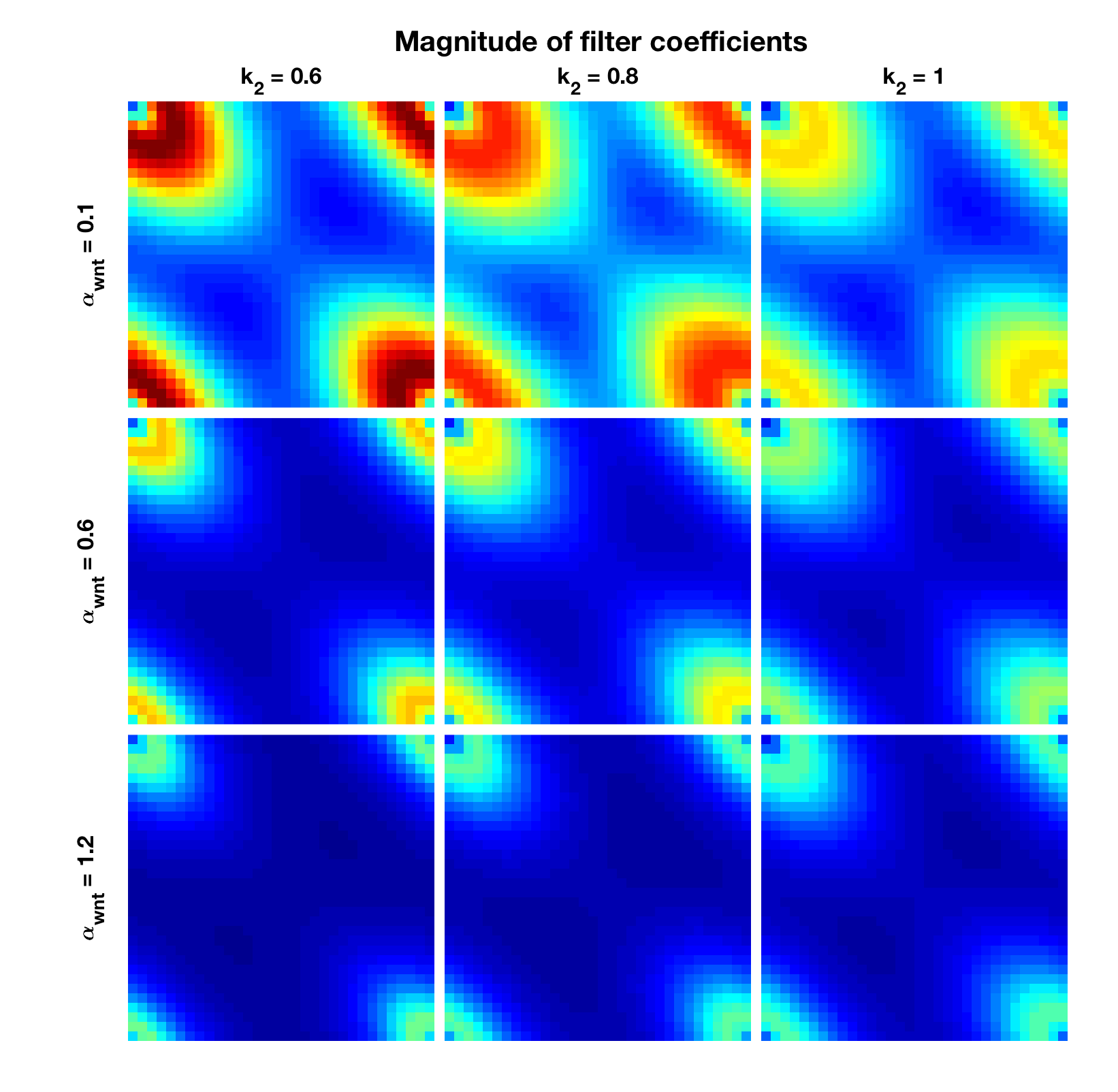}
	\caption{Full plots for the magnitude of filter coefficients for varying Wnt production and Bmp promotion of \textit{sox9} expression (see insets in Figure \ref{fig:onimaru_fully_loaded}).  The coefficients exhibit hexagonal symmetry when tiled on a hexagonal lattice.  Readouts are normalized independently to the same scale across all images.  Other parameters are from Table \ref{tab:supp_onimaru_params}.}
	\label{fig:supp_onimaru_eigs_wnt_k2}
\end{figure}

\begin{figure}[!bp]
	\centering
	\includegraphics[width=\columnwidth]{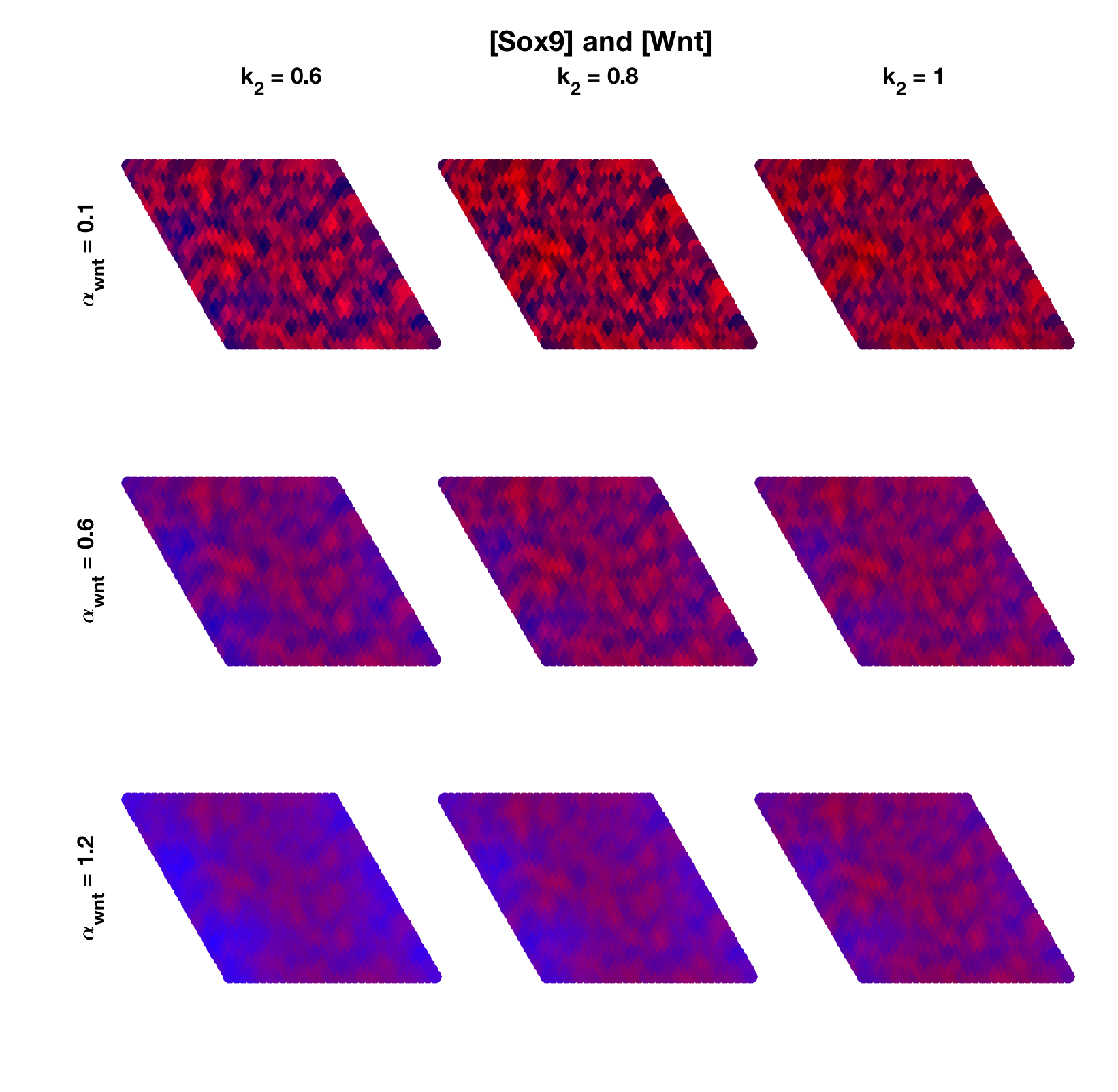}
	\caption{Full simulated readouts for Sox9 (red) and Wnt (blue) for changing Wnt production rate and Bmp promotion (see Figure \ref{fig:onimaru_fully_loaded}).  Readouts are normalized independently to the same scale across all images.  Other parameters are from Table \ref{tab:supp_onimaru_params}.}
	\label{fig:supp_onimaru_color_plates_wnt_k2}
\end{figure}

\begin{figure}[!bp]
	\centering
	\includegraphics[width=\columnwidth]{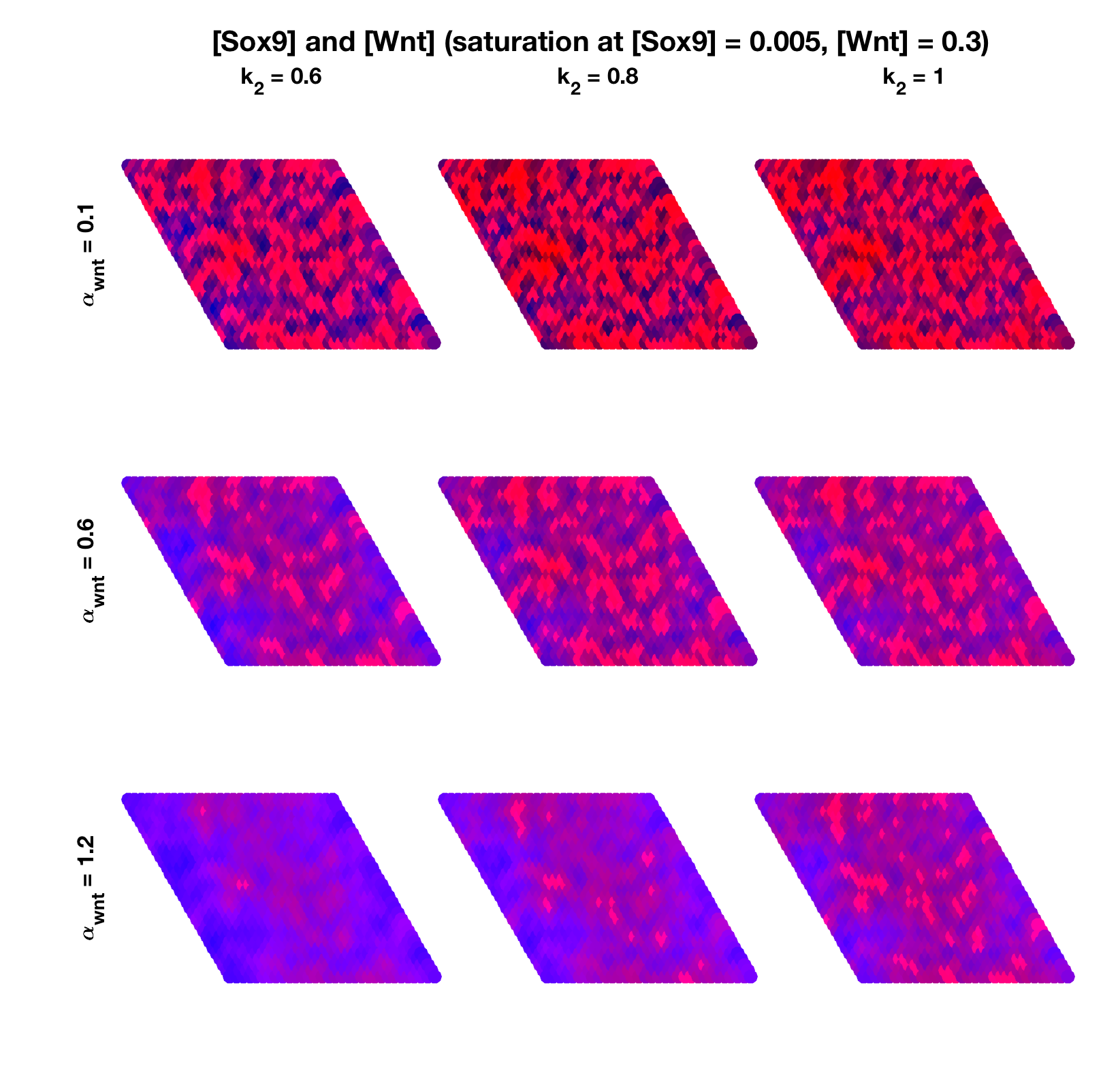}
	\caption{Full simulated readouts with saturation for Sox9 (red) and Wnt (blue) for changing Wnt production rate and Bmp promotion (see Figure \ref{fig:onimaru_fully_loaded}).  Readouts are normalized independently to the same scale across all images.  Other parameters are from Table \ref{tab:supp_onimaru_params}.}
	\label{fig:supp_onimaru_sat_color_plates_wnt_k2}
\end{figure}

\begin{figure}[!bp]
	\centering
	\includegraphics[width=\columnwidth]{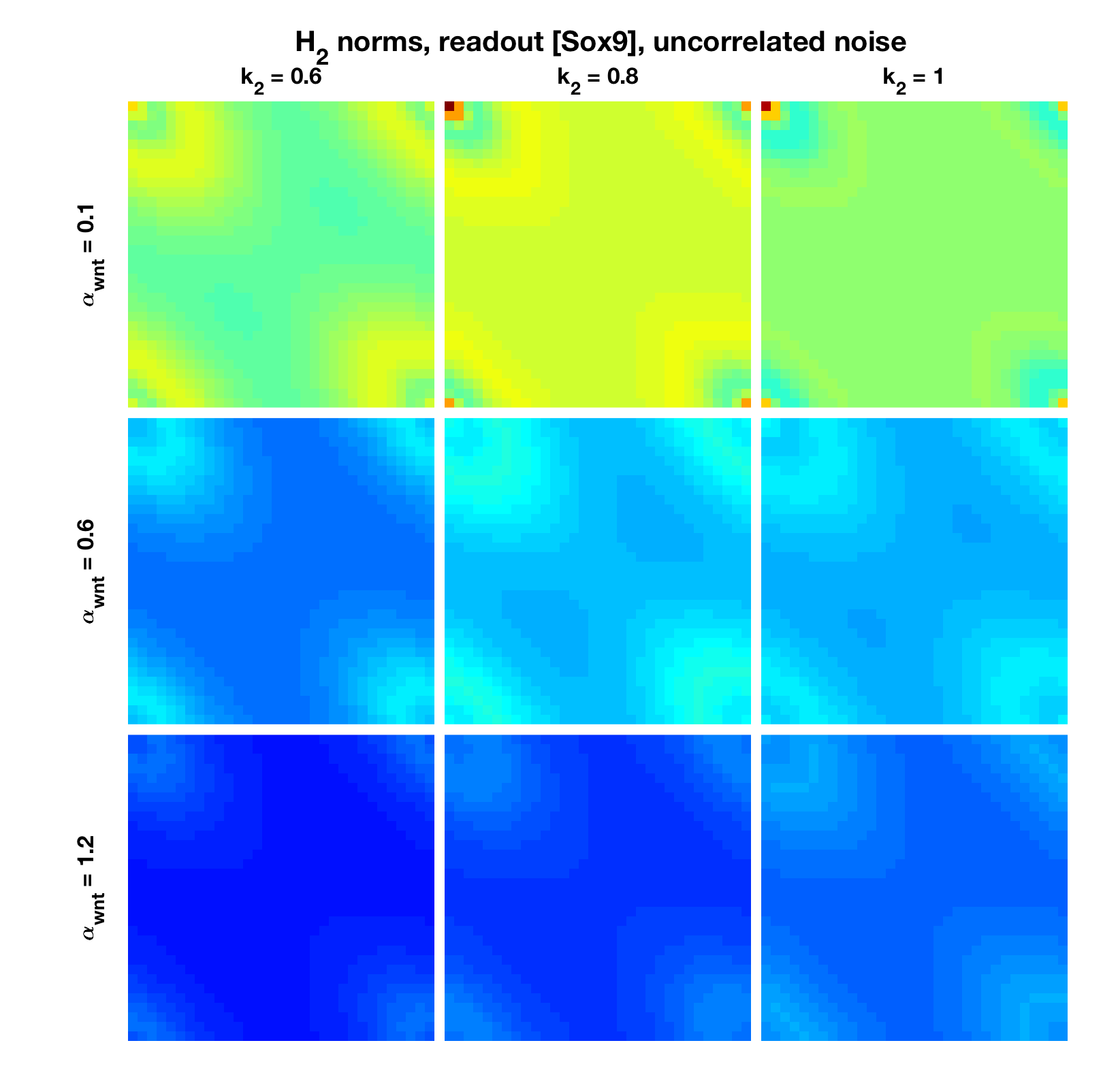}
	\caption{$\mathcal{H}_2$ norms with readout [Sox9] for 100\% uncorrelated white noise inputs to \textit{sox9}, \textit{bmp}, and \textit{wnt}.  Images are normalized to the same scale as Figure \ref{fig:supp_onimaru_h2_sox9_corr} (min. 0, max. 3.43).  Other parameters are from Table \ref{tab:supp_onimaru_params}.}
	\label{fig:supp_onimaru_h2_sox9_uncorr}
\end{figure}

\begin{figure}[!bp]
	\centering
	\includegraphics[width=\columnwidth]{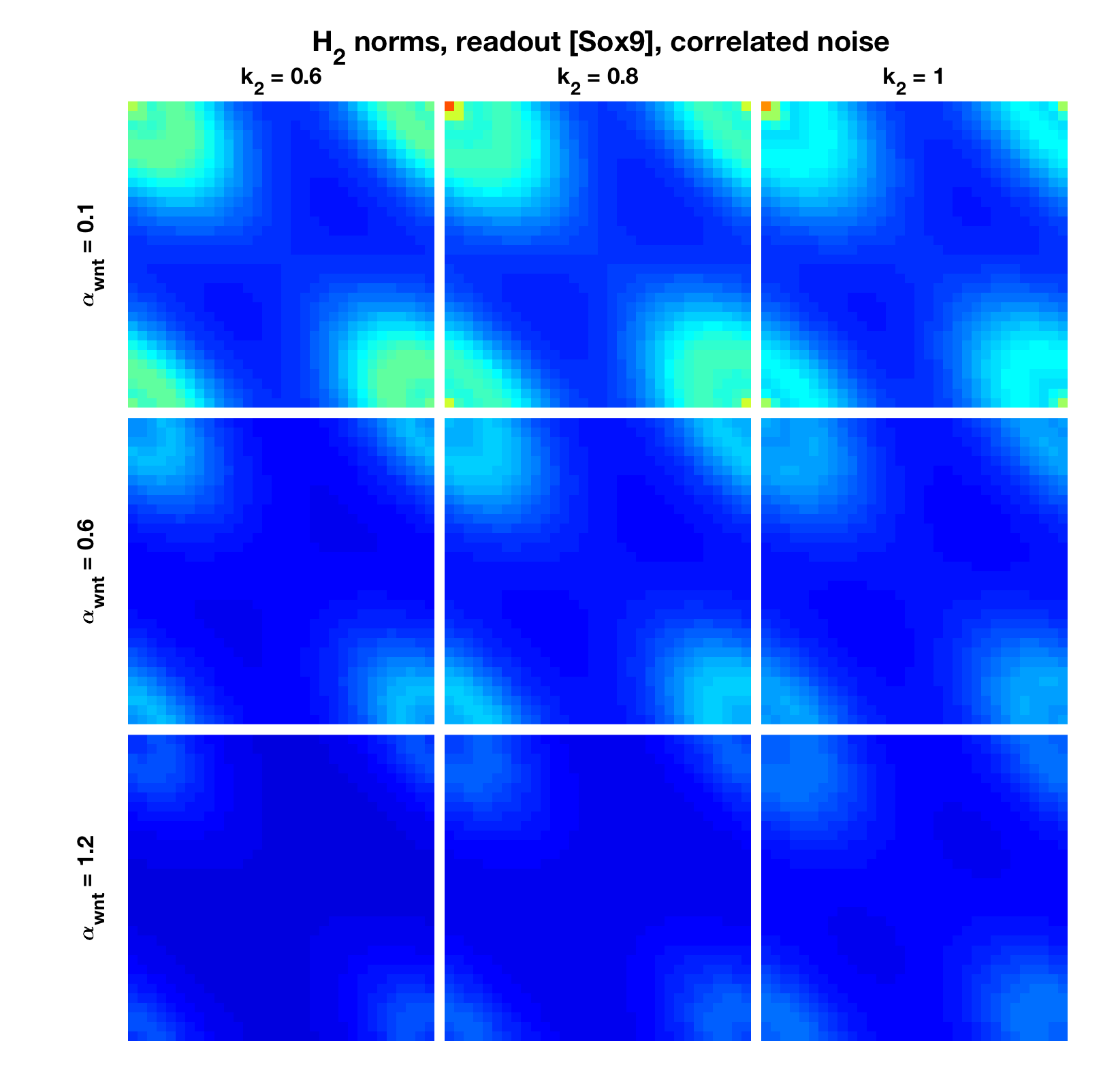}
	\caption{$\mathcal{H}_2$ norms with readout [Sox9] for 100\% correlated white noise inputs to \textit{sox9}, \textit{bmp}, and \textit{wnt}.  Images are normalized to the same scale as Figure \ref{fig:supp_onimaru_h2_sox9_uncorr} (min. 0, max. 3.43).  Other parameters are from Table \ref{tab:supp_onimaru_params}.}
	\label{fig:supp_onimaru_h2_sox9_corr}
\end{figure}

\begin{figure}[!bp]
	\centering
	\includegraphics[width=\columnwidth]{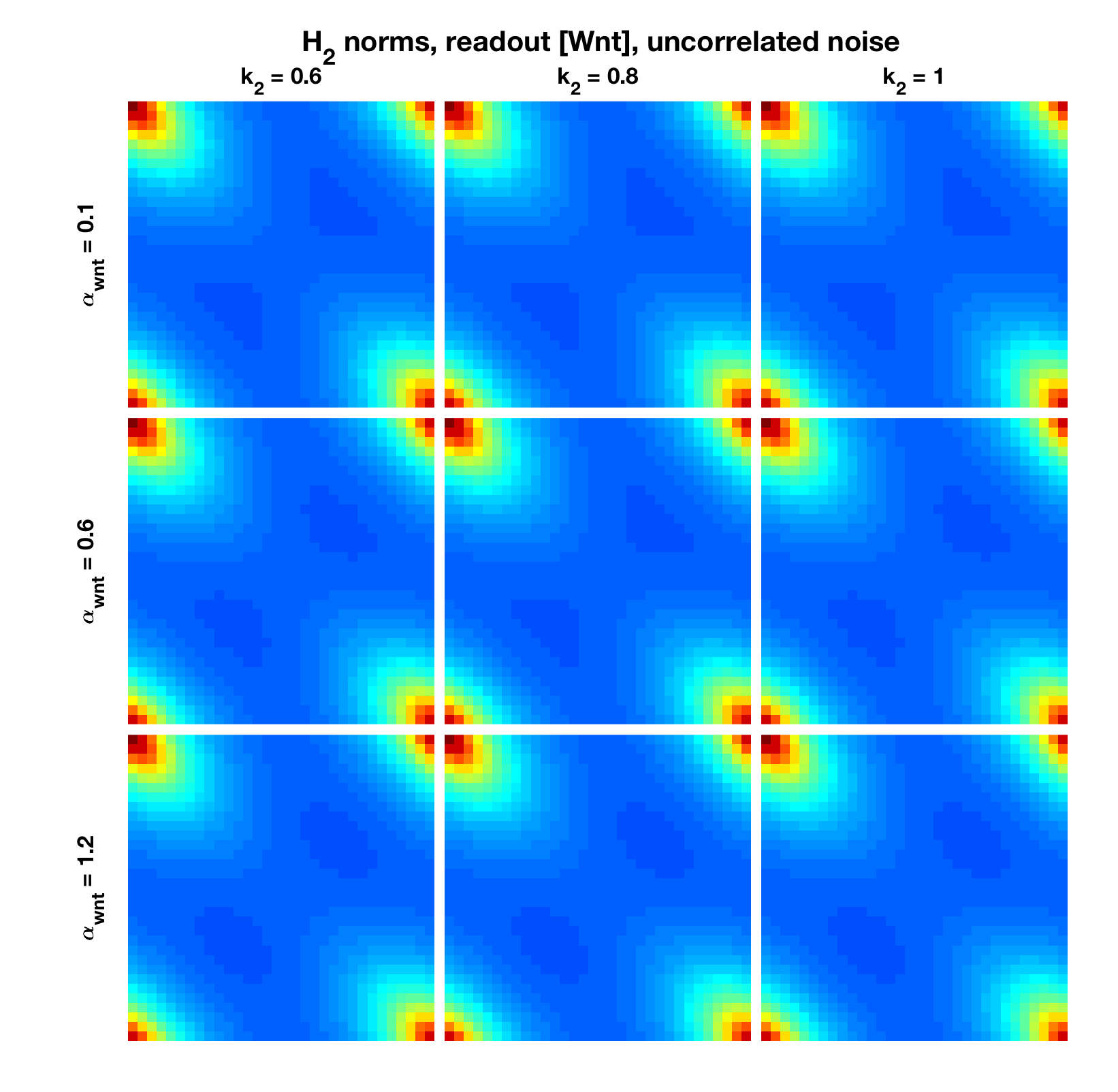}
	\caption{$\mathcal{H}_2$ norms with readout [Wnt] for 100\% uncorrelated white noise inputs to \textit{sox9}, \textit{bmp}, and \textit{wnt}.  Images are normalized to the same scale as Figure \ref{fig:supp_onimaru_h2_wnt_corr} (min. 0, max. 2.28).  Other parameters are from Table \ref{tab:supp_onimaru_params}.}
	\label{fig:supp_onimaru_h2_wnt_uncorr}
\end{figure}

\begin{figure}[!bp]
	\centering
	\includegraphics[width=\columnwidth]{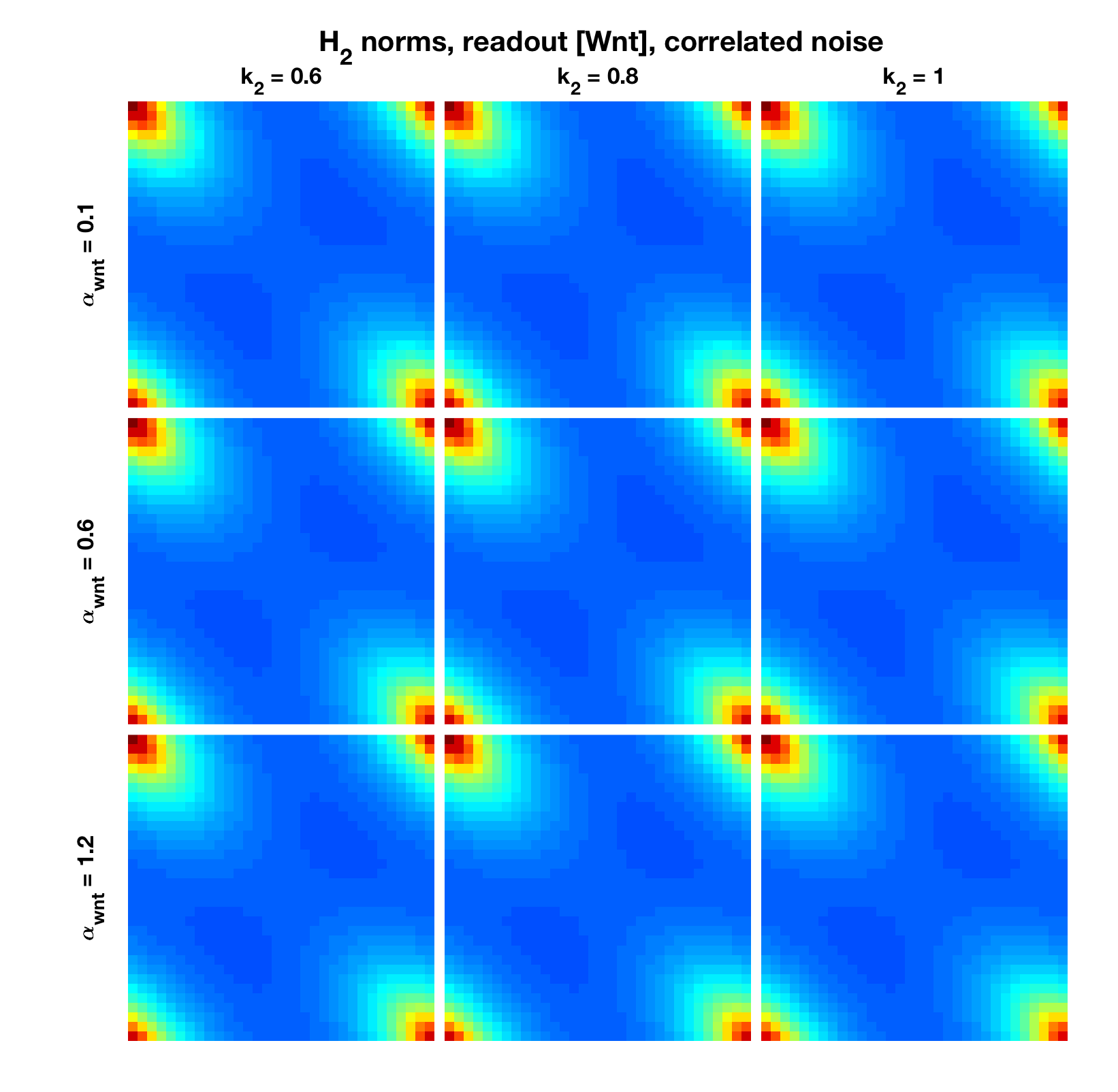}
	\caption{$\mathcal{H}_2$ norms with readout [Wnt] for 100\% correlated white noise inputs to \textit{sox9}, \textit{bmp}, and \textit{wnt}.  Images are normalized to the same scale as Figure \ref{fig:supp_onimaru_h2_wnt_uncorr} (min. 0, max. 2.28).  Other parameters are from Table \ref{tab:supp_onimaru_params}.}
	\label{fig:supp_onimaru_h2_wnt_corr}
\end{figure}

\clearpage
\printbibliography

\end{document}